%% file: Beamspace_MD_HW.tex
\documentclass[journal]{IEEEtran}
\usepackage{cite,amsthm,mathdots}
\usepackage[utf8]{inputenc}
\input{packages}

\newtheorem{lemma}{Lemma}
\newtheorem{proposition}{Proposition}

\newcommand\numberthis{\addtocounter{equation}{1}\tag{\theequation}}

\begin{document}
\usetikzlibrary{shapes.multipart,intersections}

\title{Beamspace Multidimensional ESPRIT Approaches for Simultaneous Localization and Communications}

\author{Fan~Jiang,~\IEEEmembership{Member,~IEEE,}
        Fuxi~Wen,~\IEEEmembership{Senior~Member,~IEEE,}
        Yu~Ge,~\IEEEmembership{Student~Member,~IEEE,}
        Meifang~Zhu,
        ~Henk~Wymeersch,~\IEEEmembership{Senior~Member,~IEEE,}
        ~Fredrik~Tufvesson,~\IEEEmembership{Fellow,~IEEE}% <-this % stops a space
\thanks{Fan Jiang, Yu Ge, and Henk Wymeersch are with the Department of Electrical Engineering of Chalmers University of Technology, Gothenborg, Sweden. Email: \{fan.jiang,~yuge,~henkw\}@chalmers.se.}% <-this % stops a space
\thanks{Fuxi Wen is with the School of Vehicle and Mobility, Tsinghua University, Beijing, China, Email:wenfuxi@tsinghua.edu.cn.}% <-this % stops a space
\thanks{Meifang Zhu and Fredrik Tufvesson are with the Department of Electrical and Information Technology, Lund University, Lund, Sweden, Email: \{meifang.zhu,~fredrik.tufvesson\}@eit.lth.se.}
\thanks{This work has been partly funded by the Vinnova 5GPOS project under grant 2019-03085, and by the Wallenberg AI, Autonomous Systems and Software Program (WASP) funded by the Knut and Alice Wallenberg Foundation.}
}

% \markboth{Journal of \LaTeX\ Class Files,~Vol.~14, No.~8, August~2015}%
% {Shell \MakeLowercase{\textit{et al.}}: Bare Demo of IEEEtran.cls for IEEE Journals}

\maketitle

\begin{abstract}
Modern wireless communication systems operating at high carrier frequencies are characterized by a high dimensionality of the underlying parameter space (including channel gains, angles, delays, and possibly Doppler shifts). Estimating these parameters is valuable for communication purposes, but also for localization and sensing, making channel estimation a critical component in any joint communication and localization or sensing application. The high dimensionality make it difficult to use search-based methods such as maximum likelihood. Search-free methods such as ESPRIT provide an attractive alternative, but require a complex decomposition step in both the tensor and matrix version of ESPRIT. To mitigate this, we propose, develop, and analyze a reduced complexity beamspace ESPRIT method. Complexity is reduced both by beampace processing as well as low-complex implementation of the singular value decomposition. A novel perturbation analysis provides important insights for both channel estimation and localization performance. The proposed method is compared to the tensor ESPRIT method, in terms of channel estimation, communication, localization, and sensing performance, further validating the perturbation analysis. 
\end{abstract}

\IEEEpeerreviewmaketitle

\section{Introduction}
%Communication systems operating at high carrier frequencies, such as in 5G mmWave systems at 24 GHz  \cite{osseiran2014scenarios} and foreseen in beyond 5G at 60 GHz - 200 GHz \cite{tataria20206g}, are attractive from both communication and localization points of view. For communication, the large bandwidths available at high frequencies support extremely high data rates. A sufficient link budget can be provided by use of directional transmissions, supported by electrically large antenna arrays at both transmitter and receiver. On the other hand, from a localization point of view, the sparsity of the propagation channel, combined with high delay and angle resolution, enables more accurate user localization, as well as localization with fewer base stations \cite{wymeersch20175g}. Moreover, the communication link can serve as a bi-static radar, which provides geometric information of the propagation environment, such as the locations of large reflecting or scattering surfaces or small objects \cite{chaccour2021seven}. 

Communication and localization are conventionally seen as competing for the same precious time-frequency resources, where reference signals used for localization take away resources used for communication. In that sense, one must trade off communication rate with localization accuracy \cite{destino2017trade,koirala2019localization}. At the same time, it has been appreciated that location information is useful for communication \cite{di2014location}. For instance, since communication at high carrier frequencies is highly directional, location information can be an input in the design of precoders and combiners, illuminating only the dominant geometric propagation paths \cite{jayaprakasam2017robust}. The explicit connection between communication and localization is via the channel estimation process. Typically, 5G mmWave channel estimation is based on a sparse and parametric representation of the channel, in which the parameters of the channel include the angles, distances, and gains \cite{alkhateeb2014channel}. These geometric parameters are also used for localization and mapping \cite{GeWenKimZhuJiaKimSveWym20}. In turn, this means that localization and communication are not competing, but can rather integrated as a composed function. This naturally leads to simultaneous localization and communication (SLAC), i.e., a system where localization and communications go hand in hand. 
A challenge in channel estimation for 5G and beyond 5G is the large dimensionality \cite{richter2005estimation}: a system with a moderate 1000 subcarriers and 64 antennas (e.g., an 8 by 8 planar array) at transmitter and receiver leads to a channel with 4 million parameters. Sparse estimation from such a large object is challenging and requires dedicated channel estimation methods. 

Channel estimation methods can be roughly categorized as \emph{search-based} and \emph{search-free} methods. In the search-based methods we include techniques such as maximum likelihood \cite{richter2005estimation}, MUSIC \cite{huang2012multidimensional}, orthogonal matching pursuit \cite{lee2016channel}, as well as other grid-based and sparsity-based methods \cite{tsai2018millimeter}. These methods fundamentally require non-convex optimization over (at the minimum) the underlying parameter space, which can be prohibitively complex, especially since  for both localization and communication,  auto-pairing of the parameters across dimensions is required \cite{SahUseCom17}.
%which can range from 1 dimensional (1D) (e.g., single-antenna OFDM requires searching in the delay domain), to 5D (when in addition to delay, angle-of-arrival and angle-of-departure in azimuth and elevation are unknown). While 1D or 2D optimization are reasonable, 5D optimization is computationally prohibitive. Complexity can be reduced by separate optimization in each domain, but this loses the association between dimensions (i.e., which angle corresponds to which delay). For both localization and communication, so-called auto-pairing of the parameters across dimensions is desirable \cite{SahUseCom17}. 
In contrast, search-free methods can provide a direct estimate of the channel parameters (i.e., tuples of delay and angles), often in the form of a closed-form expression \cite{gershman2010one}. Dominant in this category are the methods based on ESPRIT (for estimation of signal parameters via rotational invariance techniques) \cite{RoyKai89, XuSilRoyKai94, ZhaRakHaa21}, which exploit the fundamental structure of the signal in its various domains. While ESPRIT is limited to certain types of signals, it turns out to be well-suited to typical waveforms and antenna arrangements in 5G and beyond 5G, including the use of OFDM and uniform planar arrays \cite{wen2019survey}. 

ESPRIT methods have been developed based on matrices of the observation (called multi-dimensional ESPRIT (MD-ESPRIT)) \cite{SahUseCom17}, as well as on tensors  (called tensor ESPRIT) \cite{Haardt2018, ZhaRakHaa21} created from the observations. While the tensor version is arguably more natural, it suffers from high complexity due to the involved decomposition of the tensor.
%
% \textcolor{red}{[Fuxi: Relevant works include [ref] + discussion. (should include beamspace + positioning aspects + perturbation analysis)]}.
%
Furthermore, comparing with operation in element space, beamspace
model offers a compromise between system performance
and hardware complexity \cite{Heath2016}.
Tensor-based beamspace ESPRIT is proposed in \cite{ZhaHaa17, ZhaRakHaa21} for 3D mmWave channel estimation.
An $R$-D search-free  and general beamspace tensor-ESPRIT algorithm is proposed in \cite{WenGarKulWitWym18}. Robustness of the
method with respect to uncertainty in the number of sources,
as well as the applicability for sources with partially distinct frequencies, are demonstrated in \cite{WenSoWym20} for multidimensional harmonic retrieval. However, the performance is evaluated through numerical simulations, perturbation analysis and positioning accuracy are not considered.
In the case the observation is expressed as a large matrix or vector,
%
% \textcolor{red}{[Fuxi: Relevant works include [ref] + discussion (should include beamspace + positioning aspects + perturbation analysis).]} 
%
MD-ESPRIT for angular frequency estimation is proposed in \cite{SahUseCom17}, a perturbation analysis on angular frequency estimation is also provided. However, it is an element-space based method, while a beamspace counterpart was not developed or analyzed. 
%without estimation of the channel complex gain and evaluation on the positioning performance.

% From the above literature overview, it is clear that there are several gaps: 
% \cite{WenGarKulWitWym18, WenSoWym20} is based on tensor-based beamspace ESPRIT for angular frequency and channel estimation, but does not include positioning or perturbation analysis; 
% \cite{ZhaHaa17} on tensor-based beamspace ESPRIT for channel estimation, does not include  positioning or analytical results, and is limited to 3D, 
% \cite{SahUseCom17} MD-ESPRIT for angular frequency estimation. Not in positioning. Perturbation analysis on angular frequency estimation performance but not on complex gain and position estimation.

In this paper, we consider a beamspace channel model with OFDM and URAs at both transmitter and receiver and propose an efficient and low-complex multidimensional beamspace ESPRIT method, based on which the user location and communication rate are estimated. 
%With beamspace training, we obtain the beamspace observations used for SLAC. We then formulate the parametric channel estimation problem as the estimation of angular frequencies and complex gain. The earlier one is formulated to a multidimensional beamspace MHR problem, which will be solved by the low-complexity multidimensional beamspace ESPRIT approaches. The later one can be formulated as the maximum likehood estimation problem given angular frequencies, which can be efficiently solved by least-square method. In addition, by exploiting the underlying geometric relationship, we infer the user location with the estimated channel parameters. Furthermore, we perform the first-order perturbation analysis of the estimates with the proposed method for both channel parameters and location. 
%\textcolor{blue}{[Fan?]}
The contributions of this paper are as follows:
\begin{itemize}
    \item We develop a novel matrix-based multidimensional ESPRIT, operating in beamspace rather than channel space. To enable beamspace processing, we propose the use of modified selection matrices and show that  arbitrary precoding and combining matrices can be used, provided they are full column-rank in each dimension. We also confirm that the auto-pairing property of element-space MD-ESPRIT holds in beamspace. 
    \item We further propose a low-complexity implementation of the singular value decomposition (SVD) of the tall channel matrix in (beamspace) MD-ESPRIT, by using FFT/IFFT operations to obtain the signal subspace. We also demonstrate the low complexity through computational analysis and numerical simulations;
    \item We conduct a first-order perturbation analysis of the newly proposed MD-ESPRIT method, and evaluate both the analytical and numerical performance of the channel parameter and user location estimation. The results validate the closed-form analytical performance, especially in high signal-to-noise ratio (SNR) region;
    \item We evaluate the channel estimation and localization performance, as well as the effective achievable rate of the proposed method with both DFT and directional beams used in the system. The results reveal that our proposed beamspace MD-ESPRIT approaches achieve significant performance improvement compared to the methods in the literature;
    %\item We investigate the effective achievable rate performance with the channel estimate based on the proposed beamspace ESPRIT method. The results show that the near-optimal achievable rate performance can be achieved with the proposed method, nearly 2.4 Mbps rate improvement over existing method in working SNR region.
\end{itemize}
%To our best knowledge, this paper presents the first analytical results on SLAC with beamspace observations. 

The rest of the paper is organized as follows. In Section II, we introduce the system model including the multidimensional channel model, beamspace model, and beamspace observation model. We give an overview of the tensor-based ESPRIT approach for angular frequency estimation in Section III. After that, in section IV, we present the proposed matrix-based multidimensional beamspace ESPRIT approaches for SLAC, as well as the low-complexity implementations. Perturbation analysis is performed in Section V, where the closed-form analytical estimation performance of the proposed method is presented. Following that, we present the performance evaluation and discussions in Section VI. Finally, the conclusions are drawn in Section VII.

\subsubsection*{Notations}
%\textcolor{red}{[Shall we add our main notations here?]}
Vectors and matrices are denoted by bold lowercase and uppercase letters,
respectively. The notations $(\cdot)^*$, $(\cdot)^{\mathrm{T}}$, $(\cdot)^{\mathrm{H}}$, $(\cdot)^{-1}$, and $(\cdot)^{\dagger}$, are reserved for the conjugate, transpose, conjugate transpose, inverse, and Moore-Penrose pseudoinverse operations. The mathematical expectation is denoted by $\mathrm{E}\{ \cdot \}$. The notation $\mathrm{Diag}(\boldsymbol{a})$ is to form a diagonal matrix with $\boldsymbol{a}$ being the diagonal elements. The operations $\otimes$ and $\odot$ denote the Kronecker and Khatri-Rao products, respectively. The outer product is denoted by $\circ$. $\delta\left(\cdot \right)$ is the delta Dirac function. The notations $\Re\{ \cdot \}$ and $\Im \{ \cdot \}$ denote the operations to obtain the real and imaginary components, respectively.

%-------------------------------------------
%-------------------------------------------
%-------------------------------------------
\section{System Model}
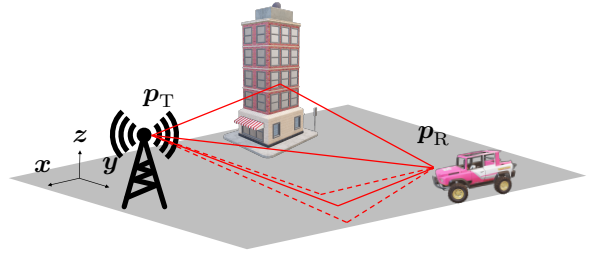
\begin{figure}
\begin{minipage}[!htb]{0.5\textwidth}
    \centering
    \begin{tikzpicture}
    \node(image) [anchor=south west] at (0,0){\includegraphics[width=0.85\textwidth]{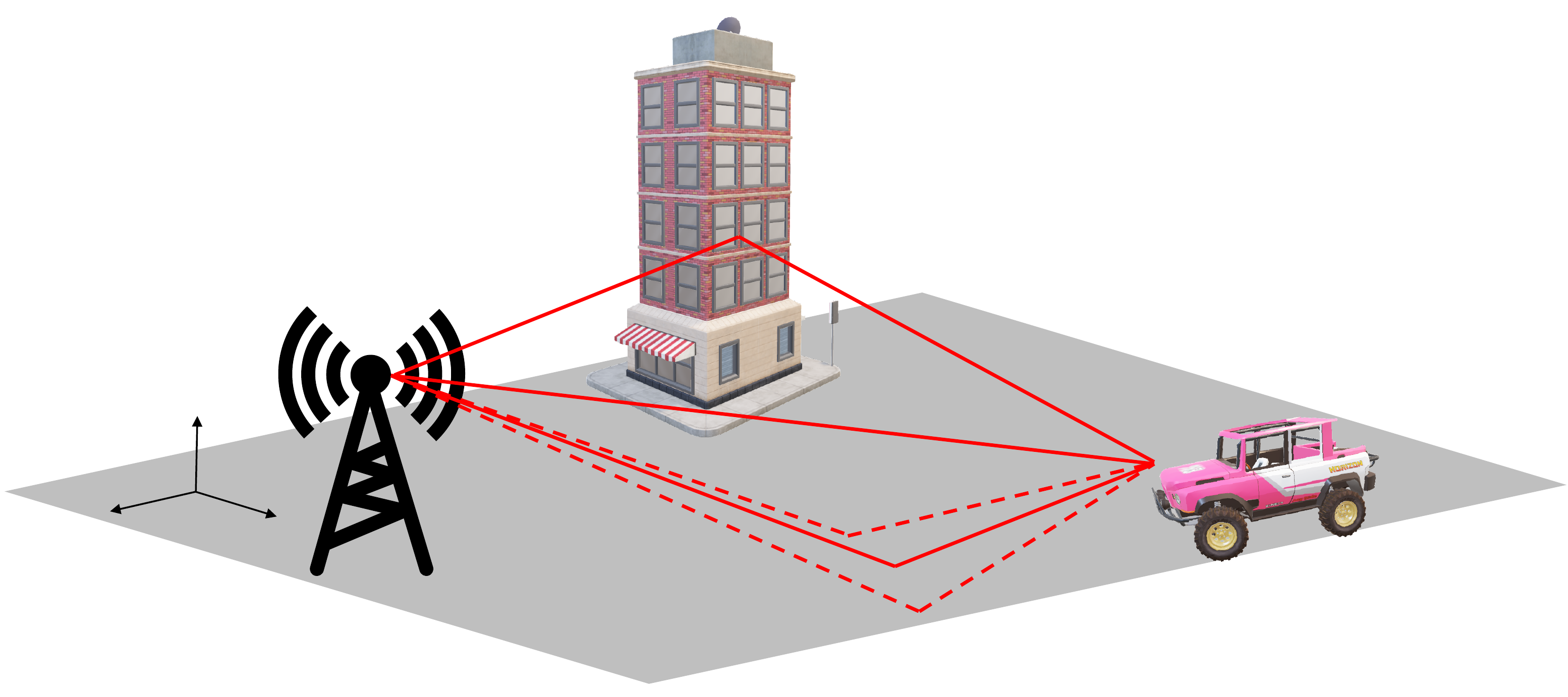}};
    %\draw [help lines] (0,0) grid (9,5);
    \node at (0.6,1.3) {$\boldsymbol{x}$};
    \node at (1.5,1.3) {$\boldsymbol{y}$};
    \node at (1.1,1.7) {$\boldsymbol{z}$};
    \node at (2.15,2.2) {$\boldsymbol{p}_{\mathrm{T}}$};
    \node at (5.8,1.7) {$\boldsymbol{p}_{\mathrm{R}}$};
    \end{tikzpicture}%\vspace{-5mm}
    \caption{3D localization and communications with multi-path propagation: vehicular user with unknown location using LOS/NLOS signals from base station with known position to localize itself; the vehicular user also expecting to know the effective achievable rate in communications.}
    \label{fig:3DScenario}
\end{minipage} \vspace{-5mm}
\end{figure}
\subsection{Multidimensional Channel Model}
We consider a 3-dimensional (3D) scenario for vehicular localization and communications as shown in Fig. \ref{fig:3DScenario}. The base station and the vehicular terminal are located at $\boldsymbol{p}_{\mathrm{T}}$ and $\boldsymbol{p}_{\mathrm{R}}$, respectively. Both the transmitter and the receiver are equipped with uniform rectangular arrays (URA) in $y-z$ plane consisting of $M_{\mathrm{T}} = M_1 \times M_2$ and $M_{\mathrm{R}} = M_3 \times M_4$ antenna elements, respectively. The antenna spacing is $\Delta d = \lambda/2$, where $\lambda$ is the wavelength. We consider a multi-path propagation scenario, where both line-of-sight (LOS) and non LOS (NLOS) paths are present in the system. The NLOS paths can be either reflection paths or diffuse multi-path depending on the smoothness of the objects in the environment. Each propagation path will be associated with AODs $\boldsymbol{\phi}_l = [\phi_{\mathrm{az}}^l, \phi_{\mathrm{el}}^l ]^{\mathrm{T}}$, AOAs $\boldsymbol{\theta}_l = [\theta_{\mathrm{az}}^l, \theta_{\mathrm{el}}^l ]^{\mathrm{T}}$, the propagation delay $\tau_l$, and the complex gain $\gamma_l$. We use $\phi_{\mathrm{az}}^l$ and $\phi_{\mathrm{el}}^l$ to denote the azimuth angle in the $x-y$ plane and the elevation angle for AODs, respectively, whereas $\theta_{\mathrm{az}}^l$ and $\theta_{\mathrm{el}}^l$ to denote the azimuth angle in the $x-y$ plane and the elevation angle for the AOAs, respectively. The geometric relationship between the positions and the channel parameters can be found in \cite{GeWenKimZhuJiaKimSveWym20, WenKulWitWym19}.

We assume OFDM transmission with $M_5$  subcarriers under narrow-band, far-field conditions, and possibly precoding and combining. To this end, we introduce  $\boldsymbol{T}_n^{\mathrm{H}} = [\boldsymbol{t}_{n, 1}, \boldsymbol{t}_{n, 2}, \cdots, \boldsymbol{t}_{n, M_n} ] \in \mathbb{C}^{N_n \times M_n}$ as a transformation matrix, allowing us to write a precoding matrix $\boldsymbol{F} =( \boldsymbol{T}_1 \otimes \boldsymbol{T}_2 )^*\in \mathbb{C}^{M_1M_2 \times N_1N_2}$ at transmitter side
and a  combining matrix $\boldsymbol{W} =\boldsymbol{T}_3 \otimes \boldsymbol{T}_4 \in \mathbb{C}^{M_3M_4 \times N_3N_4 }$ at receiver side. 
We assume no precoding is applied to the frequency domain, i.e., $\boldsymbol{T}_5 =  \boldsymbol{I}_{M_5}$. As a result, the beamspace channel matrix over the $m_5$-th subcarrier, $\boldsymbol{H}_{m_5}^{\left(\mathrm{b} \right)} \in \mathbb{C}^{N_3N_4\times N_1N_2}$ will be given by
\begin{align*}
\boldsymbol{H}_{m_5}^{\left(\mathrm{b} \right)} = \sum_{l=1}^L & \gamma_l e^{ -\jmath 2\pi ({m_5}-1) \Delta f \tau_l }  \boldsymbol{W}^{\mathrm{H}} \boldsymbol{a}_{\mathrm{R}}\left( \boldsymbol{\theta}_l \right) \boldsymbol{a}_{\mathrm{T}}^{\mathrm{T}}\left( \boldsymbol{\phi}_l \right) \boldsymbol{F}, \label{eq8} \numberthis
\end{align*}
where $\Delta f$ is the subcarrier spacing, $L$ denotes the number of propagation paths, $l=1$ is the LOS path and the rest ($l>1$) are NLOS paths. 
The steering vector corresponding to the $l$-th propagation path at transmitter side is given by \cite{WenKulWitWym19}
\begin{align*}
\boldsymbol{a}_{\mathrm{T}}\left( \boldsymbol{\phi}_l \right) = \boldsymbol{a}^{\left(M_1 \right)} \left(\omega_{l, 1} \right) \otimes \boldsymbol{a}^{\left(M_2 \right)} \left(\omega_{l, 2} \right), \numberthis
\end{align*}
where $\omega_{l, 1} = \pi \sin(\phi_{\mathrm{az}}^l) \sin (\phi_{\mathrm{el}}^l)$, $\omega_{l, 2} = \pi \cos (\phi_{\mathrm{el}}^l)$, $\otimes$ is the Kronecker product, and $\boldsymbol{a}^{(M )} (\omega )$ is a Vandermonde-structured vector defined as
\begin{align*}
\boldsymbol{a}^{\left(M \right)} \left(\omega \right) = \left[1, \exp \left(\jmath \omega \right), \cdots, \exp \left(\jmath \left(M-1 \right) \omega \right) \right]^{\mathrm{T}}.
\end{align*}
Similarly, the steering vector at the receiver side is given by
\begin{align*}
\boldsymbol{a}_{\mathrm{R}}\left( \boldsymbol{\theta}_l \right) = \boldsymbol{a}^{\left(M_3 \right)} \left(\omega_{l, 3} \right) \otimes \boldsymbol{a}^{\left(M_3 \right)} \left(\omega_{l, 4} \right), \numberthis
\end{align*}
where $\omega_{l, 3} = \pi \sin(\theta_{\mathrm{az}}^l ) \sin (\theta_{\mathrm{el}}^l )$, $\omega_{l, 4} = \pi \cos (\theta_{\mathrm{el}}^l )$.

It will be convenient to introduce $\boldsymbol{B}_n^{( N_n )} \in \mathbb{C}^{N_n \times L}$ as
\begin{align}
\boldsymbol{B}_n^{\left( N_n \right)} %= & \big[ \boldsymbol{b}^{\left(N_n\right)} \left( \omega_{1, n} \right), \boldsymbol{b}^{\left(N_n\right)} \left( \omega_{2, n} \right), \cdots, \boldsymbol{b}^{\left(N_n\right)} \left( \omega_{L, n} \right) \big] \\
= & \boldsymbol{T}_n^{\mathrm{H}} \boldsymbol{A}_n^{\left( M_n \right)}, \label{eq11} 
\end{align}
where 
$\boldsymbol{A}_n^{\left( M_n \right)} \in \mathbb{C}^{M_n \times L}$ is the Vandermonde matrix of the mode in the $n$-th dimension, given by
$\boldsymbol{A}_n^{( M_n )} =  [ \boldsymbol{a}^{(M_n)} ( \omega_{1, n} ), \boldsymbol{a}^{(M_n)} ( \omega_{2, n} ), \cdots, \boldsymbol{a}^{(M_n)} ( \omega_{L, n} ) ]$.

%-------------------------------------------------------------------------------------------------------------------------
\subsection{Beamspace Observation Model}
We consider a scenario where the channel coefficients remain static for a period of $N_\mathrm{c}$ OFDM blocks. The first $N_\mathrm{P}$ OFDM blocks will be used for pilots transmission, and the rest $N_\mathrm{D} = N_\mathrm{C} - N_\mathrm{P}$ OFDM blocks will be used for data transmission. We suppose the transmitter and receiver are equipped with $N_1\times N_2$ and $N_3\times N_4$ RF chains, respectively.\footnote{The beamspace observation in  can also be obtained with an arbitrary number of RF chains at transmitter and receivers. Let $N_{\mathrm{T}}$ and $N_{\mathrm{R}}$ denote the number of RF chains at transmitter and receiver, respectively. Generally, we can use $N_1N_2N_3N_4/N_{\mathrm{R}}$ transmission to obtain  \eqref{eq14}.} 
During the pilot transmission stage, the pilot matrix $\boldsymbol{S}_{m_5} \in \mathbb{C}^{N_1N_2 \times N_\mathrm{P}}$ is transmitted over the $m_5$-th subcarrier where $\boldsymbol{S}_{m_5}\boldsymbol{S}_{m_5}^{\mathrm{H}} = N_{\mathrm{P}}E_{\mathrm{s}} \boldsymbol{I}_{N_1N_2}$, 
%\begin{align*}
%\boldsymbol{S}_{m_5}\boldsymbol{S}_{m_5}^{\mathrm{H}} = N_{\mathrm{p}}E_{\mathrm{s}} \boldsymbol{I}_{N_1N_2}, \label{eq12} \numberthis
%\end{align*}
where we assume $N_{\mathrm{P}} \ge N_1N_2$ and the power of the pilot symbols are normalized to $E_{\mathrm{s}}$. With the transmit pilots, the corresponding beamspace received symbols over the $m_5$-th subcarrier will be
\begin{align}
\boldsymbol{Y}_{m_5} = \boldsymbol{H}_{m_5}^{\left(\mathrm{b} \right)}\boldsymbol{S}_{m_5} + \boldsymbol{Z}_{m_5}, \label{eq13} 
\end{align}
where $\boldsymbol{Y}_{m_5} \in \mathbb{C}^{N_3N_4 \times N_\mathrm{p}}$. The matrix $\boldsymbol{Z}_{m_5} \in \mathrm{C}^{N_3N_4 \times N_\mathrm{p}}$ is the beamspace noise component over the $m_5$-th subcarrier, with each entry modelled as independent and identically distributed (iid) zero mean complex Gaussian random variable. The variance of each entry is $N_0$. From \eqref{eq13}, we obtain the beamspace channel estimate $\hat{\boldsymbol{H}}_{m_5}^{\left(\mathrm{b} \right)} \in  \mathbb{C}^{N_3N_4\times N_1N_2}$ given by
\begin{align}
\hat{\boldsymbol{H}}_{m_5}^{\left(\mathrm{b} \right)} = \frac{\boldsymbol{Y}_{m_5} \boldsymbol{S}_{m_5}^{\mathrm{H}}}{N_{\mathrm{p}}E_{\mathrm{s}}} =  \boldsymbol{H}_{m_5}^{\left(\mathrm{b} \right)} + \Delta \boldsymbol{H}_{m_5}^{\left(\mathrm{b} \right)}, \label{eq14}
\end{align}
where $\Delta \boldsymbol{H}_{m_5}^{\left(\mathrm{b} \right)} \in  \mathbb{C}^{N_3N_4\times N_1N_2}$ is the channel estimation error matrix, with each entry $\Delta h_{n_1, n_2, n_3, n_4, m_5}$ modelled as iid zero mean complex Gaussian variable with variance ${N_0}/{(N_{\mathrm{p}}E_{\mathrm{s}})}$.
%\begin{align*}
%\mathbb{E}\left\{ \left| \Delta h_{n_1, n_2, n_3, n_4, m_5} \right|^2 \right\} = & \frac{N_0}{N_{\mathrm{p}}E_{\mathrm{s}}}. \label{eq15} \numberthis
%\end{align*}
%\eqref{eq15} 
%indicating that the channel estimation error will decrease if we increase the transmission power $E_{\mathrm{s}}$ and/or the length of the training sequence.

\subsection{Goals}
Collecting the beamspace channel estimates $\hat{\boldsymbol{H}}_{m_5}^{\left(\mathrm{b} \right)}$ over all subcarriers, our objectives are to estimate 
the angular frequencies $[\omega_{l,1},\ldots,\omega_{l,5}]^{\mathrm{T}}$ and $\gamma_l$ per path, as well as their uncertainties, 
and compute the estimates of the channel parameters $\boldsymbol{\eta}_l = [\boldsymbol{\phi}_l^{\mathrm{T}}, \boldsymbol{\theta}_l^{\mathrm{T}}, \tau_l, \gamma_l ]^{\mathrm{T}}$, based on the one-to-one mapping between the angular frequencies and the spatial channel parameters (TOAs, AOAs, AODs). 

Our second objective is the SLAC analysis, which involves computation of the user location and associated uncertainty, and assessment of the communication performance with beamspace observations. We use off-the-shelf solutions for localization and rate prediction, as shown below. We note that both the localization and communication performance are affected by the channel estimation. 

\subsubsection{Location Estimation}
\label{secIVD}
We employ the method from \cite{Wym18, WenKulWitWym19} for localization, due to its simple closed-form expression. Specifically, we define the unit vectors
\begin{align*} 
\boldsymbol{f}_{\mathrm{T}, l} = &  \big[\cos({\hat \phi}_{\mathrm{az}}^l )\sin({\hat \phi}_{\mathrm{el}}^l), \sin({\hat \phi}_{\mathrm{az}}^l )\sin({\hat \phi}_{\mathrm{el}}^l), \cos({\hat \phi}_{\mathrm{el}}^l)
\big]^{\mathrm{T}}, \\
\boldsymbol{f}_{\mathrm{R}, l} = &  \big[\cos({\hat \theta}_{\mathrm{az}}^l )\sin({\hat \theta}_{\mathrm{el}}^l), \sin({\hat \theta}_{\mathrm{az}}^l )\sin({\hat \theta}_{\mathrm{el}}^l), \cos({\hat \theta}_{\mathrm{el}}^l)
\big]^{\mathrm{T}},
\end{align*}
as the direction vectors for the $l$-th AODs and AOAs, respectively. For each propagation path, we define $\boldsymbol{\mu}_l=c{\hat \tau}_l \left( \boldsymbol{f}_{\mathrm{T}, l} + \boldsymbol{f}_{\mathrm{R}, l} \right)$ and $\boldsymbol{\delta}_l={\boldsymbol{p}_{\mathrm{T}} - c{\hat \tau}_l\boldsymbol{f}_{\mathrm{R}, l} }$. 
%establish
%\begin{align*} 
%\boldsymbol{p}_{\mathrm{R}} = &  \boldsymbol{p}_{\mathrm{T}} + c{\hat \tau}_l\xi_l \boldsymbol{f}_{\mathrm{T}, l} + c{\hat \tau}_l\left( 1 - \xi_l \right) \left( - \boldsymbol{f}_{\mathrm{R}, l}\right) \\
%= &\underbrace{\boldsymbol{p}_{\mathrm{T}} - c{\hat \tau}_l\boldsymbol{f}_{\mathrm{R}, l} }_{\boldsymbol{\delta}_l} + \xi_l \underbrace{c{\hat \tau}_l \left( \boldsymbol{f}_{\mathrm{T}, l} + \boldsymbol{f}_{\mathrm{R}, l} \right)}_{\boldsymbol{\mu}_l}, \label{eq:POS} \numberthis
%\end{align*}
%where $\xi_l \in \left[ 0, 1 \right]$ and $c$ is the speed of light.\footnote{For LOS path, $\xi_l$ can be arbitrary real value.} From \eqref{eq:POS}, we know that $\boldsymbol{p}_{\mathrm{R}}$ lies in the intersection of $L$ line segments. We consider the cost function $f(\boldsymbol{p}_{\mathrm{R}})$ as the sum of distance between $\boldsymbol{p}_{\mathrm{R}}$ and each line in \eqref{eq:POS}, i.e.,
%\begin{align*} 
%f\left(\boldsymbol{p}_{\mathrm{R}}\right) = \sum_{l=1}^L \iota_l \left\| \boldsymbol{p}_{\mathrm{R}} - \left( \boldsymbol{\delta}_l + \boldsymbol{\mu}_l \left( \boldsymbol{p}_{\mathrm{R}} - \boldsymbol{\delta}_l \right)^{\mathrm{T}} \boldsymbol{\mu}_l \right) \right\|^2,
%\end{align*}
%where $\left\{\iota_l \right\}_{l=1}^L \ge 0$ is the weight dependent on the SNR. 
The weighted least-squares solution of $\boldsymbol{p}_{\boldsymbol{R}}$ is given by
\begin{align*} 
\boldsymbol{\hat p}_{\mathrm{R}} = \left( \sum_{l=1}^L \boldsymbol{C}_l \right)^{-1} \sum_{l=1}^L \boldsymbol{C}_l \boldsymbol{\delta}_l, \label{eq:PosEst} \numberthis
\end{align*}
where $\boldsymbol{C}_l = \iota_l (\boldsymbol{I} - \boldsymbol{\mu}_l \boldsymbol{\mu}_l^{\mathrm{T}}/{ \| \boldsymbol{\mu}_l \|^2}  )$ and $\left\{\iota_l \right\}_{l=1}^L \ge 0$ is the weight dependent on the SNR.
%\begin{align*} 
%\boldsymbol{C}_l = \iota_l \left(\boldsymbol{I} - %\frac{\boldsymbol{\mu}_l \boldsymbol{\mu}_l^{\mathrm{T}}}{ \left\| %\boldsymbol{\mu}_l \right\|^2}  \right).
%\end{align*}
%
%With $\boldsymbol{\hat p}_{\mathrm{R}}$, the incident point corresponding to each propagation path can be recovered as the intersection of two lines $\boldsymbol{\hat p}_{\mathrm{R}} + \xi_{\mathrm{R}} \boldsymbol{f}_{\mathrm{R}, l}$ and $\boldsymbol{p}_{\mathrm{T}} + \xi_{\mathrm{T}} \boldsymbol{f}_{\mathrm{T}, l}$ \cite{Wym18, WenKulWitWym19}.

\subsubsection{Communication Rate Prediction}
 We follow the method in \cite{RaeGokZouBjoVal18} to compute the effective achievable rate; we first derive the signal-to-interference-plus-noise ratio (SINR), and then compute the sum-rate. After channel estimation, under single-stream transmission, we determine (possibly subcarrier dependent) optimal precoders $\boldsymbol{f}_{\text{com}}$ and combiners $\boldsymbol{w}_{\text{com}}$, leading to the effective channel on subcarrier $m_5$
 \begin{align}
     h_{m_5}=\underbrace{\boldsymbol{w}_{\text{com}}^{\mathrm{H}}\hat{\boldsymbol{H}}_{m_5}\boldsymbol{f}_{\text{com}}}_{U} - \underbrace{\boldsymbol{w}_{\text{com}}^{\mathrm{H}}\Delta {\boldsymbol{H}}_{m_5}\boldsymbol{f}_{\text{com}}}_{I},
 \end{align}
where $\hat{\boldsymbol{H}}_{m_5} = \sum_{l=1}^L  \hat{\gamma}_l e^{ -\jmath 2\pi ({m_5}-1) \Delta f \hat{\tau}_l }  \boldsymbol{a}_{\mathrm{R}}( \hat{\boldsymbol{\theta}}_l ) \boldsymbol{a}_{\mathrm{T}}^{\mathrm{T}}( \hat{\boldsymbol{\phi}}_l )$, and $\Delta {\boldsymbol{H}}_{m_5}$ is the estimation error, due to the observation noise and the channel estimation routine. 
We can compute a SINR over the $m_5$-th subcarrier, $\eta_{m_5}=|U|^2/(N_0+\mathrm{E}(|I|^2))$, leading to an effective achievable rate
\begin{align}
R = \frac{N_{\mathrm{C}} - N_{\mathrm{P}}}{M_5N_{\mathrm{C}}}\sum_{m_5 = 1}^{M_5} \log_2 \left( 1 + \eta_{m_5} \right). \label{eq:achievableRate}
\end{align}

\section{Tensor-based Beamspace ESPRIT}
% \textcolor{red}{Fuxi will be responsible for this part. I expect two subsections here: 1) how the problem is formulated? and 2) the algorithm to perform angular frequency estimation. Much details can be referred to the conference papers while using the pseudo code for algorithm.}
The high-dimensional beamspace channel can be naturally represented by the tensor $\boldsymbol{\mathcal{H}}^{\left( \mathrm{b} \right)} \in \mathbb{C}^{N_1\times N_2\times\cdots \times N_5}$, with  canonical polyadic model 
\begin{align*}
\boldsymbol{\mathcal{H}}^{\left( \mathrm{b} \right)} = \sum_{l=1}^{L}\gamma_l & \boldsymbol{b}^{\left( N_1 \right)} \left( \omega_{l, 1} \right) \circ \ldots   \circ \boldsymbol{b}^{\left( N_5 \right)} \left( \omega_{l, 5} \right), \label{eq9} \numberthis
\end{align*}
%\begin{align*}
%\boldsymbol{\mathcal{H}}^{\left( \mathrm{b} \right)} = \sum_{l=1}^{L}\gamma_l & \boldsymbol{b}^{\left( N_1 \right)} \left( \omega_{l, 1} \right) \circ \boldsymbol{b}^{\left( N_2 \right)} \left( \omega_{l, 2} \right) \circ \boldsymbol{b}^{\left( N_3 \right)} \left( \omega_{l, 3} \right) \\
%& \hspace{5mm} \circ \boldsymbol{b}^{\left( N_4 \right)} \left( \omega_{l, 4} \right) \circ \boldsymbol{b}^{\left( N_5 \right)} \left( \omega_{l, 5} \right), \label{eq9} \numberthis
%\end{align*}
 where $\circ$ denotes the outer product and
$\boldsymbol{b}^{\left( N_n \right)} \left( \omega_{l, n} \right) = \boldsymbol{T}_n^{\mathrm{H}} \boldsymbol{a}^{\left( M_n \right)} \left( \omega_{l, n} \right)$, 
and $N_5 = M_5$. To estimate the angular frequencies $\omega_{l, n}$ from a noisy version of $\boldsymbol{\mathcal{H}}^{\left( \mathrm{b} \right)}$, 
we briefly recap the beamspace tensor ESPRIT method from \cite{WenSoWym20}.

\subsection{Tensor Decomposition}

CANDECOMP/PARAFAC (CP) and Tucker are two widely used tensor decomposition approaches \cite{Kolda2009}.
By apply CP decomposition on noisy measurement $\hat{\boldsymbol{\mathcal{H}}}^{\left( \mathrm{b} \right)}$ in (\ref{eq9}), we obtain $\hat{\boldsymbol{\mathcal{H}}}^{\left( \mathrm{b} \right)} \approx \sum_{l=1}^{L}\lambda_l  \boldsymbol{u}_{l, 1}  \circ \boldsymbol{u}_{l, 2} \circ \boldsymbol{u}_{l, 3}\circ \boldsymbol{u}_{l, 4}\circ \boldsymbol{u}_{l, 5}$, 
%
%\begin{align*}
%\hat{\boldsymbol{\mathcal{H}}}^{\left( \mathrm{b} \right)} \approx \sum_{l=1}^{L}\lambda_l  \boldsymbol{u}_{l, 1}  \circ \boldsymbol{u}_{l, 2} \circ \boldsymbol{u}_{l, 3}\circ \boldsymbol{u}_{l, 4}\circ \boldsymbol{u}_{l, 5}, \label{eq9x} \numberthis
%\end{align*}
where $\lambda_l, l = 1, 2, \cdots, L$, are the dominant eigenvalues.
The eigenvectors in the $n$-th dimension can be described as
\begin{align}
\boldsymbol{U}_n = 
\begin{bmatrix}
\boldsymbol{u}_{1, n} & \boldsymbol{u}_{2, n} & \cdots & \boldsymbol{u}_{L, n}
     \end{bmatrix}
\in \mathbb{C}^{N_n \times L}, n = 1, 2, \cdots, 5. \label{eqUn} 
\end{align}
Since $\boldsymbol{B}_n^{\left( N_n \right)}$ is also spanned by the signal subspace, there exists a non-singular matrix $\boldsymbol{E}_n \in \mathbb{C}^{L\times L}$, satisfying
%\begin{align*}
$\boldsymbol{B}_n^{\left( N_n \right)} = \boldsymbol{U}_n \boldsymbol{E}_n.$% \label{eq16} \numberthis
%\end{align*}

\subsection{Tensor-based ESPRIT for Angular Frequency Estimation}
We first define the following selection matrices $\boldsymbol{J}_{n, 1} =   [ \boldsymbol{I}_{N_n -1}, \boldsymbol{0}_{(N_n - 1 ) \times 1}]$ and  $\boldsymbol{J}_{n, 2} =  [ \boldsymbol{0}_{(N_n - 1 ) \times 1}, \boldsymbol{I}_{N_n -1} ]$ and
%\begin{align*}
%\boldsymbol{J}_{n, 1} =  & \big[ \boldsymbol{I}_{N_n -1}, \boldsymbol{0}_{\left(N_n - 1 \right) \times 1} \big]; \\
%\boldsymbol{J}_{n, 2} =  & \big[ \boldsymbol{0}_{\left(N_n - 1 \right) \times 1}, \boldsymbol{I}_{N_n -1} \big].
%\end{align*}
let $\boldsymbol{\Phi}_n = \mathrm{Diag}\{ [\Phi_{1, n}, \Phi_{2, n}, \cdots, \Phi_{L, n} ]^{\mathrm{T}} \}$ with $\Phi_{l, n} = \exp{(\jmath \omega_{l, n})}$. ESPRIT relies on the so-called  shift invariance property to perform the angular frequency estimation, meaning that $\boldsymbol{J}_{n, 1}\boldsymbol{A}_n^{\left( N_n \right)} \boldsymbol{\Phi}_n = \boldsymbol{J}_{n, 2} \boldsymbol{A}_n^{\left( N_n \right)}$. In beamspace, generally with the transformation matrix $\boldsymbol{T}_n$, $\boldsymbol{J}_{n, 1}\boldsymbol{B}_n^{\left( N_n \right)} \boldsymbol{\Phi}_n \ne \boldsymbol{J}_{n, 2} \boldsymbol{B}_n^{\left( N_n \right)}$. However, if $\boldsymbol{T}_n$ itself has a shift invariance structure, the following proposition shows that the lost shift invariance structure can be restored.
\begin{proposition} \label{Th1}
Assume that $\boldsymbol{T}_n \in \mathbb{C}^{M_n \times N_n}$ has the shift invariance structure $\boldsymbol{J}_{n, 1} \boldsymbol{T}_n = \boldsymbol{J}_{n, 2} \boldsymbol{T}_n \boldsymbol{F}_n$ where $\boldsymbol{F}_n \in \mathbb{C}^{N_n\times N_n}$ is a non-singular matrix. If there exists a matrix $\boldsymbol{Q}_n \in \mathbb{C}^{N_n \times N_n}$, such that $\boldsymbol{Q}_n \boldsymbol{t}_{n, M_n} = \boldsymbol{0}_{N_n \times 1}$ and $\boldsymbol{Q}_n \boldsymbol{F}_n^{\mathrm{H}}\boldsymbol{t}_{n, 1} = \boldsymbol{0}_{N_n \times 1}$, 
%\begin{align*}
%\left\{ 
%\begin{array}{l}
 %   \boldsymbol{Q}_n \boldsymbol{t}_{n, M_n} = \boldsymbol{0}_{N_n \times 1}, \\
  %  \boldsymbol{Q}_n \boldsymbol{F}_n^{\mathrm{H}}\boldsymbol{t}_{n, 1} = \boldsymbol{0}_{N_n \times 1}, \label{eq17} \numberthis
%\end{array}
%\right.
%\end{align*}
then  $\boldsymbol{Q}_n \boldsymbol{B}_n^{\left( N_n \right)} \boldsymbol{\Phi}_n = \boldsymbol{Q}_n \boldsymbol{F}_n^{\mathrm{H}}\boldsymbol{B}_n^{\left( N_n \right)}$.
%\begin{align*}
%\boldsymbol{Q}_n \boldsymbol{B}_n^{\left( N_n \right)} \boldsymbol{\Phi}_n = \boldsymbol{Q}_n \boldsymbol{F}_n^{\mathrm{H}}\boldsymbol{B}_n^{\left( N_n \right)}. \label{eq18} \numberthis
%\end{align*}
\end{proposition}
% The proof is known in literature while I put here for reference
\begin{proof}
The proof can be found in \cite{WenGarKulWitWym18, XuSilRoyKai94}.
% From \eqref{eq17}, we have
%\begin{align*}
%\boldsymbol{Q}_n \boldsymbol{T}_n^{\mathrm{H}} = & \left[\boldsymbol{Q}_n \boldsymbol{t}_{n, 1}, \cdots, \boldsymbol{Q}_n \boldsymbol{t}_{n, {M_n} - 1}, \boldsymbol{0} \right] = \boldsymbol{Q}_n \boldsymbol{T}_n^{\mathrm{H}} \boldsymbol{J}_{n, 1}^{\mathrm{H}} \boldsymbol{J}_{n, 1}, \\
%\boldsymbol{Q}_n \boldsymbol{F}_n^{\mathrm{H}} \boldsymbol{T}_n^{\mathrm{H}} = & \boldsymbol{Q}_n \boldsymbol{F}_n^{\mathrm{H}} \boldsymbol{T}_n^{\mathrm{H}} \boldsymbol{J}_{n, 2}^{\mathrm{H}} \boldsymbol{J}_{n, 2}.
%\end{align*}
%Multiplying $\boldsymbol{A}_n^{\left( M_n \right)} \boldsymbol{\Phi}_n $ on both sides in the upper equation, we have
%\begin{align*}
%\boldsymbol{Q}_n \boldsymbol{B}_n^{\left( M_n \right)} \boldsymbol{\Phi}_n = & \boldsymbol{Q}_n \underbrace{\boldsymbol{T}_n^{\mathrm{H}} \boldsymbol{J}_{n, 1}^{\mathrm{H}} } \underbrace{ \boldsymbol{J}_{n, 1} \boldsymbol{A}_n^{\left( M_n \right)} \boldsymbol{\Phi}_n} \\
%= & \underbrace{ \boldsymbol{Q}_n \boldsymbol{F}_n^{\mathrm{H}} \boldsymbol{T}_n^{\mathrm{H}} \boldsymbol{J}_{n, 2}^{\mathrm{H}} \boldsymbol{J}_{n, 2} } \boldsymbol{A}_n^{\left( M_n \right)} \\
%= & \boldsymbol{Q}_n \boldsymbol{F}_n^{\mathrm{H}} \underbrace{ \boldsymbol{T}_n^{\mathrm{H}} \boldsymbol{A}_n^{\left( M_n \right)} } \\
%= & \boldsymbol{Q}_n \boldsymbol{F}_n^{\mathrm{H}} \boldsymbol{B}_n^{\left( M_n \right)}.
%\end{align*}
\end{proof}
%-------------------------------------
\begin{algorithm}[tb]
\caption{ Beamspace Tensor ESPRIT Approach \cite{WenSoWym20}} \label{alg:BTE}
\begin{algorithmic}[1]
\Require \parbox[t]{\dimexpr\linewidth- \algorithmicindent * 1}{Noisy beamspace observation $\hat{\boldsymbol{\mathcal{H}}}^{\left( \mathrm{b}\right)}$;\strut}
\Ensure \parbox[t]{\dimexpr\linewidth- \algorithmicindent * 1}{Channel parameter estimates: $\phi_{\mathrm{el}}^l$, $\phi_{\mathrm{az}}^l$, $\theta_{\mathrm{el}}^l$, $\theta_{\mathrm{az}}^l$, and $\tau_l$, for $l=1, 2, \cdots, L$.\strut}

\State Apply CP decomposition on $\hat{\boldsymbol{\mathcal{H}}}^{\left( \mathrm{b}\right)}$ to obtain $\boldsymbol{U}_n$ (\ref{eqUn});

% \State Pre-processing to obtain $\boldsymbol{F}_n$ (\ref{eqFn}) and then $\boldsymbol{p}_{n, 1}$ and $\boldsymbol{p}_{n, 2}$;

\State Determine $\boldsymbol{F}_n$ and $\mathbf{Q}_n$ from Proposition \ref{Th1};

\State Obtain $\boldsymbol{L}_{n, 1} = \boldsymbol{Q}_n $ and $\boldsymbol{L}_{n, 2} = \boldsymbol{Q}_n \boldsymbol{F}_n^{\mathrm{H}}$;

\State Obtain $\boldsymbol{\Gamma}_n = \left( \boldsymbol{L}_{n, 2} \boldsymbol{U}_n \right)^{\dagger} \left( \boldsymbol{L}_{n, 1} \boldsymbol{U}_n \right)$
%\begin{align*}
%\boldsymbol{\Gamma}_n = \left( \boldsymbol{L}_{n, 2} \boldsymbol{U}_n \right)^{\dagger} \left( \boldsymbol{L}_{n, 1} \boldsymbol{U}_n \right).
%\end{align*}

\State Perform the eigenvalue decomposition of $\boldsymbol{\Gamma}_n$ to obtain the matrix of complex eigenvalues $\boldsymbol{D}_n = \mathrm{Diag} ( \left[ d_{n, 1}, d_{n, 2}, \cdots, d_{n, L} \right]^{\mathrm{T}} )$, and the angular frequencies in the $n$-th mode are given by $\hat{\omega}_{l, n} = \Im \{\ln\left( \tilde{\Phi}_{l, n} \right)\}$
%\begin{align*}
%\omega_{l, n} = \Im \{\ln\left( \tilde{\Phi}_{l, n} \right)\},
%\end{align*}
for $l=1, 2, \cdots, L$.

\State Obtain the channel parameter estimates \cite{GeWenKimZhuJiaKimSveWym20, WenKulWitWym19}.

\end{algorithmic}
\end{algorithm}
 
Proposition \ref{Th1} indicates that we can restore the shift invariance properties by forming the matrix $\boldsymbol{Q}_n$ as a projection matrix corresponding to the orthogonal subspace spanned by  $\boldsymbol{t}_{n, M_n}, \boldsymbol{F}_n^{\mathrm{H}} \boldsymbol{t}_{n, 1}$. In particular, let $\boldsymbol{p}_{n, 1}$ and $\boldsymbol{p}_{n, 2}$ be the orthogonal column vectors for this subspace, then  $\boldsymbol{Q}_n$ can be obtained as $\boldsymbol{Q}_n = \boldsymbol{I}_{N_n} - \sum_{k=1}^2 \boldsymbol{p}_{n, k} \boldsymbol{p}_{n, k}^{\mathrm{H}}$. 
%\begin{align}
%\boldsymbol{Q}_n = \boldsymbol{I}_{N_n} - \sum_{k=1}^2 \boldsymbol{p}_{n, k} \boldsymbol{p}_{n, k}^{\mathrm{H}}.\label{eqQn}
%\end{align}
When a $\boldsymbol{Q}_n$ is determined, the shift-invariance property in the $n$-th mode can be restored with the following modified selection matrices
\begin{align}
\boldsymbol{L}_{n, 1} = & \boldsymbol{Q}_n, \label{eqLn1} \\
\boldsymbol{L}_{n, 2} = & \boldsymbol{Q}_n \boldsymbol{F}_n^{\mathrm{H}}. \label{eqLn2}
\end{align}
When the shift invariance structure does not hold for $\boldsymbol{T}_n$, i.e., $\boldsymbol{J}_{n, 1} \boldsymbol{T}_n \ne \boldsymbol{J}_{n, 2} \boldsymbol{T}_n \boldsymbol{F}_n$, we can find an approximate non-singular $N_n \times N_n$ matrix $\tilde{ \boldsymbol{F}}_n$, such that
%\begin{align*}
$\boldsymbol{J}_{n, 1} \boldsymbol{T}_n \approx \boldsymbol{J}_{n, 2} \boldsymbol{T}_n \tilde{\boldsymbol{F}}_n$.
The least-square solution is given as $\tilde{\boldsymbol{F}}_n = \left( \boldsymbol{J}_{n, 2} \boldsymbol{T}_n \right)^{\dagger} \boldsymbol{J}_{n, 1} \boldsymbol{T}_n$ from which $\boldsymbol{Q}_n$ can be obtained as before.
%\begin{align}
%\label{eqFn}
%\end{align}

Finally, the beamspace Tensor ESPRIT method \cite{WenSoWym20} is summarized in Algorithm \ref{alg:BTE}.

% Step 1: CPD of the tensor $\boldsymbol{\mathcal{K}}$ and obtain the leading $L$ signal subspace in the $n$-th mode as $\boldsymbol{U}_n \in \mathbb{C}^{N_n \times L}$;

% Step 2: Pre-processing to obtain $\boldsymbol{F}_n$ and then $\boldsymbol{p}_{n, 1}$ and $\boldsymbol{p}_{n, 2}$;

% Step 3: Obtain $\boldsymbol{L}_{n, 1} = \boldsymbol{Q}_n $ and $\boldsymbol{L}_{n, 2} = \boldsymbol{Q}_n \boldsymbol{F}_n^{\mathrm{H}}$;

% Step 4: Obtain 
% \begin{align*}
% \boldsymbol{\Gamma}_n = \left( \boldsymbol{L}_{n, 2} \boldsymbol{U}_n \right)^{\dagger} \left( \boldsymbol{L}_{n, 1} \boldsymbol{U}_n \right).
% \end{align*}

% Step 5: Perform the eigenvalue decomposition of $\boldsymbol{\Gamma}_n$ to obtain $\boldsymbol{D}_n = \mathrm{Diag} \left( \left[ d_{n, 1}, d_{n, 2}, \cdots, d_{n, L} \right]^{\mathrm{T}} \right)$, and the angular frequencies in the $n$-th mode is given by
% \begin{align*}
% \boldsymbol{\omega}_{l, n} = \ln{ \left( d_{n, l} \right) },
% \end{align*}
% for $l=1, 2, \cdots, L$.

%------------------------------------------------------------------------------------------------------------------------
\section{Matrix-based Beamspace ESPRIT}
%------------------------------------------------------------------------------------------------------------------------
%By transforming the tensor $\boldsymbol{\mathcal{H}}$ to a vector $\boldsymbol{h} \in \mathbb{C}^{M_1M_2\cdots M_5 \times 1}$ where
%\begin{align*}
%\boldsymbol{h} = & \big[h_{1, 1, 1, 1, 1}, h_{1, 1, 1, 1, 2}, \cdots, h_{1, 1, 1, 1, M_5}, \\ & \hspace{15mm}h_{1, 1, 1, 2, 1}, \cdots, h_{M_1, M_2, M_3, M_4, M_5} \big]^{\mathrm{T}},
%\end{align*}
%we have
%\begin{align*}
%\boldsymbol{h} = \mathrm{vec}_\mathrm{r} \left\{ \boldsymbol{\mathcal{H}} \right\} = \left( \boldsymbol{A}_1^{\left( M_1 \right)} \odot \boldsymbol{A}_2^{\left( M_2 \right)} \odot \cdots \odot \boldsymbol{A}_5^{\left( M_5 \right)} \right) \boldsymbol{\gamma}, \label{eq7} \numberthis
%\end{align*}

As an alternative to the tensor approach, we present a novel matrix-based approach for beamspace ESPRIT. By vectorizing the beamspace tensor $\boldsymbol{\mathcal{H}}^{\left( \mathrm{b} \right)}$, we arrive at
\begin{align*}
\boldsymbol{h}^{\left( \mathrm{b} \right)} & =  \mathrm{vec}_{\mathrm{r}} \{ \boldsymbol{\mathcal{H}}^{( \mathrm{b} )} \}=\big[h^{\left( \mathrm{b} \right)}_{1, 1, 1, 1, 1}, h^{\left( \mathrm{b} \right)}_{1, 1, 1, 1, 2}, \cdots, h^{\left( \mathrm{b} \right)}_{1, 1, 1, 1, M_5}, \\ & \hspace{15mm}h^{\left( \mathrm{b} \right)}_{1, 1, 1, 2, 1}, \cdots, h^{\left( \mathrm{b} \right)}_{N_1, N_2, N_3, N_4, M_5} \big]^{\mathrm{T}},\\
& =\left( \boldsymbol{B}_1^{\left( N_1 \right)} \odot \boldsymbol{B}_2^{\left( N_2 \right)} \odot \cdots \odot \boldsymbol{B}_5^{\left( N_5 \right)} \right) \boldsymbol{\gamma}, \label{eq10} \numberthis
\end{align*}
where $\odot$ is the Khatri-Rao product, $\boldsymbol{\gamma} = \left[ \gamma_1, \gamma_2, \cdots, \gamma_L \right]^{\mathrm{T}}$.

The proposed method has three parts: spatial smoothing for improved resolution, a novel low-complexity SVD of the channel matrix, and auto-paired angular frequency estimation.

\subsection{Spatial Smoothing}
Considering $\boldsymbol{T}_5 = \boldsymbol{I}_{M_5}$, we can apply spatial smoothing in the frequency domain. Specifically, we define the selection matrices
\begin{align*}
\boldsymbol{J}_{\ell_5} = & \left[\boldsymbol{0}_{K_5 \times \left( \ell_5 - 1 \right) }, \boldsymbol{I}_{K_5}, \boldsymbol{0}_{K_5 \times \left( L_5 - \ell_5 \right) } \right], \\
\boldsymbol{J}_{N_1, N_2, N_3, N_4, \ell_5} = & \boldsymbol{I}_{N_1} \otimes \boldsymbol{I}_{N_2} \otimes \boldsymbol{I}_{N_3} \otimes \boldsymbol{I}_{N_4} \otimes \boldsymbol{J}_{\ell_5},
\end{align*}
where $L_5 + K_5 = M_5 + 1$. We then construct a matrix $\boldsymbol{H} \in \mathbb{C}^{N_1 N_2 N_3 N_4 K_5 \times L_5}$ as
\begin{align*}
\boldsymbol{H} = & \mathcal{S}( \boldsymbol{h}^{\left( \mathrm{b} \right) } ) = \big[ \boldsymbol{J}_{N_1, N_2, N_3, N_4, 1} \boldsymbol{h}^{\left( \mathrm{b} \right) }, \\
& \hspace{6mm} \boldsymbol{J}_{N_1, N_2, N_3, N_4, 2} \boldsymbol{h}^{\left( \mathrm{b} \right) },
\cdots, \boldsymbol{J}_{N_1, N_2, N_3, N_4, L_5} \boldsymbol{h}^{\left( \mathrm{b} \right) }
\big], \label{eq19} \numberthis
\end{align*}
where we use $\mathcal{S}\left( \cdot \right)$ to denote the smoothing operation over the frequency domain.

% With DFT beams, we can also apply the forward-backward averaging techniques to further explore the data samples in beamspace. Specifically, we define a new vector $\breve{\boldsymbol{h}}^{\left( \mathrm{b} \right) } \in \mathbb{C}^{N_1N_2\cdots N_5 \times 1}$ as
%\begin{align*}
%\breve{\boldsymbol{h}}^{\left( \mathrm{b} \right) } = \boldsymbol{\Psi} \left( \boldsymbol{h}^{\left( \mathrm{b} \right) %}\right)^*, \label{eq23} \numberthis
%\end{align*}
%where $\boldsymbol{\Psi} \in \mathbb{C}^{N_1N_2\cdots N_5 \times N_1N_2\cdots N_5}$ is a permutation matrix with ones on its anti-diagonal. Similarly, we construct a matrix $\breve{\boldsymbol{H}} \in \mathbb{C}^{M_1 M_2 M_3 M_4 K_5 \times L_5}$ as
%\begin{align*}
%\breve{\boldsymbol{H}} = & \big[ \boldsymbol{J}_{M_1, M_2, M_3, M_4, 1} \breve{\boldsymbol{h}}^{\left( \mathrm{b} \right) }, %\boldsymbol{J}_{M_1, M_2, M_3, M_4, 2} \breve{\boldsymbol{h}}^{\left( \mathrm{b} \right) }, \\ 
%& \hspace{26mm} \cdots, \boldsymbol{J}_{M_1, M_2, M_3, M_4, L_5} \breve{\boldsymbol{h}}^{\left( \mathrm{b} \right) }
%\big]. \label{eq24} \numberthis
%\end{align*}
%From Appendix \ref{AppendixA}, the matrix $\breve{\boldsymbol{H}} $ can be factorized as
%\begin{align*}
%\breve{\boldsymbol{H}} = \boldsymbol{P} \mathrm{Diag} \left( \breve{\boldsymbol{\gamma}} \right) \boldsymbol{G}^{\mathrm{T}}, %\label{eq25} \numberthis
%\end{align*}
%where $\breve{\boldsymbol{\gamma}}$ is given by

%------------------------------------------------------------------------------------------------------------------------
\subsection{Signal Subspace and Shift Invariance Properties}
From Appendix \ref{AppendixA}, we can see that $\boldsymbol{H}$ can be factorized as
\begin{align*}
\boldsymbol{H} = \boldsymbol{P} \mathrm{Diag} \left( \boldsymbol{\gamma} \right) \boldsymbol{G}^{\mathrm{T}}, \label{eq20} \numberthis
\end{align*}
where $\boldsymbol{P} \in \mathbb{C}^{N_1 N_2 N_3 N_4 K_5 \times L}$ and $\boldsymbol{G} \in \mathbb{C}^{L_5 \times L} $, given by
\begin{align*}
\boldsymbol{P} = & \boldsymbol{B}_1^{\left( N_1 \right)} \odot \cdots \odot \boldsymbol{B}_4^{\left( N_4 \right)} \odot \boldsymbol{A}_5^{\left( K_5 \right)}, \label{eq21} \numberthis\\
\boldsymbol{G} = & \boldsymbol{A}_5^{\left( L_5 \right)}. \label{eq22} \numberthis
\end{align*}
The following proposition shows that the shift invariance properties can be restored with the matrix formulation.
\begin{proposition} \label{propA1}
Assume that $\boldsymbol{T}_n$, $n=1, 2, 3, 4$, has the shift invariance structure, and $\boldsymbol{P}$ is given in \eqref{eq21}, we then have
\begin{align*}
\breve{\boldsymbol{J}}_{n, 1}\boldsymbol{P} \boldsymbol{\Phi}_n = \breve{\boldsymbol{J}}_{n, 2}\boldsymbol{P}, \label{eqIV1} \numberthis
\end{align*}
where $\breve{\boldsymbol{J}}_{n, i}$, $ n = 1, 2, \cdots, N$, and $i= 1, 2$, is defined as follows:
\begin{align*}
\breve{\boldsymbol{J}}_{n, i} = \left\{ 
\begin{array}{ll}
\boldsymbol{I}_{N_1} \otimes \cdots \otimes \boldsymbol{L}_{n, i} \otimes \cdots \otimes \boldsymbol{I}_{N_4} \otimes \boldsymbol{I}_{K_5}, & n=1, 2, 3, 4; \\
\boldsymbol{I}_{N_1} \otimes \cdots \otimes \boldsymbol{I}_{N_4} \otimes \bar{\boldsymbol{J}}_{5, i}, & n=5;
\end{array}
\right.
\end{align*}
with the modified selection matrices $\boldsymbol{L}_{n, 1}$ and $\boldsymbol{L}_{n, 2}$ given in \eqref{eqLn1} and \eqref{eqLn2}, respectively, and
\begin{align*}
\bar{\boldsymbol{J}}_{5, i} = \left\{ 
\begin{array}{ll}
\left[\boldsymbol{I}_{K_5 - 1}, \boldsymbol{0}_{\left( K_5 -1 \right)\times 1} \right], & i=1; \\
\left[\boldsymbol{0}_{\left( K_5 -1 \right)\times 1}, \boldsymbol{I}_{K_5 - 1} \right], & i=2.
\end{array}
\right.
\end{align*}
\end{proposition}
\begin{proof}
The proof can be found in Appendix \ref{ProofpropA1}.
\end{proof}
Proposition \ref{propA1} indicates that the shift invariance properties can be perfectly restored if the transformation matrix in each mode has the desired shift invariance structure. When the shift invariance structure does not hold, i.e., $\boldsymbol{J}_{n, 1} \boldsymbol{T}_n \ne \boldsymbol{J}_{n, 2} \boldsymbol{T}_n \boldsymbol{F}_n$, we can find an approximate non-singular matrix $\tilde{ \boldsymbol{F}}_n$, such that
$\boldsymbol{J}_{n, 1} \boldsymbol{T}_n \approx \boldsymbol{J}_{n, 2} \boldsymbol{T}_n \tilde{\boldsymbol{F}}_n$ as before. After that, the approximation $\tilde{ \boldsymbol{F}}_n$ is used in proposition \ref{propA1}. Therefore, the proposed framework can be applied to system with arbitrary beamforming and combination matrices.

Performing singular value decomposition on $\boldsymbol{H}$, we have 
\begin{align*}
\boldsymbol{H} = \boldsymbol{U}_{\mathrm{s}} \boldsymbol{\Sigma}_{\mathrm{s}} \boldsymbol{V}_{\mathrm{s}}^{\mathrm{H}}, \numberthis \label{eq23}
\end{align*}
where the $L$ leading left singular vectors in $\boldsymbol{U}_{\mathrm{s}}$ span the column space of $\boldsymbol{H}$. Note the same space is spanned by the columns of $\boldsymbol{P}$, we then have a non-singular matrix $\boldsymbol{E} \in \mathbb{C}^{L\times L}$ such that $\boldsymbol{P} = \boldsymbol{U}_{\mathrm{s}} \boldsymbol{E}$. 
%\begin{align*}
%\boldsymbol{P} = \boldsymbol{U}_{\mathrm{s}} \boldsymbol{E}. \numberthis %\label{eq24}
%\end{align*}
Substituting into \eqref{eqIV1}, we have
\begin{align*}
\breve{\boldsymbol{J}}_{n, 1}\boldsymbol{U}_{\mathrm{s}} \tilde{\boldsymbol{\Gamma}}_n = \breve{\boldsymbol{J}}_{n, 2} \boldsymbol{U}_{\mathrm{s}}, \numberthis \label{eq26}
\end{align*}
where $\tilde{\boldsymbol{\Gamma}}_n \in \mathbb{C}^{L\times L}$ is given by
\begin{align*}
\tilde{\boldsymbol{\Gamma}}_n = \boldsymbol{E} \boldsymbol{\Phi}_n \boldsymbol{E}^{-1}. \numberthis \label{eq27}
\end{align*}
The least-square solution of $\tilde{\boldsymbol{\Gamma}}_n$ from \eqref{eq26} can be obtained as
\begin{align*}
\tilde{\boldsymbol{\Gamma}}_n \approx \left( \breve{\boldsymbol{J}}_{n, 1} \boldsymbol{U}_{\mathrm{s}} \right)^{\dagger} \breve{\boldsymbol{J}}_{n, 2} \boldsymbol{U}_{\mathrm{s}}. \numberthis \label{eq28}
\end{align*}
We perform the eigenvalue decomposition on $\tilde{\boldsymbol{\Gamma}}_n$ obtained from \eqref{eq28}, and the angular frequency in each dimension can be obtained through the eigenvalues as indicated by \eqref{eq27}.

%------------------------------------------------------------------------------------------------------------------------
\subsection{Channel Parameter Estimation}
The independent eigenvalue decomposition operations in each mode cannot achieve auto-pairing. Moreover, when identical or very close angular frequencies in the $n$-th mode is present, the eigenvectors corresponding to the same eigenvalues cannot be distinguished. This could happen, for example, with close scatterer points in the environments. To achieve auto-pairing of the angular frequencies crossing all dimensions, we construct a new matrix $\boldsymbol{K} \in \mathbb{C}^{L\times L}$ as
\begin{align*}
\boldsymbol{K} = \sum_{n=1}^{N} \beta_n \tilde{\boldsymbol{\Gamma}}_n, \numberthis \label{eq29}
\end{align*}
where $\boldsymbol{\beta} = \left[ \beta_1, \beta_2, \cdots, \beta_N \right]^{\mathrm{T}}$ is a random vector. With the results in \eqref{eq27}, we have
\begin{align*}
\boldsymbol{K} = & \sum_{n=1}^{N} \beta_n \boldsymbol{E} \boldsymbol{\Phi}_n \boldsymbol{E}^{-1} = \boldsymbol{E} \boldsymbol{\Lambda} \boldsymbol{E}^{-1}, \numberthis \label{eq30}
\end{align*}
where $\boldsymbol{\Lambda} \in \mathbb{C}^{L\times L}$ is a diagonal matrix given by $\boldsymbol{\Lambda} = \sum_{n=1}^{N} \beta_n \boldsymbol{\Phi}_n$.  
%\begin{align*}
%. \numberthis \label{eq31}
%\end{align*}
This indicates that even if the angular frequencies in the $n$-th dimension are the same, the diagonal elements in $\boldsymbol{\Lambda}$ can be likely distinct, since the angular frequencies in other dimensions can be different. As a result, instead of applying eigenvalue decomposition operations on each $\tilde{\boldsymbol{\Gamma}}_n$, we only need to perform one eigenvalue decomposition on $\boldsymbol{K}$. As long as the diagonal elements in $\boldsymbol{\Lambda}$ are distinct from each other, a unique $\boldsymbol{E}$ can be obtained, which is used to recover the angular frequencies in each mode. This leads to auto-pairing of angular frequencies in all dimensions. With the acquisition of $\boldsymbol{E}$, from \eqref{eq27}, we have
\begin{align*}
\tilde{\boldsymbol{\Phi}}_n = \boldsymbol{E}^{-1} \tilde{\boldsymbol{\Gamma}}_n \boldsymbol{E}. \numberthis \label{eq32}
\end{align*}
The angular frequency $\omega_{l, n}$ is obtained as
\begin{align*}
\omega_{l, n} = \Im \{\ln\left( \tilde{\Phi}_{l, n} \right)\}. \numberthis \label{eq33}
\end{align*}

%------------------------------------------------------------------------------------------------------------------------
\begin{algorithm}[t]
\caption{Matrix-based Beamspace ESPRIT Approach for Channel Parameter Estimation} \label{alg:AFE}
\begin{algorithmic}[1]
\Require \parbox[t]{\dimexpr\linewidth- \algorithmicindent * 1}{Noisy beamspace observation in vector $\tilde{\boldsymbol{h}}^{\left( \mathrm{b} \right)} = \mathrm{vec}_\mathrm{r} \left\{ \hat{\boldsymbol{\mathcal{H}}}^{\left( \mathrm{b}\right) } \right\}$;\strut}
\Ensure \parbox[t]{\dimexpr\linewidth- \algorithmicindent * 1}{Channel parameter estimates: $\phi_{\mathrm{el}}^l$, $\phi_{\mathrm{az}}^l$, $\theta_{\mathrm{el}}^l$, $\theta_{\mathrm{az}}^l$, $\tau_l$, and $\gamma_l$, for $l=1, 2, \cdots, L$.\strut}
\State Choose $L_5$ and let
$K_5 \gets M_5 + 1 - L_5$;
\State Construct the matrix $\tilde{\boldsymbol{H}}$ according to \eqref{eq19} with $\tilde{\boldsymbol{h}}^{\left( \mathrm{b} \right)}$;
\State Perform the SVD of $\tilde{\boldsymbol{H}}$, and extract the signal subspace $\tilde{\boldsymbol{U}}_{\mathrm{s}}$ with the leading $L$ singular vectors;
\State Compute $\tilde{\boldsymbol{ \Gamma}}_n$ from \eqref{eq28} for $n=1, 2, \cdots, N$ with $\tilde{\boldsymbol{U}}_{\mathrm{s}}$;
\State Generate random $\boldsymbol{\beta}$, and compute $\tilde{\boldsymbol{ K}}$ from \eqref{eq29};
\State Perform the ED of $\tilde{\boldsymbol{ K}}$ to extract $\tilde{\boldsymbol{E}}$ as shown in \eqref{eq30};
\State Apply $\tilde{\boldsymbol{E}}$ to compute $\tilde{\boldsymbol{ \Psi}}_n$ from \eqref{eq32} for $n=1, 2, \cdots, N$;
\State Estimate the angular frequencies ${\hat \omega}_{l, n}$ from \eqref{eq33};% for $l=1, 2, \cdots, L$ and $n=1,2, \cdots, N$;
\State Obtain the channel parameter estimates.
\end{algorithmic}
\end{algorithm}
%------------------------------------------------------------------------------------------------------------------------

With the estimated angular frequencies, we reconstruct the matrix $\hat{\boldsymbol{B}} \in \mathbb{C}^{N_1N_2\cdots N_5 \times L}$ as
\begin{align*}
\hat{\boldsymbol{B}} = \hat{\boldsymbol{B}}_1^{\left( N_1 \right)} \odot \hat{\boldsymbol{B}}_2^{\left( N_2 \right)} \odot \cdots \odot \hat{\boldsymbol{B}}_4^{\left( N_4 \right)} \odot \hat{\boldsymbol{A}}_5^{\left( M_5 \right)}, \numberthis \label{eq34}
\end{align*}
where $\hat{\boldsymbol{B}}_n^{\left( N_n \right)}$ is computed from \eqref{eq11}. 
The corresponding channel parameters associated with each propagation path can be obtained from the one-to-one mapping and $\hat{\boldsymbol{\gamma}} =  \left(\hat{\boldsymbol{B}} \right)^{\dagger} \tilde{\boldsymbol{h}}^{\left( \mathrm{b} \right)}$. 
%\begin{align*}
%\phi_{\mathrm{el}}^l = & \arccos \left(\frac{\omega_{l, 2}}{\pi} \right), \numberthis \label{eq35} \\
%\phi_{\mathrm{az}}^l = & \arcsin \left(\frac{\omega_{l, 1}}{\pi \sin \left( \phi_{\mathrm{el}}^l \right) } \right), \numberthis \label{eq36} \\
%\theta_{\mathrm{el}}^l = & \arccos \left(\frac{\omega_{l, 4}}{\pi} \right), \numberthis \label{eq37} \\
%\theta_{\mathrm{az}}^l = & \arcsin \left(\frac{\omega_{l, 3}}{\pi \sin \left( \theta_{\mathrm{el}}^l \right) } \right), \numberthis \label{eq38} \\
%\tau_l = & -\frac{\omega_{l, 5}}{2\pi \Delta f}, \numberthis \label{eq39} \\
%\hat{\boldsymbol{\gamma}} = & \left(\hat{\boldsymbol{B}} \right)^{\dagger} \tilde{\boldsymbol{h}}^{\left( \mathrm{b} \right)}. \numberthis \label{eq40}
%\end{align*}
%
The matrix-based beamspace ESPRIT algorithm is presented in Algorithm \ref{alg:AFE}.

%------------------------------------------------------------------------------------------------------------------------
%\subsection{Localization and Mapping}\label{secIVD}

\subsection{Variations of the Hybrid ESPRIT Approaches}
When $\boldsymbol{T}_n$ does not hold the shift-invariance properties, the approximation $\tilde{\boldsymbol{F}}_n$ from the least-square solution requires the full-column rank conditions of $\boldsymbol{J}_{n, 2}\boldsymbol{T}_n$. However, if the conditions are not valid, for example, in the case $M_n < N_n$, we can apply the hybrid ESPRIT approaches instead. Specifically, when $M_n < N_n$, and $\boldsymbol{T}_n$ is full-row rank\footnote{In the design of the transform matrix $\boldsymbol{T}_n$, we can easily meet the full rank conditions. When using random beams in the system, the full rank conditions are satisfied with high probability.}, we define the tensor $\breve{\boldsymbol{\mathcal{H}}}_n \in \mathbb{C}^{N_1 \times \cdots \times M_n \times \cdots \times N_5}$ as the $n$-th mode product of $\boldsymbol{\mathcal{H}}^{( \mathrm{b} )}$ and the matrix $( \boldsymbol{T}_n^{\mathrm{H}})^{\dagger}$, i.e., 
\begin{align*}
\breve{\boldsymbol{\mathcal{H}}}_n = & \boldsymbol{\mathcal{H}}^{\left( \mathrm{b} \right)} \times_n \left( \boldsymbol{T}_n^{\mathrm{H}}\right)^{\dagger}. \numberthis
\end{align*}
As a result, the corresponding canonical polyadic model of $\breve{\boldsymbol{\mathcal{H}}}_n$ can be given as
\begin{align*}
\breve{\boldsymbol{\mathcal{H}}}_n = & \sum_{l=1}^L \gamma_l \boldsymbol{b}^{\left( N_1 \right)} \left( \omega_{l, 1} \right) \circ \cdots \circ \boldsymbol{a}^{\left( M_n \right)} \left( \omega_{l, n} \right) \circ \cdots \circ \boldsymbol{b}^{\left( N_5 \right)} \left( \omega_{l, 5} \right),
\end{align*}
where we have applied the results $( \boldsymbol{T}_n^{\mathrm{H}} )^{\dagger} \boldsymbol{T}_n^{\mathrm{H}} = \boldsymbol{I}_{M_n}$. Therefore, the angular frequency estimation in the $n$-th mode can be obtained through the element-space ESPRIT approaches as we have done in the frequency domain. For simplicity, we use the new selection matrices $\boldsymbol{L}_{n, 1} \in \mathbb{C}^{(M_n - 1) \times M_n}$ and $\boldsymbol{L}_{n, 2} \in \mathbb{C}^{(M_n - 1) \times M_n}$ as
\begin{align*}
\boldsymbol{L}_{n, i} = & \left\{
\begin{array}{ll}
\left[\boldsymbol{I}_{M_n -1}, \boldsymbol{0}_{\left( M_n -1\right)\times 1 } \right], & i=1; \\
\left[ \boldsymbol{0}_{\left( M_n -1\right)\times 1 }, \boldsymbol{I}_{M_n -1} \right], & i=2,
\end{array}
\right.
\end{align*}
and $\breve{\boldsymbol{\mathcal{H}}}_n$ in the matrix-based beamspace ESPRIT approaches. This process can be repeated if the shift invariance properties does not hold in multiple $\boldsymbol{T}_n$, hence we refer the new approach as to the hybrid ESPRIT approach.

%------------------------------------------------------------------------------------------------------------------------
\subsection{Reduced-Complexity Matrix-Based Beamspace ESPRIT}
By exploiting the structure of $\tilde{\boldsymbol{H}}$, we propose a low-complexity SVD algorithm, which is the computational bottleneck in Algorithm \ref{alg:AFE}.

\subsubsection{Proposed SVD algorithm}
 To begin with, we define a set of matrices $\tilde{\boldsymbol{H}}_{n_1,n_2,n_3,n_4} \in \mathbb{C}^{K_5 \times L_5}$ as
\begin{align*}
\tilde{\boldsymbol{H}}_{n_1,n_2,n_3,n_4} = & \left[ 
\begin{array}{cccc}
\tilde{h}_{n_1,n_2,n_3,n_4,1}^{\left( b \right)} & \cdots & \tilde{h}_{n_1,n_2,n_3,n_4,L_5}^{\left( b \right)} \\
\vdots & \ddots & \vdots \\
\tilde{h}_{n_1,n_2,n_3,n_4,K_5}^{\left( b \right)} & \cdots & \tilde{h}_{n_1,n_2,n_3,n_4,N_5}^{\left( b \right)}
\end{array}
\right].
\end{align*}
It immediately arrives that the matrix $\tilde{\boldsymbol{H}}$ can be constructed as
\begin{align*}
\tilde{\boldsymbol{H}} = \left[ 
\begin{array}{c}
\tilde{\boldsymbol{H}}_{1,1,1,1} \\
\tilde{\boldsymbol{H}}_{1,1,1,2} \\
\vdots \\
%\tilde{\boldsymbol{H}}_{1,1,1,N_4} \\
\tilde{\boldsymbol{H}}_{1,1,2,1} \\
\vdots \\
\tilde{\boldsymbol{H}}_{N_1,N_2,N_3,N_4} \\
\end{array}
\right].
\end{align*}

The SVD of $\tilde{\boldsymbol{H}}$ consists of two steps: Lanczos bidiagonalization and the SVD of bidiagonalization matrix. Specifically, the Lanczos bidiagonalization is to transform $\tilde{\boldsymbol{H}}$ into an upper bidiagonal matrix $\boldsymbol{J} \in \mathbb{C}^{L_5 \times L_5}$ as $\boldsymbol{J} = \boldsymbol{U}_{\mathrm{L}}^{\mathrm{H}} \tilde{\boldsymbol{H}}\boldsymbol{V}_{\mathrm{L}}$, 
%\begin{align*}
%\boldsymbol{J} = \boldsymbol{U}_{\mathrm{L}}^{\mathrm{H}} \tilde{\boldsymbol{H}}\boldsymbol{V}_{\mathrm{L}},
%\end{align*}
where $\boldsymbol{U}_{\mathrm{L}} \in \mathbb{C}^{N_1 N_2 N_3 N_4 K_5 \times L_5}$, and $\boldsymbol{V}_{\mathrm{L}} \in \mathbb{C}^{L_5 \times L_5}$ satisfying
%\begin{align*}
$\boldsymbol{U}_{\mathrm{L}}^{\mathrm{H}} \boldsymbol{U}_{\mathrm{L}} = \boldsymbol{V}_{\mathrm{L}}\boldsymbol{V}_{\mathrm{L}}^{\mathrm{H}} = \boldsymbol{V}_{\mathrm{L}}^{\mathrm{H}}\boldsymbol{V}_{\mathrm{L}} = \boldsymbol{I}_{L_5}$.
%\end{align*}
Let $\boldsymbol{U}_{\mathrm{L}} = [\boldsymbol{u}_1, \boldsymbol{u}_2, \cdots, \boldsymbol{u}_{L_5}]$, and $\boldsymbol{V}_{\mathrm{L}} = [\boldsymbol{v}_1, \boldsymbol{v}_2, \cdots, \boldsymbol{v}_{L_5}]$. The upper bidiagonal matrix $\boldsymbol{J}$ can be represented by two vectors, i.e., $\boldsymbol{a} = [a_1, a_2, \cdots, a_{L_5}]^{\mathrm{T}}$ and $\boldsymbol{b} = [b_1, b_2, \cdots, b_{L_5 - 1}]^{\mathrm{T}}$, and
\begin{align*}
\boldsymbol{J} = \left[ 
\begin{array}{ccccc}
    a_1 & b_1 & & \cdots & 0 \\
    %& \ddots & \ddots & & \vdots \\
    & & \ddots & \ddots & \\
    \vdots & & & \ddots & b_{L_5 -1} \\
    0 & \cdots & & & a_{L_5}
\end{array}
\right].
\end{align*}
The vectors $\boldsymbol{u}_{\ell_5}$ and $\boldsymbol{v}_{\ell_5}$ can be obtained through the Lanczos recursions, i.e.,
\begin{align*} 
a_{\ell_5} \boldsymbol{u}_{\ell_5} = & \tilde{\boldsymbol{H}}\boldsymbol{v}_{\ell_5} - b_{{\ell_5}-1} \boldsymbol{u}_{{\ell_5}-1}, \numberthis\\
b_{\ell_5} \boldsymbol{v}_{{\ell_5}+1} = & \tilde{\boldsymbol{H}}^{\mathrm{H}}\boldsymbol{u}_{\ell_5} - a_{\ell_5} \boldsymbol{v}_{\ell_5}, \numberthis
\end{align*}
given the initialization vector $\boldsymbol{u}_0$ and $b_0$. After the Lanczos bidiagonalization, we perform the SVD of $\boldsymbol{J}$, i.e.,
\begin{align*} 
\boldsymbol{J} = \boldsymbol{U}_{\mathrm{J}} \boldsymbol{\Sigma}_{\mathrm{J}}\boldsymbol{V}_{\mathrm{J}}^{\mathrm{H}}, \numberthis
\end{align*}
where both $\boldsymbol{U}_{\mathrm{J}} \in \mathbb{C}^{L_5 \times L_5}$ and $\boldsymbol{V}_{\mathrm{J}} \in \mathbb{C}^{L_5 \times L_5}$ are unitary matrices. As a result, the left unitary matrix corresponding to the SVD of $\tilde{\boldsymbol{H}}$ can be obtained as
\begin{align*} 
\boldsymbol{U}_{\mathrm{H}} = \boldsymbol{U}_{\mathrm{L}}\boldsymbol{U}_{\mathrm{J}}, \numberthis
\end{align*}
and the leading $L$ singular vectors in $\boldsymbol{U}_{\mathrm{H}}$ are extracted as signal subspace, i.e., $\tilde{\boldsymbol{U}}_{\mathrm{s}} =( \boldsymbol{U}_{\mathrm{H}})_{:, 1:L}$.
Note $\tilde{\boldsymbol{H}}_{n_1,n_2,n_3,n_4}$ is a Hankel matrix, and the matrix-vector product can be efficiently implemented with fast Fourier transform (FFT) and inverse-FFT (IFFT) algorithms \cite{JiaGeZhuWym21, LuXuQia15}. As a result, the SVD of $\tilde{\boldsymbol{H}}$ can be implemented efficiently.

%------------------------------------------------------------------------------------------------------------------------
\subsubsection{Complexity Analysis}\label{secIVF}
In Algorithm \ref{alg:AFE}, we notice that the main complexity comes from the SVD of $\tilde{\boldsymbol{H}}$. From \cite{JiaGeZhuWym21, LuXuQia15}, the matrix-vector product for $\tilde{\boldsymbol{H}}_{n_1,n_2,n_3,n_4}$ can be efficiently implemented with $\mathcal{O} (N_5 \log N_5 )$ computations. Therefore, the matrix-vector product for $\tilde{\boldsymbol{H}}$ requires $\mathcal{O} (J \log N_5 )$ computations where $J=N_1 N_2 N_3 N_4 N_5$. The Lanczos bidiagonalization process can be implemented with $\mathcal{O} (JN_5 \log N_5 )$. In addition, the SVD of a bidiagonal matrix can be implemented with $\mathcal{O}(N_5^2)$. Therefore, the overall complexity to perform the SVD of $\tilde{\boldsymbol{H}}$ is $\mathcal{O} (JN_5 \log N_5 )$. However, we only need to extract the $L$ leading singular vectors; therefore, the overall complexity is further reduced to $\mathcal{O} (LJ \log N_5 )$, which is almost linearly with the $LJ$. On the other hand, the major computations in Algorithm \ref{alg:BTE} comes from the CP decomposition of the tensor $\hat{\boldsymbol{\mathcal{H}}}^{( \mathrm{b})}$ with the size of $N_1\times N_2\times N_3\times N_4\times M_5$, which requires $\mathcal{O} (2^N LJ + NL^3  )$ if the alternating least squares (ALS) algorithm with line search is used \cite{WenKulWitWym19}. As indicated in \cite{Haardt2018}, the tensor-based approaches have the complexity in the same order of the matrix based approaches, but require much more computations, which will be demonstrated in Section \ref{SubSecVIE} by evaluating the program running time. 

%------------------------------------------------------------------------------------------------------------------------
%------------------------------------------------------------------------------------------------------------------------
%------------------------------------------------------------------------------------------------------------------------
\section{Perturbation Analysis}
In this section, we provide the closed-form analytical performance of the channel parameter and location estimation. This can be done by performing the first-order perturbation analysis. It is worth pointing out that the proposed method does not require the statistical knowledge of the noise component since it is directly applied to the beamspace observations. As a result, the Cram\'{e}r–Rao bound (CRB) may not be available if noise statistics are unknown. In addition, as the noise component is not considered, the performance gap between the ESPRIT approaches and the CRB has been clearly identified in existing literature \cite{SahUseCom17, Haardt2018, LiuLiu06, LiuLiuMa07}. On the other hand, the first-order perturbation analysis, as we will show in the next section, matches well with the simulations, which is more meaningful to demonstrate the estimation performance of the proposed method.
\subsection{First-order Perturbation of Channel Parameter Estimation}
With $\boldsymbol{h}^{\left( \mathrm{b} \right) }$, we construct the noiseless matrix $\boldsymbol{H}$ from \eqref{eq19}. The SVD of $\boldsymbol{H}$ is given by
\begin{align*}
\boldsymbol{H} = \boldsymbol{U}_{\mathrm{s}} \boldsymbol{\Sigma}_{\mathrm{s}} \boldsymbol{V}_{\mathrm{s}}^{\mathrm{H}} + \boldsymbol{U}_{\mathrm{n}} \boldsymbol{\Sigma}_{\mathrm{n}} \boldsymbol{V}_{\mathrm{n}}^{\mathrm{H}}, \label{eq47} \numberthis
\end{align*} 
and $\boldsymbol{\Sigma}_{\mathrm{n}} = \boldsymbol{0}$. The estimates, $\tilde{\boldsymbol{H}}$ is expressed as
\begin{align*}
\tilde{\boldsymbol{H}} = \boldsymbol{H} + \Delta \boldsymbol{H}, \label{eq48} \numberthis
\end{align*} 
where $\Delta \boldsymbol{H} = \mathcal{S} (\Delta \boldsymbol{h})$, and $\Delta \boldsymbol{h} \in \mathbb{C}^{N_1N_2N_3N_4M_5}$ is the beamspace channel estimation error.

With the noisy estimates, the first-order perturbation $\Delta \Phi_{l, n} = \tilde{\Phi}_{l, n} - \Phi_{l, n}$ from the proposed algorithm is given as follows.
\begin{lemma}\label{lem1}
$\Delta \Phi_{l, n}$ is given by
\begin{align*}
\Delta \Phi_{l, n} =\frac{1}{\gamma_l} \boldsymbol{b}_l^{\mathrm{T}} \left( \breve{\boldsymbol{J}}_{n, 1} \boldsymbol{P} \right)^{\dagger} \left( \breve{\boldsymbol{J}}_{n, 2} - \Phi_{l, n} \breve{\boldsymbol{J}}_{n, 1} \right) \Delta \boldsymbol{H} \left( \boldsymbol{G}^{\mathrm{T}} \right)^{\dagger} \boldsymbol{b}_l, \label{eq49} \numberthis
\end{align*}
where $\boldsymbol{b}_l$ is the $l$-th column of $\boldsymbol{I}_L$. The matrix $\boldsymbol{P}$ and $\boldsymbol{G}$ are given by \eqref{eq21} and \eqref{eq22}, respectively.
\end{lemma}
\begin{proof}
The proof can be found in Appendix \ref{ProofLem1}.
\end{proof}
Denote $\boldsymbol{\chi}_{l} \in \mathbb{C}^{L_5 \times 1}$ and $\boldsymbol{\lambda}_{l, n} \in \mathbb{C}^{N_1 N_2 N_3 N_4 K_5 \times 1} $ as
\begin{align*}
\boldsymbol{\lambda}_{l, n }^{\mathrm{H}} = & \boldsymbol{b}_l^{\mathrm{T}} \left( \breve{\boldsymbol{J}}_{n, 1} \boldsymbol{P} \right)^{\dagger} \left( \breve{\boldsymbol{J}}_{n, 2} - \Phi_{l, n} \breve{\boldsymbol{J}}_{n, 1} \right), \numberthis \\
\boldsymbol{\chi}_{l}^* = & \left( \boldsymbol{G}^{\mathrm{T}} \right)^{\dagger} \boldsymbol{b}_l.
\numberthis
\end{align*}
The vector $\boldsymbol{\lambda}_{l, n } $ can also be represented by
\begin{align*}
\boldsymbol{\lambda}_{l, n } = & \Big[ \left( \boldsymbol{\lambda}_{l, n}^{\left(1, 1, 1, 1 \right)} \right)^{\mathrm{T}}, \left( \boldsymbol{\lambda}_{l, n}^{\left( 1, 1, 1, 2\right)} \right)^{\mathrm{T}}, \cdots, \left( \boldsymbol{\lambda}_{l, n}^{\left(1, 1, 1, N_1 \right)} \right)^{\mathrm{T}}, \\
& \hspace{10mm} \left( \boldsymbol{\lambda}_{l, n}^{\left(1, 1, 2, 1 \right)} \right)^{\mathrm{T}}, \cdots, \left( \boldsymbol{\lambda}_{l, n}^{\left(N_1, N_2, N_3, N_4 \right)} \right)^{\mathrm{T}} \Big]^{\mathrm{T}}, \numberthis
\end{align*}
where $\boldsymbol{\lambda}_{l, n}^{\left(n_1, n_2, n_3, n_4 \right)} \in \mathbb{C}^{K_5 \times 1}$. We further define $\boldsymbol{\xi}_{l, n } $ as
\begin{align*}
\boldsymbol{\xi}_{l, n } = & \Big[ \left( \boldsymbol{\xi}_{l, n}^{\left(1, 1, 1, 1 \right)} \right)^{\mathrm{T}}, \left( \boldsymbol{\xi}_{l, n}^{\left(1, 1, 1, 2 \right)} \right)^{\mathrm{T}}, \cdots, \left( \boldsymbol{\xi}_{l, n}^{\left(1, 1, 1, N_1 \right)} \right)^{\mathrm{T}}, \\
& \hspace{10mm} \left( \boldsymbol{\xi}_{l, n}^{\left(1, 1, 2, 1 \right)} \right)^{\mathrm{T}}, \cdots, \left( \boldsymbol{\xi}_{l, n}^{\left(N_1, N_2, N_3, N_4 \right)} \right)^{\mathrm{T}} \Big]^{\mathrm{T}}, \label{eq53} \numberthis
\end{align*}
where $\boldsymbol{\xi}_{l, n}^{\left(n_1, n_2, n_3, n_4 \right)} \in \mathbb{C}^{M_5 \times 1}$ is the convolution of $\boldsymbol{\lambda}_{l, n}^{\left(n_1, n_2, n_3, n_4 \right)}$ and $\boldsymbol{\chi}_{l}$. We then have the following proposition.
\begin{proposition} \label{prop1}
With the beamspace channel estimation error $\Delta \boldsymbol{h}$, the first-order perturbation $\Delta \Phi_{l, n}$ in \eqref{eq49} is equivalent to
\begin{align*}
\Delta \Phi_{l, n} = \frac{\boldsymbol{\xi}_{l, n }^{\mathrm{H}} \Delta \boldsymbol{h}}{\gamma_l}, \label{eq54} \numberthis
\end{align*}
where $\boldsymbol{\xi}_{l, n }$ is given by \eqref{eq53}.
\end{proposition}
\begin{proof}
The proof can be found in Appendix \ref{ProofProp1}.
\end{proof}
\emph{Remarks}: Compared to the first-order perturbation expression in Lemma \ref{lem1}, the results in Proposition \ref{prop1} has a two-fold merit. First of all, the expression in \eqref{eq54} simplifies the first-order perturbation $\Delta \Phi_{l, n} $ in \eqref{eq49}, enabling the perturbation analysis without knowing $\Delta \boldsymbol{H}$. Secondly, equation \eqref{eq54} also simplifies the analysis of the statistic properties (including the mean and variance) of the estimation errors. The latter one is more important in the perturbation analysis of the complex gain and position estimation.

From \eqref{eq33}, the perturbation of $\omega_{l, n}$ is readily given by
\begin{align*}
\Delta \omega_{l, n} = \Im \left\{ \boldsymbol{\upsilon}_{l, n}^{\mathrm{H}} \Delta \boldsymbol{h} \right\}, \label{eq55} \numberthis
\end{align*}
where $\boldsymbol{\upsilon}_{l, n} = \Phi_{l, n} \boldsymbol{\xi}_{l, n }/{\gamma_l^*}$.

We are now able to give the first-order perturbation of the channel parameters for each path.
\begin{lemma} \label{lem2}
The first-order perturbations of the channel parameters can be given as
\begin{align*}
\Delta \phi_{\mathrm{az}}^l = & \Im \left\{ \boldsymbol{\kappa}_{l, 1}^{\mathrm{H}} \Delta \boldsymbol{h} \right\}, \label{eq56} % \numberthis \\
\Delta \phi_{\mathrm{el}}^l =  \Im \left\{ \boldsymbol{\kappa}_{l, 2}^{\mathrm{H}} \Delta \boldsymbol{h} \right\},   \numberthis \\
\Delta \theta_{\mathrm{az}}^l = & \Im \left\{ \boldsymbol{\kappa}_{l, 3}^{\mathrm{H}} \Delta \boldsymbol{h} \right\}, \label{eq58}  %\numberthis \\
\Delta \theta_{\mathrm{el}}^l =  \Im \left\{ \boldsymbol{\kappa}_{l, 4}^{\mathrm{H}} \Delta \boldsymbol{h} \right\},   \numberthis \\
\Delta \tau_l = & \Im \left\{ \boldsymbol{\kappa}_{l, 5}^{\mathrm{H}} \Delta \boldsymbol{h} \right\}, \label{eq60} \numberthis \\
\Delta \gamma_l = & \boldsymbol{b}_l^{\mathrm{T}} \boldsymbol{B}^{\dagger} \Delta \boldsymbol{h} - \sum_{n=1}^5 \boldsymbol{b}_l^{\mathrm{T}} \boldsymbol{\Upsilon}_n \Im \left\{ \boldsymbol{V}_n^{\mathrm{H}} \Delta \boldsymbol{h} \right\}, \label{eq61} \numberthis
\end{align*}
where $\boldsymbol{V}_n = [\boldsymbol{\upsilon}_{1, n}, \boldsymbol{\upsilon}_{2, n}, \cdots, \boldsymbol{\upsilon}_{L, n} ]$. The matrix $\boldsymbol{\kappa}_{l, i} \in \mathbb{C}^{N_1N_2N_3N_4M_5 \times 1}$, and $\boldsymbol{\Upsilon}_n \in \mathbb{C}^{N_1N_2N_3N_4M_5 \times L}$ are given as
\begin{align*}
\boldsymbol{\kappa}_{l, 1} = & \frac{\boldsymbol{\upsilon}_{l,1} }{\pi \cos \left( \phi_{\mathrm{az}}^l \right) \sin \left( \phi_{\mathrm{el}}^l \right)} + \frac{\sin \left( \phi_{\mathrm{az}}^l \right) \cos \left( \phi_{\mathrm{el}}^l \right) \boldsymbol{\upsilon}_{l,2} }{\pi \cos \left( \phi_{\mathrm{az}}^l \right) \sin^2 \left( \phi_{\mathrm{el}}^l \right)}, \label{eq62}  \numberthis \\
\boldsymbol{\kappa}_{l, 2} = & \frac{\boldsymbol{\upsilon}_{l,2}}{\pi \sin \left( \phi_{\mathrm{el}}^l \right)}, \label{eq63}  \numberthis \\
\boldsymbol{\kappa}_{l, 3} = & \frac{\boldsymbol{\upsilon}_{l, 3}}{\pi \cos \left( \theta_{\mathrm{az}}^l \right) \sin \left( \theta_{\mathrm{el}}^l \right)} + \frac{\sin \left( \theta_{\mathrm{az}}^l \right) \cos \left( \theta_{\mathrm{el}}^l \right) \boldsymbol{\upsilon}_{l, 4} }{\pi \cos \left( \theta_{\mathrm{az}}^l \right) \sin^2 \left( \theta_{\mathrm{el}}^l \right)}, \label{eq64}  \numberthis \\
\boldsymbol{\kappa}_{l, 4} = & \frac{\boldsymbol{\upsilon}_{l, 4} }{\pi \sin \left( \theta_{\mathrm{el}}^l \right)}, %\label{eq65}  \numberthis \\
\boldsymbol{\kappa}_{l, 5} =  \frac{\boldsymbol{\upsilon}_{l, 5}}{2\pi \Delta f}, \label{eq66} \numberthis \\
\boldsymbol{\Upsilon}_n = & \boldsymbol{B}^{\dagger} \breve{\boldsymbol{B}}_n \odot \boldsymbol{\gamma}^{\mathrm{T}}, \label{eq67} \numberthis
\end{align*}
where $\boldsymbol{B}$ is given by \eqref{eq34} and $\breve{\boldsymbol{B}}_n$ is given by \eqref{MeqBC1}.
\end{lemma}
\begin{proof}
The proof can be found in Appendix \ref{ProofLem2}.
\end{proof}

In addition, when the beamspace channel estimation error is complex circular Gaussian, we have the following proposition.
\begin{proposition} \label{Prop2}
With the proposed algorithm, the expectation and covariance of the first-order perturbation of the channel parameter estimation are given as
\begin{align*}
\mathrm{E} \left( \Delta \phi_{\mathrm{az}}^l \right) = &  0, \quad \mathrm{E} \left( \| \Delta \phi_{\mathrm{az}}^l \| ^2 \right) = \frac{N_0 \left\| \boldsymbol{\kappa}_{l, 1} \right\|^2 }{2N_{\mathrm{p}}E_{\mathrm{s}}} \label{eq68}  \numberthis \\
\mathrm{E} \left( \Delta \phi_{\mathrm{el}}^l \right) = &  0, \quad \mathrm{E} \left( \| \Delta \phi_{\mathrm{el}}^l \| ^2 \right) = \frac{N_0 \left\| \boldsymbol{\kappa}_{l, 2} \right\|^2 }{2N_{\mathrm{p}}E_{\mathrm{s}}}, \label{eq69}  \numberthis \\
\mathrm{E} \left( \Delta \theta_{\mathrm{az}}^l \right) = &  0, \quad \mathrm{E} \left( \| \Delta \theta_{\mathrm{az}}^l \| ^2 \right) = \frac{N_0 \left\| \boldsymbol{\kappa}_{l, 3} \right\|^2 }{2N_{\mathrm{p}}E_{\mathrm{s}}}, \label{eq70}  \numberthis \\
\mathrm{E} \left( \Delta \theta_{\mathrm{el}}^l \right) = &  0, \quad \mathrm{E} \left( \| \Delta \theta_{\mathrm{el}}^l \| ^2 \right) = \frac{N_0 \left\| \boldsymbol{\kappa}_{l, 4} \right\|^2 }{2N_{\mathrm{p}}E_{\mathrm{s}}}, \label{eq71}  \numberthis \\
\mathrm{E} \left( \Delta \tau_l \right) = &  0, \quad \mathrm{E} \left( \|\Delta \tau_l \| ^2 \right) = \frac{N_0 \left\| \boldsymbol{\kappa}_{l, 5} \right\|^2 }{2N_{\mathrm{p}}E_{\mathrm{s}}}, \label{eq72} \numberthis \\
\mathrm{E} \left( \Delta \gamma_l \right) = &  0, \quad \mathrm{E} \left( \| \Delta \gamma_l \| ^2 \right) = \frac{N_0 \left\| \boldsymbol{\Pi}_l \right\|_{\mathrm{F}}^2 }{2N_{\mathrm{p}}E_{\mathrm{s}}}, \label{eq73} \numberthis
\end{align*}
where $\boldsymbol{\kappa}_{l, i}$, $i=1, 2, 3, 4, 5$, is given in Lemma \ref{lem2}, and $\left\| \cdot \right\|_{\mathrm{F}}$ is the Frobenius norm. $\boldsymbol{\Pi}_l$ is given by
\begin{align*}
\boldsymbol{\Pi}_l = & \left[ 
\begin{array}{cc}
    \Re \left\{ \boldsymbol{b}_l^{\mathrm{T}} \boldsymbol{B}^{\dagger} \right\} & - \Im \left\{ \boldsymbol{b}_l^{\mathrm{T}} \boldsymbol{B}^{\dagger} \right\} \\
    \Im \left\{ \boldsymbol{b}_l^{\mathrm{T}} \boldsymbol{B}^{\dagger} \right\} & \Re \left\{ \boldsymbol{b}_l^{\mathrm{T}} \boldsymbol{B}^{\dagger} \right\}
\end{array}
\right] \\
& - \sum_{n=1}^5 \left[ \begin{array}{c}
    \Re \left\{ \boldsymbol{b}_l^{\mathrm{T}} \boldsymbol{\Upsilon}_n \right\} \\
    \Im \left\{ \boldsymbol{b}_l^{\mathrm{T}} \boldsymbol{\Upsilon}_n \right\}
\end{array}
\right] \left[ \begin{array}{cc}
    \Im \left\{ \boldsymbol{V}_n^{\mathrm{H}} \right\} &
    \Re \left\{ \boldsymbol{V}_n^{\mathrm{H}} \right\}
\end{array}
\right]. \label{eq74} \numberthis
\end{align*}
\end{proposition}
\begin{proof}
The results are straightforward based on Lemma \ref{lem2}.
\end{proof}
% With estimated channel parameters, we reconstruct the channel vector $\tilde{\boldsymbol{h}}$, and the first-order perturbation is given as follows.
%  \begin{lemma}\label{lem3}
% The first-order perturbation of the channel estimate, $\Delta \boldsymbol{h} = \tilde{\boldsymbol{h}} - \boldsymbol{h}$, is given by
% \end{lemma}
% \begin{proof}
% The proof can be found in Appendix \ref{ProofLem3}.
% \end{proof}

\subsection{First-order Perturbation of Position Estimation}
With the estimates of the channel parameters, we perform the positioning with the method in Section \ref{secIVD}. Given the first-order perturbation of the channel parameters, we obtain the first-order perturbation of the position estimation as follows.
\begin{lemma}\label{lem3}
The first-order perturbation of the position estimate, $\Delta \boldsymbol{p}_{\mathrm{R}} = \hat{\boldsymbol{p}}_{\mathrm{R}} - \boldsymbol{p}_{\mathrm{R}}$, is given by
\begin{align*}
\Delta \boldsymbol{p}_{\mathrm{R}} = & \Im \left\{ \boldsymbol{\Psi} \Delta \boldsymbol{h} \right\}, \label{eq75} \numberthis
\end{align*}
where $\boldsymbol{\Psi} \in \mathbb{C}^{3\times N_1N_2N_3N_4M_5}$ is given by
\begin{align*}
\boldsymbol{\Psi} = \sum_{l=1}^L & \breve{\boldsymbol{D}}_l \left[ \boldsymbol{\kappa}_{l, 3}, \boldsymbol{\kappa}_{l, 4}, \boldsymbol{\kappa}_{l, 5} \right]^{\mathrm{H}} \\
& + \breve{\boldsymbol{E}}_l \left[
\boldsymbol{\kappa}_{l, 1}, \boldsymbol{\kappa}_{l, 2}, \boldsymbol{\kappa}_{l, 3}, \boldsymbol{\kappa}_{l, 4}, \boldsymbol{\kappa}_{l, 5} \right]^{\mathrm{H}}, \label{eq76} \numberthis
\end{align*}
where $\boldsymbol{\kappa}_{l, i}$, $i=1, 2, 3, 4, 5$, are given in Lemma \ref{lem2}. $\breve{\boldsymbol{D}}_l \in \mathbb{R}^{3\times 3}$ and $\breve{\boldsymbol{E}}_l\in \mathbb{R}^{3\times 5}$ are given by \eqref{MeqC4} and \eqref{MeqC5}, respectively.
\end{lemma}
\begin{proof}
The proof can be found in Appendix \ref{ProofLem3}.
\end{proof}
When the beamspace channel estimation error is complex circular Gaussian, we have the following proposition.
\begin{proposition} \label{Prop3}
With the proposed positioning method, the expectation and covariance of the first-order perturbation of the position estimate are given by
\begin{align*}
\mathrm{E} \left( \Delta \boldsymbol{p}_{\mathrm{R}} \right) = & \boldsymbol{0}, \quad \mathrm{E} \left( \| \Delta \boldsymbol{p}_{\mathrm{R}} \| ^2 \right) = \frac{N_0 \left\| \boldsymbol{\Psi} \right\|_{\mathrm{F}}^2 }{2N_{\mathrm{p}}E_{\mathrm{s}}}. \label{eq81}  \numberthis
\end{align*}
\end{proposition}
\begin{proof}
The results are straightforward based on Lemma \ref{lem3}.
\end{proof}
%------------------------------------------------------------------------------------------------------------------------
%------------------------------------------------------------------------------------------------------------------------
%------------------------------------------------------------------------------------------------------------------------
\section{Performance Evaluation and Discussions}
%\textcolor{red}{Expected Results:
%\begin{itemize}
%    \item Validation of channel and position estimation performance;
%    \item Channel and position estimation performance versus the range of user location uncertainty;
%    \item Channel and position estimation performance versus the range of scatter location uncertainty;
%    \item Channel and position estimation performance versus the number of pilots;
%    \item SLAC performance evaluation.
%\end{itemize}
%}

%----------------------------------------- Fig. 2: Set up ---------------------------------------------------------------
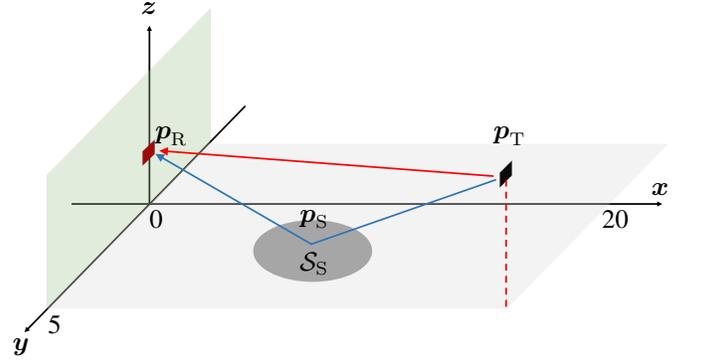
\begin{figure}
\begin{minipage}[!htb]{0.5\textwidth}
    \centering
    \input{Fig/3DSetup.tex}
    \caption{Simulation setup: 1 LOS and 1 NLOS paths, scatter point $\boldsymbol{p}_{\mathrm{S}}$ located at the the gray region $\mathcal{S}_{\mathrm{S}}$. The base station and vehicular user are equipped with URAs in $y-z$ plane.}
    \label{fig:3DSetup}
\end{minipage}
\end{figure}
%------------------------------------------------------------------------------------------------------------------------

\setcounter{table}{0}
%----------------------------------------- TABLE I: Simulation Parameters -----------------------------------------------
\begin{table}%[!htb]
\caption{Parameters in Simulations}
\centering
    \begin{tabular}{l|l}
    \hline \hline
    Parameters & Values \\
    \hline
    Carrier frequency $f_\mathrm{c}$ & 30 GHz \\
    \hline
    Bandwidth & 400 MHz \\
    \hline
    Subcarrier spacing $\Delta f$ & 120 kHz \\
    \hline 
    Number of subcarriers $M_5$ & 500 \\
    \hline 
    Wavelength $\lambda$ & 1 cm \\
    \hline
    Size of URA at transmitter $M_1 \times M_2$ & $8 \times 8$ \\
    \hline 
    Size of URA at receiver $M_3 \times M_4$ & $8 \times 8$ \\
    \hline 
    Antenna spacing $\Delta d$ & 0.5 cm\\
    \hline 
    Number of RF chains at transmitter $N_{\mathrm{T}}$ & 16 \\
    \hline 
    Number of RF chains at receiver $N_{\mathrm{R}}$ & 16 \\
    \hline
    Number of transmit beams $N_1 \times N_2$ & $4 \times 4$ \\
    \hline 
    Number of receive beams $N_3 \times N_4$ & $4\times 4$ \\
    \hline
    Base station location $\boldsymbol{p}_{\mathrm{T}}$ & $[20, 5, 8]^{\mathrm{T}}$ \\
    \hline
    %Scatter area $\mathcal{S}_{\mathrm{S}}$ & $r_{\mathrm{S}} = 2$ \\
    %\hline 
    %User location area $\mathcal{S}_{\mathrm{R}}$ & $r_{\mathrm{S}} = 1$ \\
    %\hline
    Number of pilots $N_{\mathrm{P}}$ & 32\\
    \hline
    Coherence time & 5 ms \\
    \hline \hline
    \end{tabular}
    \label{Tab1}
\end{table} 
%------------------------------------------------------------------------------------------------------------------------
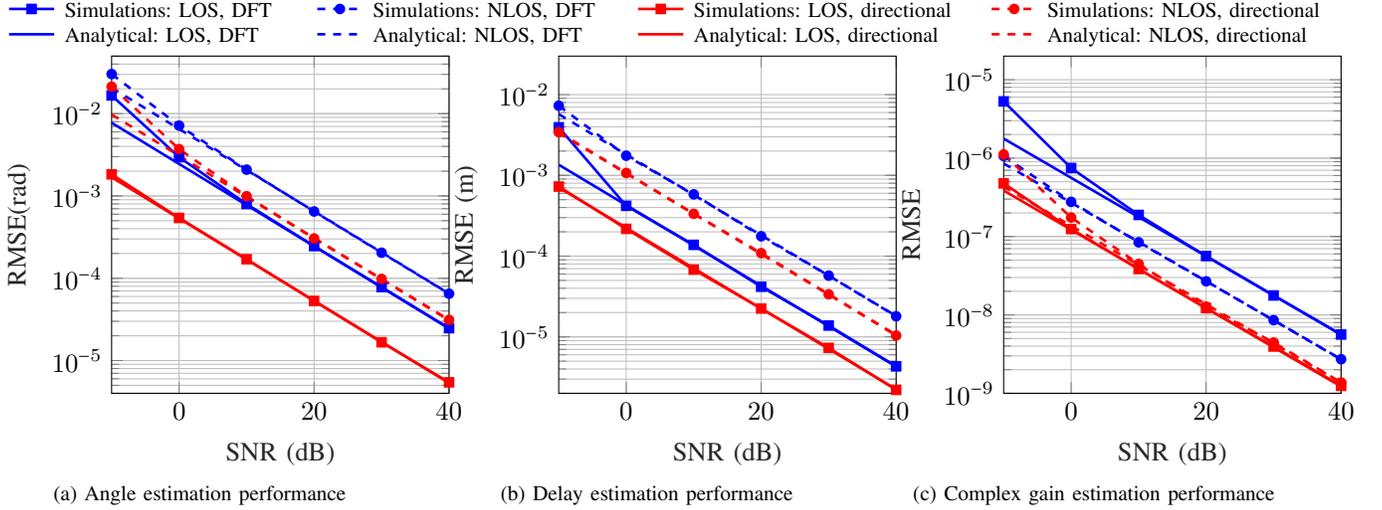
\begin{figure*}
\begin{minipage}[!htb]{0.95\textwidth}
\begin{subfigure}[t]{\textwidth}
    \centering 
    \input{Fig/LegendFig3} \vspace{-5mm}
    \centering
\end{subfigure}
\begin{subfigure}[b]{0.31\textwidth}
    \centering
    \input{Fig/Angles} \vspace{-5mm}
    \centering
    \caption{Angle estimation performance}
    \label{fig:Angles}
\end{subfigure}
\hfill
\begin{subfigure}[b]{0.31\textwidth}
    \centering
    \input{Fig/Delay} \vspace{-5mm}
    \centering
    \caption{Delay estimation performance}
    \label{fig:Delay}
\end{subfigure}
\hfill
\begin{subfigure}[b]{0.31\textwidth}
    \centering
    \input{Fig/CG} \vspace{-5mm}
    \centering
    \caption{Complex gain estimation performance}
    \label{fig:CG}
\end{subfigure}
\end{minipage}
\caption{Analytical performance validation: 1 LOS path and 1 NLOS path, $N_{\mathrm{P}} = 32$, DFT beams and directional beams.} 
\label{fig:PerValidation}
\end{figure*}
%------------------------------------------------------------------------------------------------------------------------
\subsection{Simulation Setup}
In the simulations, we consider the simultaneous localization and communications for a vehicular user, 
with  simulation parameters in Table \ref{Tab1}. %A single base station is located at $\boldsymbol{p}_{\mathrm{T}} = [20, 5, 8]^{\mathrm{T}}$. A multipath propagation environment is considered. 
Specifically, we consider one line-of-sight (LOS) and NLOS path from the ground as shown in Fig. \ref{fig:3DSetup}. We denote $\mathcal{S}_{\mathrm{S}} \triangleq \{(x, y, z) \in \mathbb{R}^3 | z=0, (x-10)^2 + (y-2.5)^2 \le r_{\mathrm{S}}^2 \}$ as the area where the scatter point is located.
The coherence time is set to 5 ms, indicating 600 coherent OFDM blocks. Among them, the first $N_{\mathrm{P}}$ OFDM blocks are used for pilot transmission, while the remaining blocks are used for data transmission.
The overall bandwidth of the system is 400 MHz, using 3300 subcarriers starting from 28 GHz. Among them, we use the every 6-th subcarrier for pilot transmission. As examples to validate the performance of the proposed method, we evaluate the estimation performance with both the DFT beams (with beam separation of $\pi/4$ in each dimension) and the directional beams with a closer beam separation of $\pi/8$ at both transmitter and receiver sides. The estimation of the clock bias and orientation of the vehicular user is out of the scope in this work that we assume perfect synchronization and known orientation at receiver in the framework. All simulations are performed in MATLAB on a laptop with a 2.3 GHz Quad-Core Intel Core i7 processor and 16 GB memory.

The geometric relationship between the positions and the angles/delays are the same as that in \cite{GeWenKimZhuJiaKimSveWym20, WenKulWitWym19}. The path loss $\rho_l$ for the $l$-th path follows the models in \cite{GeWenKimZhuJiaKimSveWym20, WenKulWitWym19}, given by
\begin{align*}
\rho_l \propto & \varsigma_l \left( \frac{\lambda }{ 4\pi d_l}\right)^2, \numberthis
\end{align*}
where $d_l$ is the propagation distance, and $\varsigma_l$ indicates the power of the $l$-th multipath component after scattering. The SNR at receiver side before signal combination is given by
\begin{align*}
\mathrm{SNR} = & \frac{ \sum_{m_5 = 1}^{M_5} \left\| \boldsymbol{H}_{m_5} \boldsymbol{F} \right\|_{\mathrm{F}}^2 E_{\mathrm{s}}}{\sum_{m_5 = 1}^{M_5} \mathrm{E}\left( \left\| \boldsymbol{Z}_{m_5} \right\|^2 \right) }, \numberthis
\end{align*}
where $\boldsymbol{H}_{m_5} \in \mathbb{C}^{M_3M_4 \times M_1M_2}$ is the channel matrix over the $m_5$-th subcarrier. 

For evaluation, we use the root mean-square error metric for both channel parameter and position estimation and compare to the closed-form first-order perturbation analytical results. In addition, we also compare the estimation and complexity performance of the proposed scheme to the existing tensor decomposition method in \cite{WenGarKulWitWym18, WenSoWym20}.

%------------------------------------------------------------------------------------------------------------------------
%\begin{figure*}
%\begin{minipage}[!htb]{0.95\textwidth}
%\begin{subfigure}[b]{0.3\textwidth}
%    \centering
%    \input{Fig/AODAz_Sum} \vspace{-5mm}
%    \centering
%    \caption{AOD azimuth angle estimation}
%    \label{fig:AODAz_Sum}
%\end{subfigure}
%\hfill
%\begin{subfigure}[b]{0.3\textwidth}
%    \centering
%    \input{Fig/AODEl_Sum} \vspace{-5mm}
%    \centering
%    \caption{AOD elevation angle estimation}
%    \label{fig:AODEl_Sum}
%\end{subfigure}
%\hfill
%\begin{subfigure}[b]{0.3\textwidth}
%    \centering
%    \input{Fig/AOAAz_Sum} \vspace{-5mm}
%    \centering
%    \caption{AOA azimuth angle estimation}
%    \label{fig:AOAAz_2Path}
%\end{subfigure}
%\end{minipage}
%\caption{Channel parameter estimation with sum beams: 1 LOS and 1 NLOS, $N_{\mathrm{P}} = 32$.} 
%\label{fig:ChPara_Sum}
%\end{figure*}

\subsection{Channel Parameter Estimation Performance}
We first evaluate the analytical perturbation performance and the numerical performance of the proposed beamspace ESPRIT approach with the system configuration given in Table \ref{Tab1}. In Fig. \ref{fig:PerValidation}, we present the channel parameter estimation performance with both the DFT beams and the directional beams in each mode at transmitter and receiver sides.\footnote{The transformation matrices are selected to cover the area of interest in the setup. In this case, the selections are made with the prior user and scatter location information.} The following observations can be seen from Fig. \ref{fig:PerValidation}.
%------------------------------------------------------------------------------------------------------------------------
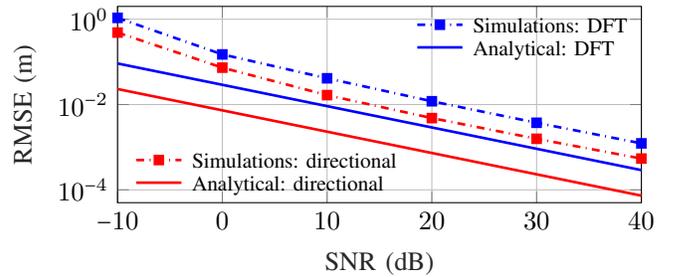
\begin{figure}
\begin{minipage}[!htb]{0.48\textwidth}
    \centering
    \input{Fig/Localization}
    % \caption{Sum beams}
    \label{fig:Pos_Sum}
\caption{Localization performance: DFT beams and directional beams, 1 LOS and 1 NLOS, $N_{\mathrm{P}} = 32$.}
\label{fig:Pos}
\end{minipage}
\end{figure}
%------------------------------------------------------------------------------------------------------------------------
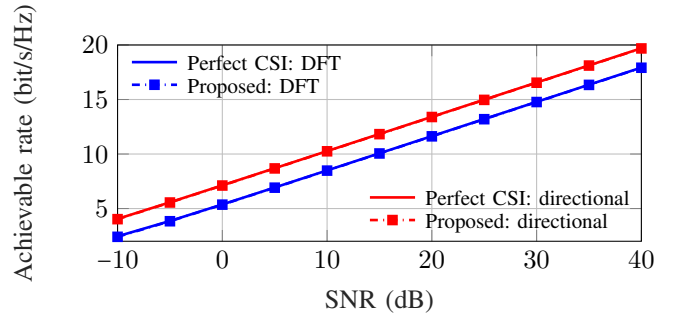
\begin{figure}
\begin{minipage}[!htb]{0.48\textwidth}
\centering
\input{Fig/SumRate}
\caption{Achievable rate performance: 1 LOS and 1 NLOS, $N_{\mathrm{P}} = 32$, DFT beams and directional beams.}
\label{fig:SumRate}
\end{minipage}
\end{figure}
\begin{itemize}
\item The analytical and numerical results for channel parameter estimation show to be well-matched, particularly in high SNR regions. This is consistent with the first-order perturbation analysis, since higher-order perturbations can be reasonably ignored when SNR is high.
\item The angle and delay estimations of the LOS path are better than that of the NLOS components. This is reasonable as the LOS path has relatively higher SNR.
\item The channel parameter estimation performance is significantly improved if the directional beams with closer beam separation is used. Particularly, if the directional beams with closer separation are used, the effective SNR can be boosted and the angle estimation performance is significantly improved, leading to the improvement of other channel parameter estimation performance.
\end{itemize} 
The first-order perturbation of angular frequency estimation has been widely studied in the literature, for example, \cite{RoyKai89, SahUseCom17, LiuLiu06, LiuLiuMa07, SorLat16, Haardt2018}; however, it is worth pointing out that the first-order perturbation of the channel parameter estimation, especially the complex gain estimation, is rarely studied. Our closed-form first-order perturbation of the channel parameter estimation matches very well with the simulations.

% \begin{figure}
% \begin{minipage}[!htb]{0.48\textwidth}
% \begin{subfigure}[b]{\textwidth}
%    \centering
%    \input{Fig/Localization_2Path}
%    \caption{DFT beams}
%    \label{fig:Pos_DFT}
% \end{subfigure}
% \begin{subfigure}[b]{\textwidth}
%    \centering
%    \input{Fig/Localization_Sum}
%    \caption{Sum beams}
%    \label{fig:Pos_Sum}
% \end{subfigure}
% \caption{Localization performance: 1 LOS and 1 NLOS, 12 MHz and 60 MHz, $N_{\mathrm{P}} = 32$.}
% \label{fig:Pos}
% \end{minipage}
% \end{figure}
%------------------------------------------------------------------------------------------------------------------------
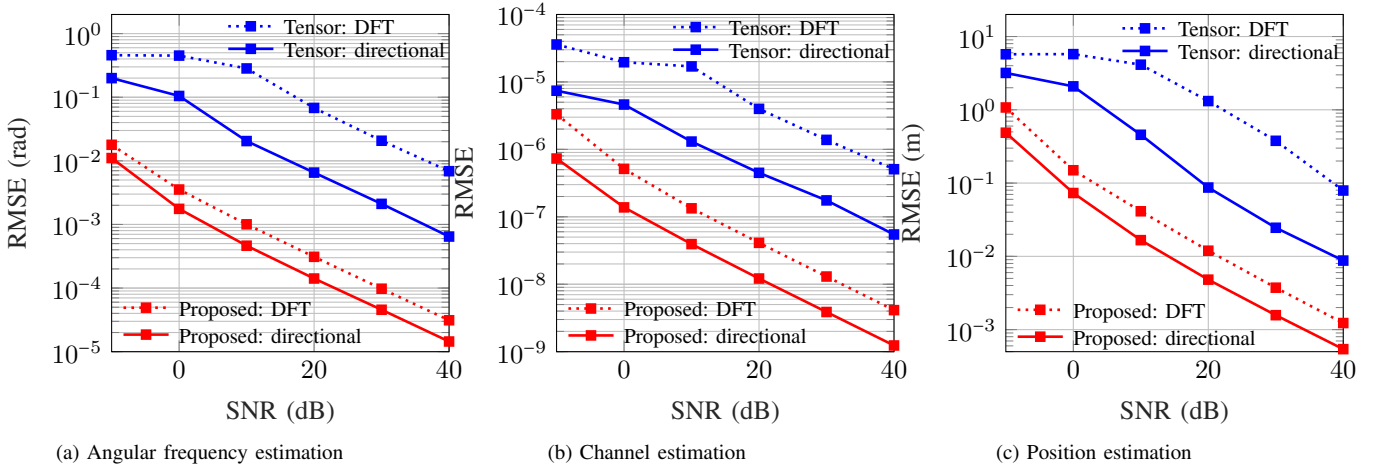
\begin{figure*}
\begin{minipage}[!htb]{0.95\textwidth}
\begin{subfigure}[b]{0.31\textwidth}
    \centering
    \input{Fig/AFCMP} \vspace{-5mm}
    \centering
    \caption{Angular frequency estimation}
    \label{fig:AFCMP}
\end{subfigure}
\hfill
\begin{subfigure}[b]{0.31\textwidth}
    \centering
    \input{Fig/CHCMP} \vspace{-5mm}
    \centering
    \caption{Channel estimation}
    \label{fig:CHCMP}
\end{subfigure}
\hfill
\begin{subfigure}[b]{0.31\textwidth}
    \centering
    \input{Fig/PosCMP} \vspace{-5mm}
    \centering
    \caption{Position estimation}
    \label{fig:AOAAz_Sum}
\end{subfigure}
\end{minipage}
\caption{Performance comparison: 1 LOS and 1 NLOS, $N_{\mathrm{P}} = 32$, DFT beams and directional beams.} 
\label{fig:PerCMP}
\end{figure*}

\begin{figure}
\begin{minipage}[!htb]{0.48\textwidth}
    \centering
    \input{Fig/RunningTime}
    \caption{Program running time performance.}
    \label{fig:RT}
\end{minipage}
\end{figure}
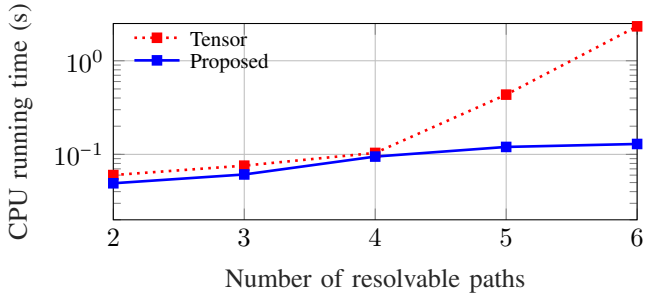
%------------------------------------------------------------------------------------------------------------------------

%------------------------------------------------------------------------------------------------------------------------
\subsection{Simultaneous Localization and Communication}
\subsubsection{Localization Performance}
In Fig. \ref{fig:Pos}, we present the localization performance with the same setup in Fig. \ref{fig:PerValidation}. From Fig. \ref{fig:Pos}, it is clear that the first-order perturbation of the position estimation is accurate (the absolute value of the difference is decreasing with SNR.) in high SNR region since higher orders can be reasonably ignored. Comparatively, the directional beams with closer beam separation can improve position estimation performance, which agrees with the analytical results. This can be attributed to the channel parameter estimation performance improvements as shown in Fig. \ref{fig:PerValidation}.
To better illustrate the gap between the analytical and the numerical results in position estimation performance, we refer to the approximation made in \eqref{MeqC2nd}, where the last three components are the second-order perturbations. Since $\Delta \boldsymbol{\mu}_l$ and $\Delta \boldsymbol{\delta}_l$ are correlated as indicated by \eqref{MeqC1} and \eqref{MeqC2}, the mean and covariance of the second-order perturbations are not zero, leading to the gap shown in Fig. \ref{fig:Pos}. This also indicates that a higher order perturbation analysis (for example, second order) is required for position estimation, which is left for future work.

\subsubsection{Achievable Rate}
Next, we evaluate the achievable rate performance of the system. The  effective achievable rate is computed as in \eqref{eq:achievableRate}. In Fig. \ref{fig:SumRate}, we show the effective achievable rate performance in the following conditions: 1) perfect channel state information (CSI) knowledge, 2) estimated CSI with the proposed method. It is obvious that the proposed beamspace ESPRIT approach achieves near-optimal achievable rate performance. This indicates that the proposed schemes reconstruct the optimal CSI for communications, agreeing with the channel estimation performance validated in Fig. \ref{fig:CHCMP}. In addition, the effective achievable rate performance with directional beams ($\pi/8$ beam separation) outperforms that with the DFT beams ($\pi/4$ beam separation), which is attributed to high effective SNR with the closer beam separation. Together with the localization performance in Fig. \ref{fig:Pos}, we conclude that the optimal spatial signal design for simultaneous localization and communications is required, which can be a future research topic.

%------------------------------------------------------------------------------------------------------------------------
\subsection{Performance Comparison to Existing Work}\label{SubSecVIE}
We now compare the estimation performance of the proposed method to existing work in \cite{WenGarKulWitWym18, WenSoWym20}. It is worth pointing out that the work in \cite{WenGarKulWitWym18, WenSoWym20} explores the ESPRIT approaches based on tensor decomposition, and has been applied to DFT beams and random beams for angular frequency estimations. We extend the tensor-based ESPRIT approaches in \cite{WenGarKulWitWym18, WenSoWym20} to arbitrary directional beams, and use it for channel parameter and position estimation. In the proposed method, we have applied the spatial smoothing over frequency domain, which potentially utilize the frequency domain samples more efficiently. As a result, the delay estimation performance in the proposed method can be significantly improved. Because of the improved delay estimation, the average angular frequency estimation performance over all dimensions is significantly improved, as shown in Fig. \ref{fig:AFCMP}. With the estimated channel parameters, we can also reconstruct the high dimensional channel $\boldsymbol{\mathcal{H}}$. It is clear from Fig. \ref{fig:CHCMP} that the proposed method achieves significantly improved channel estimation performance over the tensor decomposition based methods in \cite{WenGarKulWitWym18, WenSoWym20}. The channel parameter estimations can be used to infer user location, and it is quite obvious that the proposed method show much improved localization performance over the existing work. It is also worth pointing out that the directional beams with closer beam separation can significantly improve the estimation performance of the tensor-based ESPRIT approach. This is because in the tensor-based ESPRIT approach \cite{WenGarKulWitWym18, WenSoWym20}, using closer beam separation can significantly improve the angle estimation for multipath components. On the other hand, in the proposed scheme, we have employed the spatial smoothing techniques, which has improved the delay solution for multipath components. As a result, the performance improvement gap of using the directional beams with closer beam separation over the DFT beams is not as much as the tensor-based methods.

As indicated in Section \ref{secIVF}, the proposed method enjoys low complexity, particularly in multi-path scenarios. This can be well demonstrated with the CPU running time of the program, as shown in Fig. \ref{fig:RT}. We generate multiple scatter points from the scatter region, and evaluate the run time of the proposed method and the tensor-based method in \cite{WenGarKulWitWym18, WenSoWym20}. It is obvious that when the number of multi-path components increases, the tensor-based method requires considerable computation time, which is attributed to the high complexity in CP decomposition algorithm for high order tensors. On the other hand, the proposed method utilizes the FFT/IFFT algorithms to obtain the signal subspace, and the complexity almost grows linearly to the number of multi-path components. The results indicate that the proposed beamspace ESPRIT method has great potential in real-time positioning applications.

\section{Conclusion}
In this paper, we have proposed a novel beamspace multi-dimensional ESPRIT method for simultaneous localization and communications. The proposed method employs spatial smoothing techniques over the frequency domain and can work with arbitrary beamforming and combining matrices used at transmitter and receiver sides. In addition, we have proposed a low-complexity method to obtain the signal subspace by using the properties of Hankel matrices. With the first-order perturbation analysis, we have successfully obtained the closed-form estimation performance of the proposed method. Through numerical simulations, we have shown that the proposed method achieves promising estimation performance with both DFT and directional beams used in the system, and the numerical results match well with the analytical results. Moreover, the directional beams with closer beam separation can further improve the estimation performance. We have also shown that the program running time of the proposed method outperforms existing tensor decomposition based ESPRIT method, especially in reasonable rich multi-path scenarios. Finally, we have investigated the achievable rate performance and shown that the proposed method achieves near-optimal achievable rate performance.

\appendices
\section{Factorization of $\boldsymbol{H}$}\label{AppendixA}
From \eqref{eq10} and \eqref{eq19}, the $\ell_5$-th column of $\boldsymbol{H}$ is given by
\begin{align*}
\boldsymbol{J}_{M_1, M_2, M_3, M_4, \ell_5} \boldsymbol{h}^{\left( \mathrm{b} \right) } = & \left( \boldsymbol{I}_{M_1} \otimes \boldsymbol{I}_{M_2} \otimes \boldsymbol{I}_{M_3} \otimes \boldsymbol{I}_{M_4} \otimes \boldsymbol{J}_{\ell_5} \right) \\
& \times \left( \boldsymbol{B}_1^{\left( N_1 \right)} \odot \boldsymbol{B}_2^{\left( N_2 \right)} \odot \cdots \odot \boldsymbol{B}_5^{\left( N_5 \right)} \right) \boldsymbol{\gamma} \\
= & \left( \boldsymbol{B}_1^{\left( N_1 \right)} \odot \boldsymbol{B}_2^{\left( N_2 \right)} \odot \cdots \odot \boldsymbol{A}_5^{\left( K_5 \right)} \right) \\
& \times \mathrm{Diag}\left( \boldsymbol{\omega}_5^{\ell_5 - 1} \right) \boldsymbol{\gamma},
\end{align*}
where 
%\begin{align*}
$\boldsymbol{\omega}_5^{\ell} = \left[ \exp \left( \jmath \ell \omega_{1, 5} \right), \exp \left( \jmath \ell \omega_{2, 5} \right), \cdots, \exp \left( \jmath \ell \omega_{L, 5} \right) \right]^{\mathrm{T}}$.  
%\end{align*}
Therefore, we have
\begin{align*}
\boldsymbol{H} = & \left[ \boldsymbol{P} \mathrm{Diag}\left( \boldsymbol{\omega}_5^0 \right) \boldsymbol{\gamma}, \boldsymbol{P} \mathrm{Diag}\left( \boldsymbol{\omega}_5^1 \right) \boldsymbol{\gamma}, \cdots, \boldsymbol{P} \mathrm{Diag}\left( \boldsymbol{\omega}_5^{L_5 - 1} \right) \boldsymbol{\gamma}
\right] \\
= & \boldsymbol{P} \mathrm{Diag}\left( \boldsymbol{\gamma} \right) \left[\boldsymbol{\omega}_5^{0}, \boldsymbol{\omega}_5^{1}, \cdots, \boldsymbol{\omega}_5^{L_5 - 1} \right] \\
= & \boldsymbol{P} \mathrm{Diag}\left( \boldsymbol{\gamma} \right) \boldsymbol{G},
\end{align*}
where $\boldsymbol{P}$ and $\boldsymbol{G}$ are given in \eqref{eq21} and \eqref{eq22}, respectively.

\section{Proofs} \label{AppendixB}
%-----------------------------------------------------------------------------
\subsection{Proof of Proposition \ref{propA1}} \label{ProofpropA1}
For $n= 1, 2, 3, 4$, using the properties of the Khatri-Rao products \cite{LiuLiu06}, we have
\begin{align*}
\breve{\boldsymbol{J}}_{n, 1}\boldsymbol{P} \boldsymbol{\Phi}_n & = \left( \boldsymbol{I}_{N_1} \otimes \cdots \otimes \boldsymbol{L}_{n, 1} \otimes \cdots \otimes \boldsymbol{I}_{N_4} \otimes \boldsymbol{I}_{K_5} \right) \\
& \hspace{-5mm}\times \left( \boldsymbol{B}_1^{\left( N_1 \right)} \odot \cdots \odot \boldsymbol{B}_n^{\left( N_n \right)} \odot \cdots \odot \boldsymbol{B}_4^{\left( N_4 \right)} \odot \boldsymbol{A}_5^{\left( K_5 \right)} \right) \boldsymbol{\Phi}_n \\
& \hspace{-5mm} = \boldsymbol{B}_1^{\left( N_1 \right)} \odot \cdots \odot \boldsymbol{L}_{n, 1} \boldsymbol{B}_n^{\left( N_n \right)}\boldsymbol{\Phi}_n \odot \cdots \odot \boldsymbol{B}_4^{\left( N_4 \right)} \odot \boldsymbol{A}_5^{\left( K_5 \right)} \\
& \hspace{-5mm} = \boldsymbol{B}_1^{\left( N_1 \right)} \odot \cdots \odot \boldsymbol{L}_{n, 2} \boldsymbol{B}_n^{\left( N_n \right)} \odot \cdots \odot \boldsymbol{B}_4^{\left( N_4 \right)} \odot \boldsymbol{A}_5^{\left( K_5 \right)} \\
& \hspace{-5mm} = \breve{\boldsymbol{J}}_{n, 2}\boldsymbol{P},
\end{align*}
where the results $\boldsymbol{L}_{n, 1} \boldsymbol{B}_n^{\left( N_n \right)}\boldsymbol{\Phi}_n = \boldsymbol{L}_{n, 2} \boldsymbol{B}_n^{\left( N_n \right)}$ in Proposition \ref{Th1} is used in the above derivation. Following the similar process, we also have $\breve{\boldsymbol{J}}_{5, 1}\boldsymbol{P} \boldsymbol{\Phi}_5 = \breve{\boldsymbol{J}}_{5, 2}\boldsymbol{P}$.

%-----------------------------------------------------------------------------
\subsection{Proof of Lemma \ref{lem1}} \label{ProofLem1}
To begin with, we introduce the following results for the perturbation of signal subspace $\boldsymbol{U}_{\mathrm{s}}$.
\begin{lemma} \label{MlemB1}
Given the perturbed $\tilde{\boldsymbol{H}}$ in \eqref{eq48}, and the subspace decomposition as
\begin{align*}
\tilde{\boldsymbol{H}} = \tilde{\boldsymbol{U}}_{\mathrm{s}} \tilde{\boldsymbol{\Sigma}}_{\mathrm{s}} \tilde{\boldsymbol{V}}_{\mathrm{s}}^{\mathrm{H}} + \tilde{\boldsymbol{U}}_{\mathrm{n}} \tilde{\boldsymbol{\Sigma}}_{\mathrm{n}} \tilde{\boldsymbol{V}}_{\mathrm{n}}^{\mathrm{H}},
\end{align*}
the first-order approximation of the $\Delta \boldsymbol{U}_{\mathrm{s}} = \tilde{\boldsymbol{U}}_{\mathrm{s}} - \boldsymbol{U}_{\mathrm{s}}$ is given by
\begin{align*}
\Delta \boldsymbol{U}_{\mathrm{s}} = \boldsymbol{U}_{\mathrm{n}}\boldsymbol{U}_{\mathrm{n}}^{\mathrm{H}} \Delta \boldsymbol{H} \boldsymbol{V}_{\mathrm{s}} \boldsymbol{\Sigma}_{\mathrm{s}}^{-1} + \boldsymbol{U}_{\mathrm{s}} \boldsymbol{R}, \label{MeqB1} \numberthis
\end{align*}
where $\boldsymbol{R}$ is an anti-Hermitian matrix that depends on $\Delta \boldsymbol{H}$.
\end{lemma}
\begin{proof}
The proof can be seen in \cite{SahUseCom17}.
\end{proof}
With $\Delta \boldsymbol{U}_{\mathrm{s}}$, the first-order approximation of $\Delta \boldsymbol{\Gamma}_n = \tilde{\boldsymbol{\Gamma}}_n - \boldsymbol{\Gamma}_n$ can be obtained as follows. 
\begin{lemma} \label{MlemB2}
$\Delta \boldsymbol{\Gamma}_n$ is given by
\begin{align*}
\Delta \boldsymbol{\Gamma}_n = \left( \breve{\boldsymbol{J}}_{n, 1}\boldsymbol{U}_{\mathrm{s}} \right)^{\dagger} \left(\breve{\boldsymbol{J}}_{n, 2} \Delta \boldsymbol{U}_{\mathrm{s}} - \breve{\boldsymbol{J}}_{n, 1} \Delta \boldsymbol{U}_{\mathrm{s}} \boldsymbol{\Gamma}_n \right). \label{MeqB2} \numberthis
\end{align*}
\end{lemma}
\begin{proof}
Using the fact that $\Delta \left( \boldsymbol{A}^{\dagger} \right)=-\boldsymbol{A}^{\dagger} \Delta\boldsymbol{A} \boldsymbol{A}^{\dagger} +  \left( \boldsymbol{A}^{\mathrm{H}} \boldsymbol{A} \right)^{-1} \Delta\boldsymbol{A}^{\mathrm{H}} \left( \boldsymbol{I} - \boldsymbol{A} \boldsymbol{A}^{\dagger} \right)$, we find  
from \eqref{eq28}  and using $( \boldsymbol{I} - \breve{\boldsymbol{J}}_{n, 1} \boldsymbol{U}_{\mathrm{s}} ( \breve{\boldsymbol{J}}_{n, 1} \boldsymbol{U}_{\mathrm{s}} ) ^{\dagger} ) \breve{\boldsymbol{J}}_{n, 2} \boldsymbol{U}_{\mathrm{s}} = \boldsymbol{0}$, we obtain
\begin{align*}
\Delta \boldsymbol{\Gamma}_n = & \Delta \left( \breve{\boldsymbol{J}}_{n, 1} \boldsymbol{U}_{\mathrm{s}} \right)^{\dagger} \breve{\boldsymbol{J}}_{n, 2} \boldsymbol{U}_{\mathrm{s}} + \left( \breve{\boldsymbol{J}}_{n, 1} \boldsymbol{U}_{\mathrm{s}} \right)^{\dagger} \breve{\boldsymbol{J}}_{n, 2} \Delta \boldsymbol{U}_{\mathrm{s}} 
\end{align*}
from which \eqref{MeqB2} follows. 
%\begin{align*}
%\left( \boldsymbol{I} - \breve{\boldsymbol{J}}_{n, 1} \boldsymbol{U}_{\mathrm{s}} \left( \breve{\boldsymbol{J}}_{n, 1} \boldsymbol{U}_{\mathrm{s}} \right) ^{\dagger} \right) \breve{\boldsymbol{J}}_{n, 2} \boldsymbol{U}_{\mathrm{s}} = \boldsymbol{0}.
%\end{align*}
\end{proof}
Denote $\boldsymbol{E} = [\boldsymbol{e}_1, \boldsymbol{e}_2, \cdots, \boldsymbol{e}_L]$ and $\boldsymbol{E}^{-1} = [\boldsymbol{\epsilon}_1, \boldsymbol{\epsilon}_2, \cdots, \boldsymbol{\epsilon}_L]^{\mathrm{T}}$. Then the first-order approximation of $\Delta \Phi_{l, n} = \tilde{\Phi}_{l, n} - {\Phi}_{l, n}$ is given as follows.
\begin{lemma} \label{MlemB3}
$\Delta \Phi_{l, n}$ is given by
\begin{align*}
\Delta \Phi_{l, n} = \boldsymbol{\epsilon}_l^{\mathrm{T}} \Delta \boldsymbol{\Gamma}_n \boldsymbol{e}_l. \label{MeqB3} \numberthis
\end{align*}
\end{lemma}
\begin{proof}
From \eqref{eq32}, we have
\begin{align*}
\Delta \boldsymbol{\Phi}_n = & \boldsymbol{E}^{-1} \Delta \boldsymbol{\Gamma}_n \boldsymbol{E} + \boldsymbol{E}^{-1} \boldsymbol{\Gamma}_n \Delta \boldsymbol{E} - \boldsymbol{E}^{-1} \Delta \boldsymbol{E} \boldsymbol{E}^{-1} \boldsymbol{\Gamma}_n \boldsymbol{E} \\
= & \boldsymbol{E}^{-1} \Delta \boldsymbol{\Gamma}_n \boldsymbol{E} + \boldsymbol{\Phi}_n \boldsymbol{E}^{-1} \Delta \boldsymbol{E} - \boldsymbol{E}^{-1} \Delta \boldsymbol{E} \boldsymbol{\Phi}_n. \numberthis
\end{align*} 
We extract the diagonal element from the right hand side, and the last two terms are canceled, leading to the results in \eqref{MeqB3}.
\end{proof}

In order to prove Lemma \ref{lem1}, we apply Lemma \ref{MlemB1} to Lemmas \ref{MlemB2}--\ref{MlemB3}. We define $\boldsymbol{C} = \Delta \boldsymbol{H} \boldsymbol{V}_{\mathrm{s}} \boldsymbol{\Sigma}_{\mathrm{s}}^{-1}$, which allows us to write: %Using the results in \eqref{MeqB1}, we then have $\Delta \boldsymbol{\Gamma}_n$ as
\begin{align}
\Delta \boldsymbol{\Gamma}_n% = & \left( \breve{\boldsymbol{J}}_{n, 1} \boldsymbol{U}_{\mathrm{s}} \right)^{\dagger} \big(\breve{\boldsymbol{J}}_{n, 2} \left( \boldsymbol{I} - \boldsymbol{U}_{\mathrm{s}}\boldsymbol{U}_{\mathrm{s}}^{\mathrm{H}} \right) \boldsymbol{C} - \breve{\boldsymbol{J}}_{n, 1} \left( \boldsymbol{I} - \boldsymbol{U}_{\mathrm{s}} \boldsymbol{U}_{\mathrm{s}}^{\mathrm{H}} \right) \boldsymbol{C} \boldsymbol{\Gamma}_n \big) \\
% & + \left( \breve{\boldsymbol{J}}_{n, 1} \boldsymbol{U}_{\mathrm{s}} \right)^{\dagger} \left( \breve{\boldsymbol{J}}_{n, 2} \boldsymbol{U}_{\mathrm{s}} \boldsymbol{R} - \breve{\boldsymbol{J}}_{n, 1} \boldsymbol{U}_{\mathrm{s}} \boldsymbol{R} \boldsymbol{\Gamma}_n \right) \\
= & \left( \breve{\boldsymbol{J}}_{n, 1} \boldsymbol{U}_{\mathrm{s}} \right)^{\dagger} \left( \breve{\boldsymbol{J}}_{n, 2} \boldsymbol{C} - \breve{\boldsymbol{J}}_{n, 1} \boldsymbol{C} \boldsymbol{\Gamma}_n \right) - \boldsymbol{\Gamma}_n \boldsymbol{U}_{\mathrm{s}}^{\mathrm{H}} \boldsymbol{C} \\
& + \boldsymbol{U}_{\mathrm{s}}^{\mathrm{H}} \boldsymbol{C} \boldsymbol{\Gamma}_n + \boldsymbol{\Gamma}_n \boldsymbol{R} - \boldsymbol{R} \boldsymbol{\Gamma}_n.\notag 
\\
\Delta \Phi_{l, n} =& \boldsymbol{\epsilon}_l^{\mathrm{T}} \left( \breve{\boldsymbol{J}}_{n, 1} \boldsymbol{U}_{\mathrm{s}} \right)^{\dagger} \left( \breve{\boldsymbol{J}}_{n, 2} - \Phi_{l, n} \breve{\boldsymbol{J}}_{n, 1} \right) \boldsymbol{C} \boldsymbol{e}_l,
\end{align} 
%With the result in \eqref{MeqB3}, $\Delta \Phi_{l, n}$ is further obtained as
%\begin{align*}
%\Delta \Phi_{l, n} = \boldsymbol{\epsilon}_l^{\mathrm{T}} \left( \breve{\boldsymbol{J}}_{n, 1} \boldsymbol{U}_{\mathrm{s}} \right)^{\dagger} \left( \breve{\boldsymbol{J}}_{n, 2} - \Phi_{l, n} \breve{\boldsymbol{J}}_{n, 1} \right) \boldsymbol{C} \boldsymbol{e}_l,
%\numberthis
%\end{align*} 
where $\boldsymbol{\Gamma}_n \boldsymbol{e}_l = \Phi_{l, n} \boldsymbol{e}_l$ and $\boldsymbol{\epsilon}_l \boldsymbol{\Gamma}_n  = \Phi_{l, n} \boldsymbol{\epsilon}_l$ are used in the derivation. Finally, from \eqref{eq20} and \eqref{eq23}, we have $\boldsymbol{U}_{\mathrm{s}} = \boldsymbol{P} \boldsymbol{E}^{-1}$, and $\boldsymbol{\Sigma}_{\mathrm{s}} \boldsymbol{V}_{\mathrm{s}}^{\mathrm{H}} = \boldsymbol{E} \mathrm{Diag} \left( \boldsymbol{\gamma} \right) \boldsymbol{G}^{\mathrm{T}} $ leading to 
\begin{align*}
\Delta \Phi_{l, n} = & \boldsymbol{\epsilon}_l^{\mathrm{T}} \left( \breve{\boldsymbol{J}}_{n, 1} \boldsymbol{P} \boldsymbol{E}^{-1} \right)^{\dagger} \left( \breve{\boldsymbol{J}}_{n, 2} - \Phi_{l, n} \breve{\boldsymbol{J}}_{n, 1} \right) \Delta \boldsymbol{H} \left( \boldsymbol{\Sigma}_{\mathrm{s}} \boldsymbol{V}_{\mathrm{s}}^{\mathrm{H}} \right)^{\dagger} \boldsymbol{e}_l \\
%= & \boldsymbol{\epsilon}_l^{\mathrm{T}} \boldsymbol{E} \left( \breve{\boldsymbol{J}}_{n, 1} \boldsymbol{P} \right)^{\dagger} \left( \breve{\boldsymbol{J}}_{n, 2} - \Phi_{l, n} \breve{\boldsymbol{J}}_{n, 1} \right) \Delta \boldsymbol{H} \left( \boldsymbol{G}^{\mathrm{T}} \right)^{\dagger} \\
%& \left( \mathrm{Diag} \left( \boldsymbol{\gamma} \right) \right)^{-1} \boldsymbol{E}^{-1} \boldsymbol{e}_l \\
= & \frac{1}{\gamma_l} \boldsymbol{b}_l^{\mathrm{T}} \left( \breve{\boldsymbol{J}}_{n, 1} \boldsymbol{P} \right)^{\dagger} \left( \breve{\boldsymbol{J}}_{n, 2} - \Phi_{l, n} \breve{\boldsymbol{J}}_{n, 1} \right) \Delta \boldsymbol{H} \left( \boldsymbol{G}^{\mathrm{T}} \right)^{\dagger} \boldsymbol{b}_l,
\end{align*} 
which is the result in \eqref{eq49}.

%-----------------------------------------------------------------------------
\subsection{Proof of Proposition \ref{prop1}} \label{ProofProp1}
We first present the following lemma.
\begin{lemma} \label{MlemB4}
Given $\boldsymbol{a} = [\boldsymbol{a}_{1, 1, 1, 1}^{\mathrm{T}}, \boldsymbol{a}_{1, 1, 1, 2}^{\mathrm{T}}, \cdots, \boldsymbol{a}_{N_1, N_2, N_3, N_4}^{\mathrm{T}} ]^{\mathrm{T}}$
%\begin{align*}
%\boldsymbol{a} = [\boldsymbol{a}_{1, 1, 1, 1}^{\mathrm{T}}, \boldsymbol{a}_{1, 1, 1, 2}^{\mathrm{T}}, \cdots, \boldsymbol{a}_{1, 1, 1, N_4}^{\mathrm{T}}, \boldsymbol{a}_{1, 1, 2, 1}^{\mathrm{T}}, \cdots, \boldsymbol{a}_{N_1, N_2, N_3, N_4}^{\mathrm{T}} ]^{\mathrm{T}},
%\end{align*}
where $\boldsymbol{a}_{n_1, n_2, n_3, n_4} \in \mathbb{C}^{K_5 \times 1}$ and $\boldsymbol{b} \in \mathbb{C}^{L_5 \times 1}$, we have $\boldsymbol{a}^{\mathrm{H}} \Delta \boldsymbol{H} \boldsymbol{b}^* = \boldsymbol{c}^{\mathrm{H}} \Delta \boldsymbol{h}$, 
%\begin{align*}
%\boldsymbol{a}^{\mathrm{H}} \Delta \boldsymbol{H} \boldsymbol{b}^* = \boldsymbol{c}^{\mathrm{H}} \Delta \boldsymbol{h}, \label{BB1} \numberthis
%\end{align*}
where $\boldsymbol{c} = [\boldsymbol{c}_{1, 1, 1, 1}^{\mathrm{T}}, \boldsymbol{c}_{1, 1, 1, 2}^{\mathrm{T}},  \cdots, \boldsymbol{c}_{N_1, N_2, N_3, N_4}^{\mathrm{T}} ]^{\mathrm{T}}$
%\begin{align*}
%\boldsymbol{c} = [\boldsymbol{c}_{1, 1, 1, 1}^{\mathrm{T}}, \boldsymbol{c}_{1, 1, 1, 2}^{\mathrm{T}}, \cdots, \boldsymbol{c}_{1, 1, 1, N_4}^{\mathrm{T}}, \boldsymbol{c}_{1, 1, 2, 1}^{\mathrm{T}}, \cdots, \boldsymbol{c}_{N_1, N_2, N_3, N_4}^{\mathrm{T}} ]^{\mathrm{T}},
%\end{align*}
and $\boldsymbol{c}_{n_1, n_2, n_3, n_4} \in \mathbb{C}^{M_5 \times 1}$ is the convolution of $\boldsymbol{a}_{n_1, n_2, n_3, n_4}$ and $\boldsymbol{b}$.
\end{lemma}
\begin{proof}
Let $\boldsymbol{y} = \Delta \boldsymbol{H}\boldsymbol{b}^*$, where $\boldsymbol{y} \in \mathbb{C}^{N_1N_2N_3N_4K_5 \times 1}$ is given by $\boldsymbol{y} = [\boldsymbol{y}_{1, 1, 1, 1}^{\mathrm{T}}, \boldsymbol{y}_{1, 1, 1, 2}^{\mathrm{T}},\cdots, \boldsymbol{y}_{N_1, N_2, N_3, N_4}^{\mathrm{T}} ]^{\mathrm{T}},$
%\begin{align*}
%\boldsymbol{y} = [\boldsymbol{y}_{1, 1, 1, 1}^{\mathrm{T}}, \boldsymbol{y}_{1, 1, 1, 2}^{\mathrm{T}}, \cdots, \boldsymbol{y}_{1, 1, 1, N_4}^{\mathrm{T}}, \boldsymbol{y}_{1, 1, 2, 1}^{\mathrm{T}}, \cdots, \boldsymbol{y}_{N_1, N_2, N_3, N_4}^{\mathrm{T}} ]^{\mathrm{T}},
%\end{align*}
and $\boldsymbol{y}_{n_1, n_2, n_3, n_4} \in \mathbb{C}^{K_5 \times 1}$. Accordingly, we have
%\begin{align*}
$\boldsymbol{y}_{n_1, n_2, n_3, n_4} = \Delta \boldsymbol{H}_{n_1, n_2, n_3, n_4} \boldsymbol{b}^*$. 
%\end{align*}
As a result, 
\begin{align*}
\boldsymbol{a}^{\mathrm{H}} \Delta \boldsymbol{H} \boldsymbol{b}^* %= & \sum_{n_1 = 1}^{N_1} \sum_{n_2 = 1}^{N_2} \sum_{n_3 = 1}^{N_3} \sum_{n_4 = 1}^{N_4} \boldsymbol{a}_{n_1, n_2, n_3, n_4}^{\mathrm{H}} \boldsymbol{y}_{n_1, n_2, n_3, n_4} \\
%= & \sum_{n_1 = 1}^{N_1} \sum_{n_2 = 1}^{N_2} \sum_{n_3 = 1}^{N_3} \sum_{n_4 = 1}^{N_4} \sum_{k_5 = 1}^{K_5} a_{n_1, n_2, n_3, n_4, k_5}^* \\
%& \hspace{8mm}\cdot \sum_{\ell_5 = 1}^{L_5} \Delta h_{n_1, n_2, n_3, n_4, k_5 + \ell_5 - 1} b_{\ell_5}^* \\
= & \sum_{n_1 = 1}^{N_1} \sum_{n_2 = 1}^{N_2} \sum_{n_3 = 1}^{N_3} \sum_{n_4 = 1}^{N_4} \sum_{m_5 = 1}^{M_5} \Delta h_{n_1, n_2, n_3, n_4, m_5} \\
& \hspace{8mm}\underbrace{\sum_{\ell_5 = 1}^{L_5} a_{n_1, n_2, n_3, n_4, m_5 + 1 - \ell_5}^* b_{\ell_5}^*}_{c_{n_1, n_2, n_3, n_4, m_5}^*} %\\
%= & \boldsymbol{c}^{\mathrm{H}} \Delta \boldsymbol{h}, \numberthis
\end{align*}
which shows that $\boldsymbol{a}^{\mathrm{H}} \Delta \boldsymbol{H} \boldsymbol{b}^*=\boldsymbol{c}^{\mathrm{H}} \Delta \boldsymbol{h}$.
\end{proof}
With the results in Lemma \ref{lem1}, we obtain the results in \eqref{eq54} by applying the results in Lemma \ref{MlemB4}.
%-----------------------------------------------------------------------------
\subsection{Proof of Lemma \ref{lem2}} \label{ProofLem2}
%The relationship between the angular frequencies and the channel parameters are given as follows.
%\begin{align*}
%\omega_{l, 1} = & \pi \sin \left( \phi_{\mathrm{az}}^l \right) \sin \left( \phi_{\mathrm{el}}^l \right), \\
%\omega_{l, 2} = & \pi \cos \left(\phi_{\mathrm{el}}^l \right), \\
%\omega_{l, 3} = & \pi \sin \left( \theta_{\mathrm{az}}^l \right) \sin \left( \theta_{\mathrm{el}}^l \right), \\
%\omega_{l, 4} = & \pi \cos \left(\theta_{\mathrm{el}}^l \right), \\
%\omega_{l, 5} = & -2\pi \Delta f \tau_l.
%\end{align*}
With the relationship between the angular frequencies and channel parameters, we establish the following first-order perturbations:
$\Delta \omega_{l, 1} =  \pi \cos \left( \phi_{\mathrm{az}}^l \right) \sin \left( \phi_{\mathrm{el}}^l \right) \Delta \phi_{\mathrm{az}}^l + \pi \sin \left( \phi_{\mathrm{az}}^l \right) \cos \left( \phi_{\mathrm{el}}^l \right) \Delta \phi_{\mathrm{el}}^l$, 
$\Delta \omega_{l, 2} =  - \pi \sin \left( \phi_{\mathrm{el}}^l \right) \Delta \phi_{\mathrm{el}}^l$, 
$\Delta \omega_{l, 3} =  \pi \cos \left( \theta_{\mathrm{az}}^l \right) \sin \left( \theta_{\mathrm{el}}^l \right) \Delta \theta_{\mathrm{az}}^l  + \pi \sin \left( \theta_{\mathrm{az}}^l \right) \cos \left( \theta_{\mathrm{el}}^l \right) \Delta \theta_{\mathrm{el}}^l$, 
$\Delta \omega_{l, 4} =  - \pi \sin \left( \theta_{\mathrm{el}}^l \right) \Delta \theta_{\mathrm{el}}^l$, and 
$\Delta \omega_{l, 5} =  -2\pi \Delta f \Delta \tau_l$.
%\begin{align*}
%\Delta \omega_{l, 1} = & \pi \cos \left( \phi_{\mathrm{az}}^l \right) \sin \left( \phi_{\mathrm{el}}^l \right) \Delta \phi_{\mathrm{az}}^l \\
%& + \pi \sin \left( \phi_{\mathrm{az}}^l \right) \cos \left( \phi_{\mathrm{el}}^l \right) \Delta \phi_{\mathrm{el}}^l, \\
%\Delta \omega_{l, 2} = & - \pi \sin \left( \phi_{\mathrm{el}}^l \right) \Delta \phi_{\mathrm{el}}^l, \\
%\Delta \omega_{l, 3} = & \pi \cos \left( \theta_{\mathrm{az}}^l \right) \sin \left( \theta_{\mathrm{el}}^l \right) \Delta \theta_{\mathrm{az}}^l \\
%& + \pi \sin \left( \theta_{\mathrm{az}}^l \right) \cos \left( \theta_{\mathrm{el}}^l \right) \Delta \theta_{\mathrm{el}}^l, \\
%\Delta \omega_{l, 4} = & - \pi \sin \left( \theta_{\mathrm{el}}^l \right) \Delta \theta_{\mathrm{el}}^l, \\
%\Delta \omega_{l, 5} = & -2\pi \Delta f \Delta \tau_l.
%\end{align*}
Using the results in \eqref{eq55}, it is straightforward to obtain the results in \eqref{eq56}--\eqref{eq60}. 

Obtaining \eqref{eq61} is more involved. 
First of all, with $\Delta \omega_{l, n}$, we have $\Delta \left( \exp \left( \jmath m_n \omega_{l, n} \right) \right) = \jmath m_n \exp \left(\jmath m_n \omega_{l, n} \right) \Delta \omega_{l, n}$, 
%\begin{align*}
%\Delta \left( \exp \left( \jmath m_n \omega_{l, n} \right) \right) = \jmath m_n \exp \left(\jmath m_n \omega_{l, n} \right) \Delta \omega_{l, n}.
%\end{align*}
so that the first-order perturbation of $\boldsymbol{A}_n^{\left( M_n \right) }$ is given by $\Delta \boldsymbol{A}_n^{\left( M_n \right) } = \jmath \mathrm{Diag} \left( \boldsymbol{m}_n \right) \boldsymbol{A}_n^{\left( M_n \right) } \mathrm{Diag} \left( \Delta \boldsymbol{\omega}_n \right)$. 
%\begin{align*}
%\Delta \boldsymbol{A}_n^{\left( M_n \right) } = \jmath \mathrm{Diag} \left( \boldsymbol{m}_n \right) \boldsymbol{A}_n^{\left( M_n \right) } \mathrm{Diag} \left( \Delta \boldsymbol{\omega}_n \right). \numberthis
%\end{align*}
where $\boldsymbol{m}_n = [0, 1, \cdots, M_n - 1 ]^{\mathrm{T}}$ and $\Delta \boldsymbol{\omega}_n = [ \Delta \omega_{1, n}, \Delta \omega_{2, n}, \cdots, \Delta \omega_{L, n} ]^{\mathrm{T}} = \Im \{ \boldsymbol{V}_n^{\mathrm{H}} \Delta \boldsymbol{h} \}$.
From this, we see that the first-order perturbation of $\boldsymbol{B}_n^{\left( N_n \right) }$ is given by $\Delta \boldsymbol{B}_n^{\left( N_n \right) } = \boldsymbol{T}_n^{\mathrm{H}} \Delta \boldsymbol{A}_n^{\left( M_n \right) } = \breve{\boldsymbol{T}}_n^{\mathrm{H}} \boldsymbol{A}_n^{\left( M_n \right) } \odot \Delta \boldsymbol{\omega}_n^{\mathrm{T}}$
%\begin{align*}
%\Delta \boldsymbol{B}_n^{\left( N_n \right) } = \boldsymbol{T}_n^{\mathrm{H}} \Delta \boldsymbol{A}_n^{\left( M_n \right) } = \breve{\boldsymbol{T}}_n^{\mathrm{H}} \boldsymbol{A}_n^{\left( M_n \right) } \odot \Delta \boldsymbol{\omega}_n^{\mathrm{T}}, \numberthis
%\end{align*}
where $\breve{\boldsymbol{T}}_n = - \jmath \mathrm{Diag} \left( \boldsymbol{m}_n \right) \boldsymbol{T}_n$. 
%\begin{align*}
%\breve{\boldsymbol{T}}_n = - \jmath \mathrm{Diag} \left( \boldsymbol{m}_n \right) \boldsymbol{T}_n.
%\end{align*}
From \eqref{eq34}, we derive the first-order perturbation $\Delta \boldsymbol{B} $ as $\Delta \boldsymbol{B} = \sum_{n=1}^5 \breve{\boldsymbol{B}}_n \odot \Delta \boldsymbol{\omega}_n^{\mathrm{T}}$, 
%%\begin{align*}
%\Delta \boldsymbol{B} = \sum_{n=1}^5 \breve{\boldsymbol{B}}_n %\odot \Delta \boldsymbol{\omega}_n^{\mathrm{T}}, \numberthis
%\end{align*}
where
\begin{align*}
\breve{\boldsymbol{B}}_n = & \boldsymbol{B}_1^{\left( N_1 \right)} \odot \cdots \odot \breve{\boldsymbol{T}}_n^{\mathrm{H}} \boldsymbol{A}_n^{\left( M_n \right) }  \odot \cdots \odot \boldsymbol{B}_5^{\left( N_5 \right)}. \label{MeqBC1} \numberthis
\end{align*}
From the proof of Lemma \ref{MlemB2}, we recall that $\Delta \left( \boldsymbol{B}^{\dagger} \right) = -\boldsymbol{B}^{\dagger} \Delta\boldsymbol{B} \boldsymbol{B}^{\dagger} +  \left( \boldsymbol{B}^{\mathrm{H}} \boldsymbol{B} \right)^{-1} \Delta\boldsymbol{B}^{\mathrm{H}} \left( \boldsymbol{I} - \boldsymbol{B} \boldsymbol{B}^{\dagger} \right)$.
%\begin{align*}
%\Delta \left( \boldsymbol{B}^{\dagger} \right) = -\boldsymbol{B}^{\dagger} \Delta\boldsymbol{B} \boldsymbol{B}^{\dagger} +  \left( \boldsymbol{B}^{\mathrm{H}} \boldsymbol{B} \right)^{-1} \Delta\boldsymbol{B}^{\mathrm{H}} \left( \boldsymbol{I} - \boldsymbol{B} \boldsymbol{B}^{\dagger} \right). \numberthis
%\end{align*}
Secondly, we note that the estimate of the channel gain $\boldsymbol{\gamma}$ is given as
\begin{align*}
\hat{\boldsymbol{\gamma}} %= & \left( \boldsymbol{B}^{\dagger} + \Delta \left( \boldsymbol{B}^{\dagger} \right) \right) \left( \boldsymbol{B} \boldsymbol{\gamma} + \Delta \boldsymbol{h} \right) \\
= & \boldsymbol{\gamma} + \boldsymbol{B}^{\dagger} \Delta \boldsymbol{h} + \Delta \left( \boldsymbol{B}^{\dagger} \right) \boldsymbol{B} \boldsymbol{\gamma} + \Delta \left( \boldsymbol{B}^{\dagger} \right) \Delta \boldsymbol{h}, \numberthis
\end{align*}
which indicates that the first-order perturbation of the channel gain estimation is
\begin{align*}
\Delta \boldsymbol{\gamma} = & \boldsymbol{B}^{\dagger} \Delta \boldsymbol{h} + \Delta \left( \boldsymbol{B}^{\dagger} \right) \boldsymbol{B} \boldsymbol{\gamma} \\
%= & \boldsymbol{B}^{\dagger} \Delta \boldsymbol{h} - \boldsymbol{B}^{\dagger}\Delta\boldsymbol{B} \boldsymbol{\gamma} + \left( \boldsymbol{B}^{\mathrm{H}} \boldsymbol{B} \right)^{-1} \Delta\boldsymbol{B}^{\mathrm{H}} \boldsymbol{B} \boldsymbol{\gamma} \\
%& - \left( \boldsymbol{B}^{\mathrm{H}} \boldsymbol{B} \right)^{-1} \Delta\boldsymbol{B}^{\mathrm{H}} \boldsymbol{B} \boldsymbol{\gamma}\\
%= & \boldsymbol{B}^{\dagger} \Delta \boldsymbol{h} - \boldsymbol{B}^{\dagger}\Delta\boldsymbol{B} \boldsymbol{\gamma} %\\
= & \boldsymbol{B}^{\dagger} \Delta \boldsymbol{h} - \sum_{n=1}^5 \boldsymbol{\Upsilon}_n \Im \left\{ \boldsymbol{V}_n^{\mathrm{H}} \Delta \boldsymbol{h} \right\}. \numberthis
\end{align*}

%-----------------------------------------------------------------------------
% \subsection{Proof of Lemma \ref{lem3}} \label{ProofLem3}
% From \eqref{eq7}, we have the first-order perturbation $\Delta \tilde{\boldsymbol{h}} = \boldsymbol{}$

%-----------------------------------------------------------------------------
\subsection{Proof of Lemma \ref{lem3}} \label{ProofLem3}
To begin with, we have the first-order approximations of $\Delta \boldsymbol{f}_{\mathrm{T}, l}$ and $\Delta \boldsymbol{f}_{\mathrm{R}, l}$ as $\Delta \boldsymbol{f}_{\mathrm{T}, l} =  \boldsymbol{\Omega}_{\mathrm{T}, l} \left[ \Delta \phi_{\mathrm{az}}^l, \Delta \phi_{\mathrm{el}}^l \right]^{\mathrm{T}}$ and $\Delta \boldsymbol{f}_{\mathrm{R}, l} =  \boldsymbol{\Omega}_{\mathrm{R}, l} \left[ \Delta \theta_{\mathrm{az}}^l, \Delta \theta_{\mathrm{el}}^l \right]^{\mathrm{T}}$, 
%\begin{align*}
%\Delta \boldsymbol{f}_{\mathrm{T}, l} = & \boldsymbol{\Omega}_{\mathrm{T}, l} \left[ \Delta \phi_{\mathrm{az}}^l, \Delta \phi_{\mathrm{el}}^l \right]^{\mathrm{T}}, \numberthis \\
%\Delta \boldsymbol{f}_{\mathrm{R}, l} = & \boldsymbol{\Omega}_{\mathrm{R}, l} \left[ \Delta \theta_{\mathrm{az}}^l, \Delta \theta_{\mathrm{el}}^l \right]^{\mathrm{T}}, \numberthis
%\end{align*}
respectively, where $\boldsymbol{\Omega}_{\mathrm{T}, l} \in \mathbb{R}^{3\times 2}$ and $\boldsymbol{\Omega}_{\mathrm{R}, l} \in \mathbb{R}^{3\times 2}$ are given by
\begin{align*}
\boldsymbol{\Omega}_{\mathrm{T}, l} = & \left[ 
\begin{array}{cc}
-\sin \left( \phi_{\mathrm{az}}^l \right) \sin \left( \phi_{\mathrm{el}}^l \right) & \cos \left( \phi_{\mathrm{az}}^l \right) \cos \left( \phi_{\mathrm{el}}^l \right) \\
\cos \left( \phi_{\mathrm{az}}^l \right) \sin \left( \phi_{\mathrm{el}}^l \right) & \sin \left( \phi_{\mathrm{az}}^l \right) \cos \left( \phi_{\mathrm{el}}^l \right) \\
0 & - \sin \left( \phi_{\mathrm{el}}^l \right)
\end{array}
\right], \numberthis \\
\boldsymbol{\Omega}_{\mathrm{R}, l} = & \left[ 
\begin{array}{cc}
-\sin \left( \theta_{\mathrm{az}}^l \right) \sin \left( \theta_{\mathrm{el}}^l \right) & \cos \left( \theta_{\mathrm{az}}^l \right) \cos \left( \theta_{\mathrm{el}}^l \right) \\
\cos \left( \theta_{\mathrm{az}}^l \right) \sin \left( \theta_{\mathrm{el}}^l \right) & \sin \left( \theta_{\mathrm{az}}^l \right) \cos \left( \theta_{\mathrm{el}}^l \right) \\
0 & - \sin \left( \theta_{\mathrm{el}}^l \right)
\end{array}
\right]. \numberthis
\end{align*}
% \begin{align*}
% & \Delta \boldsymbol{f}_{\mathrm{T}, l} \\
% = & \left[
% \begin{array}{c}
% -\sin \left( \phi_{\mathrm{az}}^l \right) \sin \left( \phi_{\mathrm{el}}^l \right) \Delta \phi_{\mathrm{az}}^l + \cos \left( \phi_{\mathrm{az}}^l \right) \cos \left( \phi_{\mathrm{el}}^l \right) \Delta \phi_{\mathrm{el}}^l  \\
% \cos \left( \phi_{\mathrm{az}}^l \right) \sin \left( \phi_{\mathrm{el}}^l \right) \Delta \phi_{\mathrm{az}}^l + \sin \left( \phi_{\mathrm{az}}^l \right) \cos \left( \phi_{\mathrm{el}}^l \right) \Delta \phi_{\mathrm{el}}^l \\
% - \sin \left( \phi_{\mathrm{el}}^l \right) \Delta \phi_{\mathrm{el}}^l
% \end{array}
% \right], \numberthis \\
% & \Delta \boldsymbol{f}_{\mathrm{R}, l} \\
% = & \left[
% \begin{array}{c}
% -\sin \left( \theta_{\mathrm{az}}^l \right) \sin \left( \theta_{\mathrm{el}}^l \right) \Delta \theta_{\mathrm{az}}^l + \cos \left( \theta_{\mathrm{az}}^l \right) \cos \left( \theta_{\mathrm{el}}^l \right) \Delta \theta_{\mathrm{el}}^l  \\
% \cos \left( \theta_{\mathrm{az}}^l \right) \sin \left( \theta_{\mathrm{el}}^l \right) \Delta \theta_{\mathrm{az}}^l + \sin \left( \theta_{\mathrm{az}}^l \right) \cos \left( \theta_{\mathrm{el}}^l \right) \Delta \theta_{\mathrm{el}}^l \\
% - \sin \left( \theta_{\mathrm{el}}^l \right) \Delta \theta_{\mathrm{el}}^l
% \end{array}
% \right]. \numberthis
% \end{align*}
As a result, the first-order approximations of $\Delta \boldsymbol{\delta}_l$ and $\Delta \boldsymbol{\mu}_l$ are obtained as
\begin{align*}
\Delta \boldsymbol{\delta}_l = & - c \left[
\tau_l \boldsymbol{\Omega}_{\mathrm{R}, l},  \boldsymbol{f}_{\mathrm{R}, l}
\right] \left[\Delta \theta_{\mathrm{az}}^l, \Delta \theta_{\mathrm{el}}^l, \Delta \tau_l \right]^{\mathrm{T}}, \label{MeqC1} \numberthis \\
\Delta \boldsymbol{\mu}_l = & c \left[\tau_l \boldsymbol{\Omega}_{\mathrm{T}, l}, \tau_l \boldsymbol{\Omega}_{\mathrm{R}, l}, \left(\boldsymbol{f}_{\mathrm{T}, l} + \boldsymbol{f}_{\mathrm{R}, l} \right) \right] \\
& \hspace{9mm}\left[\Delta \phi_{\mathrm{az}}^l, \Delta \phi_{\mathrm{el}}^l, \Delta \theta_{\mathrm{az}}^l, \Delta \theta_{\mathrm{el}}^l, \Delta \tau_l \right]^{\mathrm{T}}. \label{MeqC2} \numberthis
\end{align*}
% \begin{align*}
% \Delta \boldsymbol{\delta}_l = & - c \Delta \tau_l \boldsymbol{f}_{\mathrm{R}, l} - c \tau_l \Delta \boldsymbol{f}_{\mathrm{R}, l}, \numberthis \\
% \Delta \boldsymbol{\mu}_l = & c \Delta \tau_l \left(\boldsymbol{f}_{\mathrm{T}, l} + \boldsymbol{f}_{\mathrm{R}, l} \right) + c  \tau_l \left( \Delta \boldsymbol{f}_{\mathrm{T}, l} + \Delta \boldsymbol{f}_{\mathrm{R}, l} \right). \numberthis
% \end{align*}
Since $\boldsymbol{C}_l = (\boldsymbol{I} - \boldsymbol{\mu}_l \boldsymbol{\mu}_l^{\mathrm{T}}/{ \| \boldsymbol{\mu}_l \|^2}  )$, $\Delta \boldsymbol{C}_l$ can be given as\footnote{We omit $\iota_l$.}
% \begin{align*}
% \Delta \boldsymbol{C}_l = \frac{2\boldsymbol{\mu}_l^{\mathrm{T}} \Delta \boldsymbol{\mu}_l \left( \boldsymbol{I} - \boldsymbol{C}_l \right) - \Delta \boldsymbol{\mu}_l \boldsymbol{\mu}_l^{\mathrm{T}} - \boldsymbol{\mu}_l \Delta \boldsymbol{\mu}_l^{\mathrm{T}} }{\left\|\boldsymbol{\mu}_l\right\|^2 }. \numberthis
% \end{align*}
\begin{align*}
\Delta \boldsymbol{C}_l = & 2\left( \boldsymbol{\mu}_l^{\mathrm{T}} \boldsymbol{\mu}_l \right)^{-2} \left( \boldsymbol{\mu}_l^{\mathrm{T}} \Delta \boldsymbol{\mu}_l \right) \boldsymbol{\mu}_l \boldsymbol{\mu}_l^{\mathrm{T}} \\
& - \left( \boldsymbol{\mu}_l^{\mathrm{T}} \boldsymbol{\mu}_l \right)^{-1} \left( \Delta \boldsymbol{\mu}_l \boldsymbol{\mu}_l^{\mathrm{T}} + \boldsymbol{\mu}_l \Delta \boldsymbol{\mu}_l^{\mathrm{T}} \right). \numberthis
\end{align*}
According to \eqref{eq:PosEst}, and define $\boldsymbol{C} = \sum_{l=1}^L \boldsymbol{C}_l$, we have
\begin{align*}
\hat{\boldsymbol{p}}_{\mathrm{R}} = & \left( \sum_{l=1}^L \boldsymbol{C}_l + \Delta \boldsymbol{C}_l \right)^{-1} \sum_{l=1}^L \left( \boldsymbol{C}_l + \Delta \boldsymbol{C}_l \right) \left( \boldsymbol{\delta}_l + \Delta \boldsymbol{\delta}_l \right) \\
%= & \left( \boldsymbol{I} + \sum_{l=1}^L \boldsymbol{C}^{-1} \Delta \boldsymbol{C}_l \right)^{-1} \boldsymbol{C}^{-1} \\
%& \sum_{l=1}^L \boldsymbol{C}_l \boldsymbol{\delta}_l + \Delta \boldsymbol{C}_l \boldsymbol{\delta}_l + \boldsymbol{C}_l \Delta \boldsymbol{\delta}_l + \Delta \boldsymbol{C}_l \Delta \boldsymbol{\delta}_l \\
%\approx & \left( \boldsymbol{I} - \sum_{l=1}^L \boldsymbol{C}^{-1} \Delta \boldsymbol{C}_l + \boldsymbol{C}^{-1} \sum_{l=1}^L \Delta \boldsymbol{C}_l \boldsymbol{C}^{-1} \sum_{l=1}^L \Delta \boldsymbol{C}_l \right) \\
%& \boldsymbol{C}^{-1} \sum_{l=1}^L \boldsymbol{C}_l \boldsymbol{\delta}_l + \Delta \boldsymbol{C}_l \boldsymbol{\delta}_l + \boldsymbol{C}_l \Delta \boldsymbol{\delta}_l + \Delta \boldsymbol{C}_l \Delta \boldsymbol{\delta}_l \\
\approx & \boldsymbol{p}_{\mathrm{R}} + \boldsymbol{C}^{-1} \sum_{l=1}^L \Delta \boldsymbol{C}_l \left( \boldsymbol{\delta}_l - \boldsymbol{p}_{\mathrm{R}} \right) + \boldsymbol{C}^{-1} \sum_{l=1}^L \boldsymbol{C}_l \Delta \boldsymbol{\delta}_l \\
&x
+ \boldsymbol{C}^{-1} \sum_{l=1}^L \Delta \boldsymbol{C}_l \Delta \boldsymbol{\delta}_l - \boldsymbol{C}^{-1} \sum_{l=1}^L \Delta \boldsymbol{C}_l \boldsymbol{C}^{-1} \sum_{l=1}^L \boldsymbol{C}_l \Delta \boldsymbol{\delta}_l \\
& + \boldsymbol{C}^{-1} \sum_{l=1}^L \Delta \boldsymbol{C}_l \boldsymbol{C}^{-1} \sum_{l=1}^L \Delta \boldsymbol{C}_l \left( \boldsymbol{p}_{\mathrm{R}} - \boldsymbol{\delta}_l\right), \label{MeqC2nd} \numberthis
\end{align*}
where higher order approximation is omitted. Therefore, the first-order perturbation of the position estimate, i.e., $\Delta \boldsymbol{p}_{\mathrm{R}} = \hat{\boldsymbol{p}}_{\mathrm{R}} - \boldsymbol{p}_{\mathrm{R}}$, can be obtained as
\begin{align*}
\Delta \boldsymbol{p}_{\mathrm{R}} = &  \boldsymbol{C}^{-1} \sum_{l=1}^L \Delta \boldsymbol{C}_l \left( \boldsymbol{\delta}_l - \boldsymbol{p}_{\mathrm{R}} \right) + \boldsymbol{C}^{-1} \sum_{l=1}^L \boldsymbol{C}_l \Delta \boldsymbol{\delta}_l \\
& = \boldsymbol{C}^{-1} \sum_{l=1}^L \breve{\boldsymbol{C}}_l \Delta \boldsymbol{\mu}_l + \boldsymbol{C}_l \Delta \boldsymbol{\delta}_l, \label{MeqC3} \numberthis
\end{align*}
where $\breve{\boldsymbol{C}}_l \in \mathbb{R}^{3\times 3}$ is given by
\begin{align*}
\breve{\boldsymbol{C}}_l = & \frac{2\boldsymbol{\mu}_l^{\mathrm{T}} \left( \boldsymbol{\delta}_l - \boldsymbol{p}_{\mathrm{R}} \right) \boldsymbol{\mu}_l \boldsymbol{\mu}_l^{\mathrm{T}}}{\left\|\boldsymbol{\mu}_l\right\|^4} \\
& - \frac{\boldsymbol{\mu}_l^{\mathrm{T}} \left( \boldsymbol{\delta}_l - \boldsymbol{p}_{\mathrm{R}} \right) \boldsymbol{I} + \boldsymbol{\mu}_l \left( \boldsymbol{\delta}_l - \boldsymbol{p}_{\mathrm{R}} \right)^{\mathrm{T}} }{\left\|\boldsymbol{\mu}_l\right\|^2}. \numberthis
\end{align*}
% \begin{align*}
% \breve{\boldsymbol{C}}_l = & \frac{ 2 \left( \boldsymbol{I} - \boldsymbol{C}_l \right) \left( \boldsymbol{\delta}_l - \boldsymbol{p}_{\mathrm{R}} \right)  \boldsymbol{\mu}_l^{\mathrm{T}} - \boldsymbol{\mu}_l^{\mathrm{T}} \left( \boldsymbol{\delta}_l - \boldsymbol{p}_{\mathrm{R}} \right) \boldsymbol{I} }{\left\|\boldsymbol{\mu}_l\right\|^2} \\
% & - \frac{ \boldsymbol{\mu}_l \left( \boldsymbol{\delta}_l - \boldsymbol{p}_{\mathrm{R}} \right)^{\mathrm{T}} }{\left\|\boldsymbol{\mu}_l\right\|^2}. \numberthis
% \end{align*}
Substituting the results in \eqref{MeqC1} and \eqref{MeqC2} to \eqref{MeqC3}, we arrive at 
\begin{align*}
\Delta \boldsymbol{p}_{\mathrm{R}} = \sum_{l=1}^L & \breve{\boldsymbol{D}}_l \left[\Delta \theta_{\mathrm{az}}^l, \Delta \theta_{\mathrm{el}}^l, \Delta \tau_l \right]^{\mathrm{T}} \\
& + \breve{\boldsymbol{E}}_l \left[\Delta \phi_{\mathrm{az}}^l, \Delta \phi_{\mathrm{el}}^l, \Delta \theta_{\mathrm{az}}^l, \Delta \theta_{\mathrm{el}}^l, \Delta \tau_l \right]^{\mathrm{T}}, \numberthis
\end{align*}
where $\breve{\boldsymbol{D}}_l \in \mathbb{R}^{3\times 3}$ and $\breve{\boldsymbol{E}}_l\in \mathbb{R}^{3\times 5}$
\begin{align*}
\breve{\boldsymbol{D}}_l = & -c \left(\sum_{\ell=1}^L \boldsymbol{C}_\ell \right)^{-1} \boldsymbol{C}_l \left[
\tau_l \boldsymbol{\Omega}_{\mathrm{R}, l},  \boldsymbol{f}_{\mathrm{R}, l}
\right], \label{MeqC4} \numberthis \\
\breve{\boldsymbol{E}}_l = & c \left(\sum_{\ell=1}^L \boldsymbol{C}_\ell \right)^{-1} \breve{\boldsymbol{C}}_l \left[\tau_l \boldsymbol{\Omega}_{\mathrm{T}, l}, \tau_l \boldsymbol{\Omega}_{\mathrm{R}, l}, \left(\boldsymbol{f}_{\mathrm{T}, l} + \boldsymbol{f}_{\mathrm{R}, l} \right) \right]. \label{MeqC5} \numberthis
\end{align*}
With the results in Lemma \ref{lem2}, we derive the results in \eqref{eq75}.

\section*{Acknowledgment}

The authors would like to thank the editors and autonomous reviewers for your time and effort in handling the review of our paper.

\bibliographystyle{IEEEtran}
\bibliography{IEEEabrv,reference}

\clearpage
\newpage 
\renewcommand{\thepage}{}

\section*{\LARGE{Supplementary Material}}

\section{Proofs of Lemmas} \label{AppendixB}
%-----------------------------------------------------------------------------
\subsection{Proof of Lemma \ref{lem1}} \label{ProofLem1}
To begin with, we introduce the following results for the perturbation of signal subspace $\boldsymbol{U}_{\mathrm{s}}$.
\begin{lemma} \label{lemB1}
Given the perturbed $\tilde{\boldsymbol{H}}$ in \eqref{eq48}, and the subspace decomposition as
\begin{align*}
\tilde{\boldsymbol{H}} = \tilde{\boldsymbol{U}}_{\mathrm{s}} \tilde{\boldsymbol{\Sigma}}_{\mathrm{s}} \tilde{\boldsymbol{V}}_{\mathrm{s}}^{\mathrm{H}} + \tilde{\boldsymbol{U}}_{\mathrm{n}} \tilde{\boldsymbol{\Sigma}}_{\mathrm{n}} \tilde{\boldsymbol{V}}_{\mathrm{n}}^{\mathrm{H}},
\end{align*}
the first-order approximation of the $\Delta \boldsymbol{U}_{\mathrm{s}} = \tilde{\boldsymbol{U}}_{\mathrm{s}} - \boldsymbol{U}_{\mathrm{s}}$ is given by
\begin{align*}
\Delta \boldsymbol{U}_{\mathrm{s}} = \boldsymbol{U}_{\mathrm{n}}\boldsymbol{U}_{\mathrm{n}}^{\mathrm{H}} \Delta \boldsymbol{H} \boldsymbol{V}_{\mathrm{s}} \boldsymbol{\Sigma}_{\mathrm{s}}^{-1} + \boldsymbol{U}_{\mathrm{s}} \boldsymbol{R}, \label{eqB1} \numberthis
\end{align*}
where $\boldsymbol{R}$ is an antihermitian matrix that depends on $\Delta \boldsymbol{H}$.
\end{lemma}
\begin{proof}
The proof can be seen in \cite{SahUseCom17}.
\end{proof}
With $\Delta \boldsymbol{U}_{\mathrm{s}}$, the first-order approximation of $\Delta \boldsymbol{\Gamma}_n = \tilde{\boldsymbol{\Gamma}}_n - \boldsymbol{\Gamma}_n$ can be obtained as follows. 
\begin{lemma} \label{lemB2}
$\Delta \boldsymbol{\Gamma}_n$ is given by
\begin{align*}
\Delta \boldsymbol{\Gamma}_n = \left( \breve{\boldsymbol{J}}_{n, 1}\boldsymbol{U}_{\mathrm{s}} \right)^{\dagger} \left(\breve{\boldsymbol{J}}_{n, 2} \Delta \boldsymbol{U}_{\mathrm{s}} - \breve{\boldsymbol{J}}_{n, 1} \Delta \boldsymbol{U}_{\mathrm{s}} \boldsymbol{\Gamma}_n \right). \label{eqB2} \numberthis
\end{align*}
\end{lemma}
\begin{proof}
Using $\Delta\left( \boldsymbol{T}^{-1} \right) = - \boldsymbol{T}^{-1} \Delta \boldsymbol{T} \boldsymbol{T}^{-1}$, the perturbation of the inverse, $\Delta \left( \left( \boldsymbol{A}^{\mathrm{H}} \boldsymbol{A} \right)^{-1} \right)$ is given by
\begin{align*}
\Delta \left( \left( \boldsymbol{A}^{\mathrm{H}} \boldsymbol{A} \right)^{-1} \right) = & -\left( \boldsymbol{A}^{\mathrm{H}} \boldsymbol{A} \right)^{-1} \Delta \left( \boldsymbol{A}^{\mathrm{H}} \boldsymbol{A} \right) \left( \boldsymbol{A}^{\mathrm{H}} \boldsymbol{A} \right)^{-1} \\
= & -\left( \boldsymbol{A}^{\mathrm{H}} \boldsymbol{A} \right)^{-1} \Delta\boldsymbol{A}^{\mathrm{H}} \boldsymbol{A} \left( \boldsymbol{A}^{\mathrm{H}} \boldsymbol{A} \right)^{-1} \\
& - \left( \boldsymbol{A}^{\mathrm{H}} \boldsymbol{A} \right)^{-1} \boldsymbol{A}^{\mathrm{H}} \Delta\boldsymbol{A} \left( \boldsymbol{A}^{\mathrm{H}} \boldsymbol{A} \right)^{-1}. \numberthis
\end{align*}
In addition,
\begin{align*}
\Delta \left( \boldsymbol{A}^{\dagger} \right) = & -\left( \boldsymbol{A}^{\mathrm{H}} \boldsymbol{A} \right)^{-1} \Delta\boldsymbol{A}^{\mathrm{H}} \boldsymbol{A} \left( \boldsymbol{A}^{\mathrm{H}} \boldsymbol{A} \right)^{-1} \boldsymbol{A}^{\mathrm{H}} \\
& - \left( \boldsymbol{A}^{\mathrm{H}} \boldsymbol{A} \right)^{-1} \boldsymbol{A}^{\mathrm{H}} \Delta\boldsymbol{A} \left( \boldsymbol{A}^{\mathrm{H}} \boldsymbol{A} \right)^{-1} \boldsymbol{A}^{\mathrm{H}} \\
& + \left( \boldsymbol{A}^{\mathrm{H}} \boldsymbol{A} \right)^{-1} \Delta\boldsymbol{A}^{\mathrm{H}} \\
= & -\boldsymbol{A}^{\dagger} \Delta\boldsymbol{A} \boldsymbol{A}^{\dagger} +  \left( \boldsymbol{A}^{\mathrm{H}} \boldsymbol{A} \right)^{-1} \Delta\boldsymbol{A}^{\mathrm{H}} \left( \boldsymbol{I} - \boldsymbol{A} \boldsymbol{A}^{\dagger} \right). \label{eqBA1} \numberthis
\end{align*}
From \eqref{eq28}, we have
\begin{align*}
\Delta \boldsymbol{\Gamma}_n = & \Delta \left( \breve{\boldsymbol{J}}_{n, 1} \boldsymbol{U}_{\mathrm{s}} \right)^{\dagger} \breve{\boldsymbol{J}}_{n, 2} \boldsymbol{U}_{\mathrm{s}} + \left( \breve{\boldsymbol{J}}_{n, 1} \boldsymbol{U}_{\mathrm{s}} \right)^{\dagger} \breve{\boldsymbol{J}}_{n, 2} \Delta \boldsymbol{U}_{\mathrm{s}} \\
= & \left( \breve{\boldsymbol{J}}_{n, 1} \boldsymbol{U}_{\mathrm{s}} \right)^{\dagger} \left(\breve{\boldsymbol{J}}_{n, 2} \Delta \boldsymbol{U}_{\mathrm{s}} - \breve{\boldsymbol{J}}_{n, 1} \Delta \boldsymbol{U}_{\mathrm{s}} \boldsymbol{\Gamma}_n \right). \numberthis
\end{align*}
In the above derivation, the following equality is used.
\begin{align*}
\left( \boldsymbol{I} - \breve{\boldsymbol{J}}_{n, 1} \boldsymbol{U}_{\mathrm{s}} \left( \breve{\boldsymbol{J}}_{n, 1} \boldsymbol{U}_{\mathrm{s}} \right) ^{\dagger} \right) \breve{\boldsymbol{J}}_{n, 2} \boldsymbol{U}_{\mathrm{s}} = \boldsymbol{0}.
\end{align*}
\end{proof}
Denote $\boldsymbol{E} = [\boldsymbol{e}_1, \boldsymbol{e}_2, \cdots, \boldsymbol{e}_L]$ and $\boldsymbol{E}^{-1} = [\boldsymbol{\epsilon}_1, \boldsymbol{\epsilon}_2, \cdots, \boldsymbol{\epsilon}_L]^{\mathrm{T}}$. Then the first-order approximation of $\Delta \Phi_{l, n} = \tilde{\Phi}_{l, n} - {\Phi}_{l, n}$ is given as follows.
\begin{lemma} \label{lemB3}
$\Delta \Phi_{l, n}$ is given by
\begin{align*}
\Delta \Phi_{l, n} = \boldsymbol{\epsilon}_l^{\mathrm{T}} \Delta \boldsymbol{\Gamma}_n \boldsymbol{e}_l. \label{eqB3} \numberthis
\end{align*}
\end{lemma}
\begin{proof}
From \eqref{eq32}, we have
\begin{align*}
\Delta \boldsymbol{\Phi}_n = & \boldsymbol{E}^{-1} \Delta \boldsymbol{\Gamma}_n \boldsymbol{E} + \boldsymbol{E}^{-1} \boldsymbol{\Gamma}_n \Delta \boldsymbol{E} - \boldsymbol{E}^{-1} \Delta \boldsymbol{E} \boldsymbol{E}^{-1} \boldsymbol{\Gamma}_n \boldsymbol{E} \\
= & \boldsymbol{E}^{-1} \Delta \boldsymbol{\Gamma}_n \boldsymbol{E} + \boldsymbol{\Phi}_n \boldsymbol{E}^{-1} \Delta \boldsymbol{E} - \boldsymbol{E}^{-1} \Delta \boldsymbol{E} \boldsymbol{\Phi}_n. \numberthis
\end{align*} 
We extract the diagonal element from the right hand side, and the last two terms are canceled, leading to the results in \eqref{eqB3}.
\end{proof}
We are now able to apply the results in Lemma \ref{lemB1} and Lemma \ref{lemB2} to \eqref{eqB3}. For simplicity, we define $\boldsymbol{C} = \Delta \boldsymbol{H} \boldsymbol{V}_{\mathrm{s}} \boldsymbol{\Sigma}_{\mathrm{s}}^{-1}$. Using the results in \eqref{eqB1}, we then have $\Delta \boldsymbol{\Gamma}_n$ as
\begin{align*}
\Delta \boldsymbol{\Gamma}_n = & \left( \breve{\boldsymbol{J}}_{n, 1} \boldsymbol{U}_{\mathrm{s}} \right)^{\dagger} \big(\breve{\boldsymbol{J}}_{n, 2} \left( \boldsymbol{I} - \boldsymbol{U}_{\mathrm{s}}\boldsymbol{U}_{\mathrm{s}}^{\mathrm{H}} \right) \boldsymbol{C} - \breve{\boldsymbol{J}}_{n, 1} \left( \boldsymbol{I} - \boldsymbol{U}_{\mathrm{s}} \boldsymbol{U}_{\mathrm{s}}^{\mathrm{H}} \right) \boldsymbol{C} \boldsymbol{\Gamma}_n \big) \\
& + \left( \breve{\boldsymbol{J}}_{n, 1} \boldsymbol{U}_{\mathrm{s}} \right)^{\dagger} \left( \breve{\boldsymbol{J}}_{n, 2} \boldsymbol{U}_{\mathrm{s}} \boldsymbol{R} - \breve{\boldsymbol{J}}_{n, 1} \boldsymbol{U}_{\mathrm{s}} \boldsymbol{R} \boldsymbol{\Gamma}_n \right) \\
= & \left( \breve{\boldsymbol{J}}_{n, 1} \boldsymbol{U}_{\mathrm{s}} \right)^{\dagger} \left( \breve{\boldsymbol{J}}_{n, 2} \boldsymbol{C} - \breve{\boldsymbol{J}}_{n, 1} \boldsymbol{C} \boldsymbol{\Gamma}_n \right) - \boldsymbol{\Gamma}_n \boldsymbol{U}_{\mathrm{s}}^{\mathrm{H}} \boldsymbol{C} \\
& + \boldsymbol{U}_{\mathrm{s}}^{\mathrm{H}} \boldsymbol{C} \boldsymbol{\Gamma}_n + \boldsymbol{\Gamma}_n \boldsymbol{R} - \boldsymbol{R} \boldsymbol{\Gamma}_n.
\numberthis
\end{align*} 
With the result in \eqref{eqB3}, $\Delta \Phi_{l, n}$ is further obtained as
\begin{align*}
\Delta \Phi_{l, n} = \boldsymbol{\epsilon}_l^{\mathrm{T}} \left( \breve{\boldsymbol{J}}_{n, 1} \boldsymbol{U}_{\mathrm{s}} \right)^{\dagger} \left( \breve{\boldsymbol{J}}_{n, 2} - \Phi_{l, n} \breve{\boldsymbol{J}}_{n, 1} \right) \boldsymbol{C} \boldsymbol{e}_l,
\numberthis
\end{align*} 
where $\boldsymbol{\Gamma}_n \boldsymbol{e}_l = \Phi_{l, n} \boldsymbol{e}_l$ and $\boldsymbol{\epsilon}_l \boldsymbol{\Gamma}_n  = \Phi_{l, n} \boldsymbol{\epsilon}_l$ are used in the derivation.

In addition, from \eqref{eq23} and \eqref{eq24}, we have $\boldsymbol{U}_{\mathrm{s}} = \boldsymbol{P} \boldsymbol{E}^{-1}$, and $\boldsymbol{\Sigma}_{\mathrm{s}} \boldsymbol{V}_{\mathrm{s}}^{\mathrm{H}} = \boldsymbol{E} \mathrm{Diag} \left( \boldsymbol{\gamma} \right) \boldsymbol{G}^{\mathrm{T}} $. With these results, $\Delta \Phi_{l, n}$ is rewritten as
\begin{align*}
\Delta \Phi_{l, n} = & \boldsymbol{\epsilon}_l^{\mathrm{T}} \left( \breve{\boldsymbol{J}}_{n, 1} \boldsymbol{P} \boldsymbol{E}^{-1} \right)^{\dagger} \left( \breve{\boldsymbol{J}}_{n, 2} - \Phi_{l, n} \breve{\boldsymbol{J}}_{n, 1} \right) \Delta \boldsymbol{H} \left( \boldsymbol{\Sigma}_{\mathrm{s}} \boldsymbol{V}_{\mathrm{s}}^{\mathrm{H}} \right)^{\dagger} \boldsymbol{e}_l \\
= & \boldsymbol{\epsilon}_l^{\mathrm{T}} \boldsymbol{E} \left( \breve{\boldsymbol{J}}_{n, 1} \boldsymbol{P} \right)^{\dagger} \left( \breve{\boldsymbol{J}}_{n, 2} - \Phi_{l, n} \breve{\boldsymbol{J}}_{n, 1} \right) \Delta \boldsymbol{H} \left( \boldsymbol{G}^{\mathrm{T}} \right)^{\dagger} \\
& \left( \mathrm{Diag} \left( \boldsymbol{\gamma} \right) \right)^{-1} \boldsymbol{E}^{-1} \boldsymbol{e}_l \\
= & \frac{1}{\gamma_l} \boldsymbol{b}_l^{\mathrm{T}} \left( \breve{\boldsymbol{J}}_{n, 1} \boldsymbol{P} \right)^{\dagger} \left( \breve{\boldsymbol{J}}_{n, 2} - \Phi_{l, n} \breve{\boldsymbol{J}}_{n, 1} \right) \Delta \boldsymbol{H} \left( \boldsymbol{G}^{\mathrm{T}} \right)^{\dagger} \boldsymbol{b}_l,
\numberthis
\end{align*} 
which is the result in \eqref{eq49}.

%-----------------------------------------------------------------------------
\subsection{Proof of Proposition \ref{prop1}} \label{ProofProp1}
To begin with, we have the following results.
\begin{lemma} \label{lemB4}
Given
\begin{align*}
\boldsymbol{a} = [\boldsymbol{a}_{1, 1, 1, 1}^{\mathrm{T}}, \boldsymbol{a}_{1, 1, 1, 2}^{\mathrm{T}}, \cdots, \boldsymbol{a}_{1, 1, 1, N_4}^{\mathrm{T}}, \boldsymbol{a}_{1, 1, 2, 1}^{\mathrm{T}}, \cdots, \boldsymbol{a}_{N_1, N_2, N_3, N_4}^{\mathrm{T}} ]^{\mathrm{T}},
\end{align*}
where $\boldsymbol{a}_{n_1, n_2, n_3, n_4} \in \mathbb{C}^{K_5 \times 1}$ and $\boldsymbol{b} \in \mathbb{C}^{L_5 \times 1}$, we have
\begin{align*}
\boldsymbol{a}^{\mathrm{H}} \Delta \boldsymbol{H} \boldsymbol{b}^* = \boldsymbol{c}^{\mathrm{H}} \Delta \boldsymbol{h}, \label{BB1} \numberthis
\end{align*}
where 
\begin{align*}
\boldsymbol{c} = [\boldsymbol{c}_{1, 1, 1, 1}^{\mathrm{T}}, \boldsymbol{c}_{1, 1, 1, 2}^{\mathrm{T}}, \cdots, \boldsymbol{c}_{1, 1, 1, N_4}^{\mathrm{T}}, \boldsymbol{c}_{1, 1, 2, 1}^{\mathrm{T}}, \cdots, \boldsymbol{c}_{N_1, N_2, N_3, N_4}^{\mathrm{T}} ]^{\mathrm{T}},
\end{align*}
and $\boldsymbol{c}_{n_1, n_2, n_3, n_4} \in \mathbb{C}^{M_5 \times 1}$ is the convolution of $\boldsymbol{a}_{n_1, n_2, n_3, n_4}$ and $\boldsymbol{b}$.
\end{lemma}
\begin{proof}
Let $\boldsymbol{y} = \Delta \boldsymbol{H}\boldsymbol{b}^*$, where $\boldsymbol{y} \in \mathbb{C}^{N_1N_2N_3N_4K_5 \times 1}$ is given by
\begin{align*}
\boldsymbol{y} = [\boldsymbol{y}_{1, 1, 1, 1}^{\mathrm{T}}, \boldsymbol{y}_{1, 1, 1, 2}^{\mathrm{T}}, \cdots, \boldsymbol{y}_{1, 1, 1, N_4}^{\mathrm{T}}, \boldsymbol{y}_{1, 1, 2, 1}^{\mathrm{T}}, \cdots, \boldsymbol{y}_{N_1, N_2, N_3, N_4}^{\mathrm{T}} ]^{\mathrm{T}},
\end{align*}
and $\boldsymbol{y}_{n_1, n_2, n_3, n_4} \in \mathbb{C}^{K_5 \times 1}$. Accordingly, we have
\begin{align*}
\boldsymbol{y}_{n_1, n_2, n_3, n_4} = \Delta \boldsymbol{H}_{n_1, n_2, n_3, n_4} \boldsymbol{b}^*.
\end{align*}
As a result, the left hand side of \eqref{BB1} is computed as
\begin{align*}
\boldsymbol{a}^{\mathrm{H}} \Delta \boldsymbol{H} \boldsymbol{b}^* = & \sum_{n_1 = 1}^{N_1} \sum_{n_2 = 1}^{N_2} \sum_{n_3 = 1}^{N_3} \sum_{n_4 = 1}^{N_4} \boldsymbol{a}_{n_1, n_2, n_3, n_4}^{\mathrm{H}} \boldsymbol{y}_{n_1, n_2, n_3, n_4} \\
= & \sum_{n_1 = 1}^{N_1} \sum_{n_2 = 1}^{N_2} \sum_{n_3 = 1}^{N_3} \sum_{n_4 = 1}^{N_4} \sum_{k_5 = 1}^{K_5} a_{n_1, n_2, n_3, n_4, k_5}^* \\
& \hspace{8mm}\cdot \sum_{\ell_5 = 1}^{L_5} \Delta h_{n_1, n_2, n_3, n_4, k_5 + \ell_5 - 1} b_{\ell_5}^* \\
= & \sum_{n_1 = 1}^{N_1} \sum_{n_2 = 1}^{N_2} \sum_{n_3 = 1}^{N_3} \sum_{n_4 = 1}^{N_4} \sum_{m_5 = 1}^{M_5} \Delta h_{n_1, n_2, n_3, n_4, m_5} \\
& \hspace{8mm}\underbrace{\sum_{\ell_5 = 1}^{L_5} a_{n_1, n_2, n_3, n_4, m_5 + 1 - \ell_5}^* b_{\ell_5}^*}_{c_{n_1, n_2, n_3, n_4, m_5}^*} \\
= & \boldsymbol{c}^{\mathrm{H}} \Delta \boldsymbol{h}, \numberthis
\end{align*}
which complete the proof.
\end{proof}
With the results in Lemma \ref{lemB4}, we directly arrive at \eqref{eq54}.

%-----------------------------------------------------------------------------
\subsection{Proof of Lemma \ref{lem2}} \label{ProofLem2}
%The relationship between the angular frequencies and the channel parameters are given as follows.
%\begin{align*}
%\omega_{l, 1} = & \pi \sin \left( \phi_{\mathrm{az}}^l \right) \sin \left( \phi_{\mathrm{el}}^l \right), \\
%\omega_{l, 2} = & \pi \cos \left(\phi_{\mathrm{el}}^l \right), \\
%\omega_{l, 3} = & \pi \sin \left( \theta_{\mathrm{az}}^l \right) \sin \left( \theta_{\mathrm{el}}^l \right), \\
%\omega_{l, 4} = & \pi \cos \left(\theta_{\mathrm{el}}^l \right), \\
%\omega_{l, 5} = & -2\pi \Delta f \tau_l.
%\end{align*}
With the relationship between the angular frequencies and channel parameters, we establish the following first-order perturbation:
\begin{align*}
\Delta \omega_{l, 1} = & \pi \cos \left( \phi_{\mathrm{az}}^l \right) \sin \left( \phi_{\mathrm{el}}^l \right) \Delta \phi_{\mathrm{az}}^l \\
& + \pi \sin \left( \phi_{\mathrm{az}}^l \right) \cos \left( \phi_{\mathrm{el}}^l \right) \Delta \phi_{\mathrm{el}}^l, \\
\Delta \omega_{l, 2} = & - \pi \sin \left( \phi_{\mathrm{el}}^l \right) \Delta \phi_{\mathrm{el}}^l, \\
\Delta \omega_{l, 3} = & \pi \cos \left( \theta_{\mathrm{az}}^l \right) \sin \left( \theta_{\mathrm{el}}^l \right) \Delta \theta_{\mathrm{az}}^l \\
& + \pi \sin \left( \theta_{\mathrm{az}}^l \right) \cos \left( \theta_{\mathrm{el}}^l \right) \Delta \theta_{\mathrm{el}}^l, \\
\Delta \omega_{l, 4} = & - \pi \sin \left( \theta_{\mathrm{el}}^l \right) \Delta \theta_{\mathrm{el}}^l, \\
\Delta \omega_{l, 5} = & -2\pi \Delta f \Delta \tau_l.
\end{align*}
Using the results in \eqref{eq55}, it is straightforward to obtain the results in \eqref{eq56}-\eqref{eq60}.

In addition, with $\Delta \omega_{l, n}$, we have
\begin{align*}
\Delta \left( \exp \left( \jmath m_n \omega_{l, n} \right) \right) = \jmath m_n \exp \left(\jmath m_n \omega_{l, n} \right) \Delta \omega_{l, n}.
\end{align*}
Therefore, the first-order perturbation of $\boldsymbol{A}_n^{\left( M_n \right) }$ is given by
\begin{align*}
\Delta \boldsymbol{A}_n^{\left( M_n \right) } = \jmath \mathrm{Diag} \left( \boldsymbol{m}_n \right) \boldsymbol{A}_n^{\left( M_n \right) } \mathrm{Diag} \left( \Delta \boldsymbol{\omega}_n \right). \numberthis
\end{align*}
where $\boldsymbol{m}_n = [0, 1, \cdots, M_n - 1 ]^{\mathrm{T}}$ and $\Delta \boldsymbol{\omega}_n = [ \Delta \omega_{1, n}, \Delta \omega_{2, n}, \cdots, \Delta \omega_{L, n} ]^{\mathrm{T}} = \Im \{ \boldsymbol{V}_n^{\mathrm{H}} \Delta \boldsymbol{h} \}$.

The first-order perturbation of $\boldsymbol{B}_n^{\left( N_n \right) }$ is given by
\begin{align*}
\Delta \boldsymbol{B}_n^{\left( N_n \right) } = \boldsymbol{T}_n^{\mathrm{H}} \Delta \boldsymbol{A}_n^{\left( M_n \right) } = \breve{\boldsymbol{T}}_n^{\mathrm{H}} \boldsymbol{A}_n^{\left( M_n \right) } \odot \Delta \boldsymbol{\omega}_n^{\mathrm{T}}, \numberthis
\end{align*}
where 
\begin{align*}
\breve{\boldsymbol{T}}_n = - \jmath \mathrm{Diag} \left( \boldsymbol{m}_n \right) \boldsymbol{T}_n.
\end{align*}
From \eqref{eq34}, we derive the first-order perturbation $\Delta \boldsymbol{B} $ as
\begin{align*}
\Delta \boldsymbol{B} = \sum_{n=1}^5 \breve{\boldsymbol{B}}_n \odot \Delta \boldsymbol{\omega}_n^{\mathrm{T}}, \numberthis
\end{align*}
where
\begin{align*}
\breve{\boldsymbol{B}}_n = & \boldsymbol{B}_1^{\left( N_1 \right)} \odot \cdots \odot \breve{\boldsymbol{T}}_n^{\mathrm{H}} \boldsymbol{A}_n^{\left( M_n \right) }  \odot \cdots \odot \boldsymbol{B}_5^{\left( N_5 \right)}. \label{eqBC1} \numberthis
\end{align*}
According to \eqref{eqBA1}, we have the perturbation $\Delta \left( \boldsymbol{B}^{\dagger} \right)$ as
\begin{align*}
\Delta \left( \boldsymbol{B}^{\dagger} \right) = -\boldsymbol{B}^{\dagger} \Delta\boldsymbol{B} \boldsymbol{B}^{\dagger} +  \left( \boldsymbol{B}^{\mathrm{H}} \boldsymbol{B} \right)^{-1} \Delta\boldsymbol{B}^{\mathrm{H}} \left( \boldsymbol{I} - \boldsymbol{B} \boldsymbol{B}^{\dagger} \right). \numberthis
\end{align*}
As a result, the estimate of the channel gain $\boldsymbol{\gamma}$ is given as
\begin{align*}
\hat{\boldsymbol{\gamma}} = & \left( \boldsymbol{B}^{\dagger} + \Delta \left( \boldsymbol{B}^{\dagger} \right) \right) \left( \boldsymbol{B} \boldsymbol{\gamma} + \Delta \boldsymbol{h} \right) \\
= & \boldsymbol{\gamma} + \boldsymbol{B}^{\dagger} \Delta \boldsymbol{h} + \Delta \left( \boldsymbol{B}^{\dagger} \right) \boldsymbol{B} \boldsymbol{\gamma} + \Delta \left( \boldsymbol{B}^{\dagger} \right) \Delta \boldsymbol{h}, \numberthis
\end{align*}
which indicates that the first-order perturbation of the channel gain estimation is
\begin{align*}
\Delta \boldsymbol{\gamma} = & \boldsymbol{B}^{\dagger} \Delta \boldsymbol{h} + \Delta \left( \boldsymbol{B}^{\dagger} \right) \boldsymbol{B} \boldsymbol{\gamma} \\
= & \boldsymbol{B}^{\dagger} \Delta \boldsymbol{h} - \boldsymbol{B}^{\dagger}\Delta\boldsymbol{B} \boldsymbol{\gamma} + \left( \boldsymbol{B}^{\mathrm{H}} \boldsymbol{B} \right)^{-1} \Delta\boldsymbol{B}^{\mathrm{H}} \boldsymbol{B} \boldsymbol{\gamma} \\
& - \left( \boldsymbol{B}^{\mathrm{H}} \boldsymbol{B} \right)^{-1} \Delta\boldsymbol{B}^{\mathrm{H}} \boldsymbol{B} \boldsymbol{\gamma}\\
= & \boldsymbol{B}^{\dagger} \Delta \boldsymbol{h} - \boldsymbol{B}^{\dagger}\Delta\boldsymbol{B} \boldsymbol{\gamma} \\
= & \boldsymbol{B}^{\dagger} \Delta \boldsymbol{h} - \sum_{n=1}^5 \boldsymbol{\Upsilon}_n \Im \left\{ \boldsymbol{V}_n^{\mathrm{H}} \Delta \boldsymbol{h} \right\}. \numberthis
\end{align*}

%-----------------------------------------------------------------------------
% \subsection{Proof of Lemma \ref{lem3}} \label{ProofLem3}
% From \eqref{eq7}, we have the first-order perturbation $\Delta \tilde{\boldsymbol{h}} = \boldsymbol{}$

%-----------------------------------------------------------------------------
\subsection{Proof of Lemma \ref{lem3}} \label{ProofLem3}
To begin with, we have the first-order approximations of $\Delta \boldsymbol{f}_{\mathrm{T}, l}$ and $\Delta \boldsymbol{f}_{\mathrm{R}, l}$ as
\begin{align*}
\Delta \boldsymbol{f}_{\mathrm{T}, l} = & \boldsymbol{\Omega}_{\mathrm{T}, l} \left[ \Delta \phi_{\mathrm{az}}^l, \Delta \phi_{\mathrm{el}}^l \right]^{\mathrm{T}}, \numberthis \\
\Delta \boldsymbol{f}_{\mathrm{R}, l} = & \boldsymbol{\Omega}_{\mathrm{R}, l} \left[ \Delta \theta_{\mathrm{az}}^l, \Delta \theta_{\mathrm{el}}^l \right]^{\mathrm{T}}, \numberthis
\end{align*}
respectively, where $\boldsymbol{\Omega}_{\mathrm{T}, l} \in \mathbb{R}^{3\times 2}$ and $\boldsymbol{\Omega}_{\mathrm{R}, l} \in \mathbb{R}^{3\times 2}$ are given by
\begin{align*}
\boldsymbol{\Omega}_{\mathrm{T}, l} = & \left[ 
\begin{array}{cc}
-\sin \left( \phi_{\mathrm{az}}^l \right) \sin \left( \phi_{\mathrm{el}}^l \right) & \cos \left( \phi_{\mathrm{az}}^l \right) \cos \left( \phi_{\mathrm{el}}^l \right) \\
\cos \left( \phi_{\mathrm{az}}^l \right) \sin \left( \phi_{\mathrm{el}}^l \right) & \sin \left( \phi_{\mathrm{az}}^l \right) \cos \left( \phi_{\mathrm{el}}^l \right) \\
0 & - \sin \left( \phi_{\mathrm{el}}^l \right)
\end{array}
\right], \numberthis \\
\boldsymbol{\Omega}_{\mathrm{R}, l} = & \left[ 
\begin{array}{cc}
-\sin \left( \theta_{\mathrm{az}}^l \right) \sin \left( \theta_{\mathrm{el}}^l \right) & \cos \left( \theta_{\mathrm{az}}^l \right) \cos \left( \theta_{\mathrm{el}}^l \right) \\
\cos \left( \theta_{\mathrm{az}}^l \right) \sin \left( \theta_{\mathrm{el}}^l \right) & \sin \left( \theta_{\mathrm{az}}^l \right) \cos \left( \theta_{\mathrm{el}}^l \right) \\
0 & - \sin \left( \theta_{\mathrm{el}}^l \right)
\end{array}
\right]. \numberthis
\end{align*}
% \begin{align*}
% & \Delta \boldsymbol{f}_{\mathrm{T}, l} \\
% = & \left[
% \begin{array}{c}
% -\sin \left( \phi_{\mathrm{az}}^l \right) \sin \left( \phi_{\mathrm{el}}^l \right) \Delta \phi_{\mathrm{az}}^l + \cos \left( \phi_{\mathrm{az}}^l \right) \cos \left( \phi_{\mathrm{el}}^l \right) \Delta \phi_{\mathrm{el}}^l  \\
% \cos \left( \phi_{\mathrm{az}}^l \right) \sin \left( \phi_{\mathrm{el}}^l \right) \Delta \phi_{\mathrm{az}}^l + \sin \left( \phi_{\mathrm{az}}^l \right) \cos \left( \phi_{\mathrm{el}}^l \right) \Delta \phi_{\mathrm{el}}^l \\
% - \sin \left( \phi_{\mathrm{el}}^l \right) \Delta \phi_{\mathrm{el}}^l
% \end{array}
% \right], \numberthis \\
% & \Delta \boldsymbol{f}_{\mathrm{R}, l} \\
% = & \left[
% \begin{array}{c}
% -\sin \left( \theta_{\mathrm{az}}^l \right) \sin \left( \theta_{\mathrm{el}}^l \right) \Delta \theta_{\mathrm{az}}^l + \cos \left( \theta_{\mathrm{az}}^l \right) \cos \left( \theta_{\mathrm{el}}^l \right) \Delta \theta_{\mathrm{el}}^l  \\
% \cos \left( \theta_{\mathrm{az}}^l \right) \sin \left( \theta_{\mathrm{el}}^l \right) \Delta \theta_{\mathrm{az}}^l + \sin \left( \theta_{\mathrm{az}}^l \right) \cos \left( \theta_{\mathrm{el}}^l \right) \Delta \theta_{\mathrm{el}}^l \\
% - \sin \left( \theta_{\mathrm{el}}^l \right) \Delta \theta_{\mathrm{el}}^l
% \end{array}
% \right]. \numberthis
% \end{align*}
As a result, the first-order approximations of $\Delta \boldsymbol{\delta}_l$ and $\Delta \boldsymbol{\mu}_l$ are obtained as
\begin{align*}
\Delta \boldsymbol{\delta}_l = & - c \left[
\tau_l \boldsymbol{\Omega}_{\mathrm{R}, l},  \boldsymbol{f}_{\mathrm{R}, l}
\right] \left[\Delta \theta_{\mathrm{az}}^l, \Delta \theta_{\mathrm{el}}^l, \Delta \tau_l \right]^{\mathrm{T}}, \label{eqC1} \numberthis \\
\Delta \boldsymbol{\mu}_l = & c \left[\tau_l \boldsymbol{\Omega}_{\mathrm{T}, l}, \tau_l \boldsymbol{\Omega}_{\mathrm{R}, l}, \left(\boldsymbol{f}_{\mathrm{T}, l} + \boldsymbol{f}_{\mathrm{R}, l} \right) \right] \\
& \hspace{9mm}\left[\Delta \phi_{\mathrm{az}}^l, \Delta \phi_{\mathrm{el}}^l, \Delta \theta_{\mathrm{az}}^l, \Delta \theta_{\mathrm{el}}^l, \Delta \tau_l \right]^{\mathrm{T}}. \label{eqC2} \numberthis
\end{align*}

% \begin{align*}
% \Delta \boldsymbol{\delta}_l = & - c \Delta \tau_l \boldsymbol{f}_{\mathrm{R}, l} - c \tau_l \Delta \boldsymbol{f}_{\mathrm{R}, l}, \numberthis \\
% \Delta \boldsymbol{\mu}_l = & c \Delta \tau_l \left(\boldsymbol{f}_{\mathrm{T}, l} + \boldsymbol{f}_{\mathrm{R}, l} \right) + c  \tau_l \left( \Delta \boldsymbol{f}_{\mathrm{T}, l} + \Delta \boldsymbol{f}_{\mathrm{R}, l} \right). \numberthis
% \end{align*}
Since $\boldsymbol{C}_l = (\boldsymbol{I} - \boldsymbol{\mu}_l \boldsymbol{\mu}_l^{\mathrm{T}}/{ \| \boldsymbol{\mu}_l \|^2}  )$, $\Delta \boldsymbol{C}_l$ can be given as\footnote{We drop $\iota_l$ if each path is equality treated.}
% \begin{align*}
% \Delta \boldsymbol{C}_l = \frac{2\boldsymbol{\mu}_l^{\mathrm{T}} \Delta \boldsymbol{\mu}_l \left( \boldsymbol{I} - \boldsymbol{C}_l \right) - \Delta \boldsymbol{\mu}_l \boldsymbol{\mu}_l^{\mathrm{T}} - \boldsymbol{\mu}_l \Delta \boldsymbol{\mu}_l^{\mathrm{T}} }{\left\|\boldsymbol{\mu}_l\right\|^2 }. \numberthis
% \end{align*}
\begin{align*}
\Delta \boldsymbol{C}_l = & 2\left( \boldsymbol{\mu}_l^{\mathrm{T}} \boldsymbol{\mu}_l \right)^{-2} \left( \boldsymbol{\mu}_l^{\mathrm{T}} \Delta \boldsymbol{\mu}_l \right) \boldsymbol{\mu}_l \boldsymbol{\mu}_l^{\mathrm{T}} \\
& - \left( \boldsymbol{\mu}_l^{\mathrm{T}} \boldsymbol{\mu}_l \right)^{-1} \left( \Delta \boldsymbol{\mu}_l \boldsymbol{\mu}_l^{\mathrm{T}} + \boldsymbol{\mu}_l \Delta \boldsymbol{\mu}_l^{\mathrm{T}} \right). \numberthis
\end{align*}
According to \eqref{eq:PosEst}, and define $\boldsymbol{C} = \sum_{l=1}^L \boldsymbol{C}_l$, we have
\begin{align*}
\hat{\boldsymbol{p}}_{\mathrm{R}} = & \left( \sum_{l=1}^L \boldsymbol{C}_l + \Delta \boldsymbol{C}_l \right)^{-1} \sum_{l=1}^L \left( \boldsymbol{C}_l + \Delta \boldsymbol{C}_l \right) \left( \boldsymbol{\delta}_l + \Delta \boldsymbol{\delta}_l \right) \\
= & \left( \boldsymbol{I} + \sum_{l=1}^L \boldsymbol{C}^{-1} \Delta \boldsymbol{C}_l \right)^{-1} \boldsymbol{C}^{-1} \\
& \sum_{l=1}^L \boldsymbol{C}_l \boldsymbol{\delta}_l + \Delta \boldsymbol{C}_l \boldsymbol{\delta}_l + \boldsymbol{C}_l \Delta \boldsymbol{\delta}_l + \Delta \boldsymbol{C}_l \Delta \boldsymbol{\delta}_l \\
\approx & \left( \boldsymbol{I} - \sum_{l=1}^L \boldsymbol{C}^{-1} \Delta \boldsymbol{C}_l + \boldsymbol{C}^{-1} \sum_{l=1}^L \Delta \boldsymbol{C}_l \boldsymbol{C}^{-1} \sum_{l=1}^L \Delta \boldsymbol{C}_l \right) \\
& \boldsymbol{C}^{-1} \sum_{l=1}^L \boldsymbol{C}_l \boldsymbol{\delta}_l + \Delta \boldsymbol{C}_l \boldsymbol{\delta}_l + \boldsymbol{C}_l \Delta \boldsymbol{\delta}_l + \Delta \boldsymbol{C}_l \Delta \boldsymbol{\delta}_l \\
\approx & \boldsymbol{p}_{\mathrm{R}} + \boldsymbol{C}^{-1} \sum_{l=1}^L \Delta \boldsymbol{C}_l \left( \boldsymbol{\delta}_l - \boldsymbol{p}_{\mathrm{R}} \right) + \boldsymbol{C}^{-1} \sum_{l=1}^L \boldsymbol{C}_l \Delta \boldsymbol{\delta}_l \\
& + \boldsymbol{C}^{-1} \sum_{l=1}^L \Delta \boldsymbol{C}_l \Delta \boldsymbol{\delta}_l - \boldsymbol{C}^{-1} \sum_{l=1}^L \Delta \boldsymbol{C}_l \boldsymbol{C}^{-1} \sum_{l=1}^L \boldsymbol{C}_l \Delta \boldsymbol{\delta}_l \\
& + \boldsymbol{C}^{-1} \sum_{l=1}^L \Delta \boldsymbol{C}_l \boldsymbol{C}^{-1} \sum_{l=1}^L \Delta \boldsymbol{C}_l \left( \boldsymbol{p}_{\mathrm{R}} - \boldsymbol{\delta}_l\right), \label{eqC2nd} \numberthis
\end{align*}
where higher order approximation is omitted. Therefore, the first-order perturbation of the position estimate, i.e., $\Delta \boldsymbol{p}_{\mathrm{R}} = \hat{\boldsymbol{p}}_{\mathrm{R}} - \boldsymbol{p}_{\mathrm{R}}$, can be obtained as
\begin{align*}
\Delta \boldsymbol{p}_{\mathrm{R}} = &  \boldsymbol{C}^{-1} \sum_{l=1}^L \Delta \boldsymbol{C}_l \left( \boldsymbol{\delta}_l - \boldsymbol{p}_{\mathrm{R}} \right) + \boldsymbol{C}^{-1} \sum_{l=1}^L \boldsymbol{C}_l \Delta \boldsymbol{\delta}_l \\
& = \boldsymbol{C}^{-1} \sum_{l=1}^L \breve{\boldsymbol{C}}_l \Delta \boldsymbol{\mu}_l + \boldsymbol{C}_l \Delta \boldsymbol{\delta}_l, \label{eqC3} \numberthis
\end{align*}
where $\breve{\boldsymbol{C}}_l \in \mathbb{R}^{3\times 3}$ is given by
\begin{align*}
\breve{\boldsymbol{C}}_l = & \frac{2\boldsymbol{\mu}_l^{\mathrm{T}} \left( \boldsymbol{\delta}_l - \boldsymbol{p}_{\mathrm{R}} \right) \boldsymbol{\mu}_l \boldsymbol{\mu}_l^{\mathrm{T}}}{\left\|\boldsymbol{\mu}_l\right\|^4} \\
& - \frac{\boldsymbol{\mu}_l^{\mathrm{T}} \left( \boldsymbol{\delta}_l - \boldsymbol{p}_{\mathrm{R}} \right) \boldsymbol{I} + \boldsymbol{\mu}_l \left( \boldsymbol{\delta}_l - \boldsymbol{p}_{\mathrm{R}} \right)^{\mathrm{T}} }{\left\|\boldsymbol{\mu}_l\right\|^2}. \numberthis
\end{align*}
% \begin{align*}
% \breve{\boldsymbol{C}}_l = & \frac{ 2 \left( \boldsymbol{I} - \boldsymbol{C}_l \right) \left( \boldsymbol{\delta}_l - \boldsymbol{p}_{\mathrm{R}} \right)  \boldsymbol{\mu}_l^{\mathrm{T}} - \boldsymbol{\mu}_l^{\mathrm{T}} \left( \boldsymbol{\delta}_l - \boldsymbol{p}_{\mathrm{R}} \right) \boldsymbol{I} }{\left\|\boldsymbol{\mu}_l\right\|^2} \\
% & - \frac{ \boldsymbol{\mu}_l \left( \boldsymbol{\delta}_l - \boldsymbol{p}_{\mathrm{R}} \right)^{\mathrm{T}} }{\left\|\boldsymbol{\mu}_l\right\|^2}. \numberthis
% \end{align*}
Substituting the results in \eqref{eqC1} and \eqref{eqC2} to \eqref{eqC3}, we arrive at 
\begin{align*}
\Delta \boldsymbol{p}_{\mathrm{R}} = \sum_{l=1}^L & \breve{\boldsymbol{D}}_l \left[\Delta \theta_{\mathrm{az}}^l, \Delta \theta_{\mathrm{el}}^l, \Delta \tau_l \right]^{\mathrm{T}} \\
& + \breve{\boldsymbol{E}}_l \left[\Delta \phi_{\mathrm{az}}^l, \Delta \phi_{\mathrm{el}}^l, \Delta \theta_{\mathrm{az}}^l, \Delta \theta_{\mathrm{el}}^l, \Delta \tau_l \right]^{\mathrm{T}}, \numberthis
\end{align*}
where $\breve{\boldsymbol{D}}_l \in \mathbb{R}^{3\times 3}$ and $\breve{\boldsymbol{E}}_l\in \mathbb{R}^{3\times 5}$
\begin{align*}
\breve{\boldsymbol{D}}_l = & -c \left(\sum_{\ell=1}^L \boldsymbol{C}_\ell \right)^{-1} \boldsymbol{C}_l \left[
\tau_l \boldsymbol{\Omega}_{\mathrm{R}, l},  \boldsymbol{f}_{\mathrm{R}, l}
\right], \label{eqC4} \numberthis \\
\breve{\boldsymbol{E}}_l = & c \left(\sum_{\ell=1}^L \boldsymbol{C}_\ell \right)^{-1} \breve{\boldsymbol{C}}_l \left[\tau_l \boldsymbol{\Omega}_{\mathrm{T}, l}, \tau_l \boldsymbol{\Omega}_{\mathrm{R}, l}, \left(\boldsymbol{f}_{\mathrm{T}, l} + \boldsymbol{f}_{\mathrm{R}, l} \right) \right]. \label{eqC5} \numberthis
\end{align*}
With the results in Lemma \ref{lem2}, we derive the results in \eqref{eq75}.

\end{document}

%% file: packages.tex
\usepackage{pgfplots}
\usetikzlibrary{shapes.multipart,intersections}
\usepackage{cite}
\usepackage{amsmath,amssymb,amsfonts,steinmetz}
\usepackage{mathtools,algpseudocode,algorithm,MnSymbol}
\usepackage{graphicx}
\usepackage{textcomp}
\usepackage{xcolor}
\def\BibTeX{{\rm B\kern-.05em{\sc i\kern-.025em b}\kern-.08em
    T\kern-.1667em\lower.7ex\hbox{E}\kern-.125emX}}

\usepackage{tikz}
\usetikzlibrary{arrows.meta}
\usetikzlibrary{calc}
\makeatletter
\newcommand{\gettikzxy}[3]{%
  \tikz@scan@one@point\pgfutil@firstofone#1\relax
  \edef#2{\the\pgf@x}%
  \edef#3{\the\pgf@y}%
}
\makeatother
\usepackage[draft]{todonotes}   % notes showed
\usepackage{upgreek} %use for mathbf symbols like \bm{theta} = \mathbf{\uptheta}
\usepackage{seqsplit}
\usepackage{hyperref}
\usepackage{balance}
\usepackage[font={small}]{caption}
\usepackage{subcaption}

%% file: Fig/3DSetup.tex
\begin{tikzpicture}
    \node(image) [anchor=south west] at (0,0){\includegraphics[width=\textwidth]{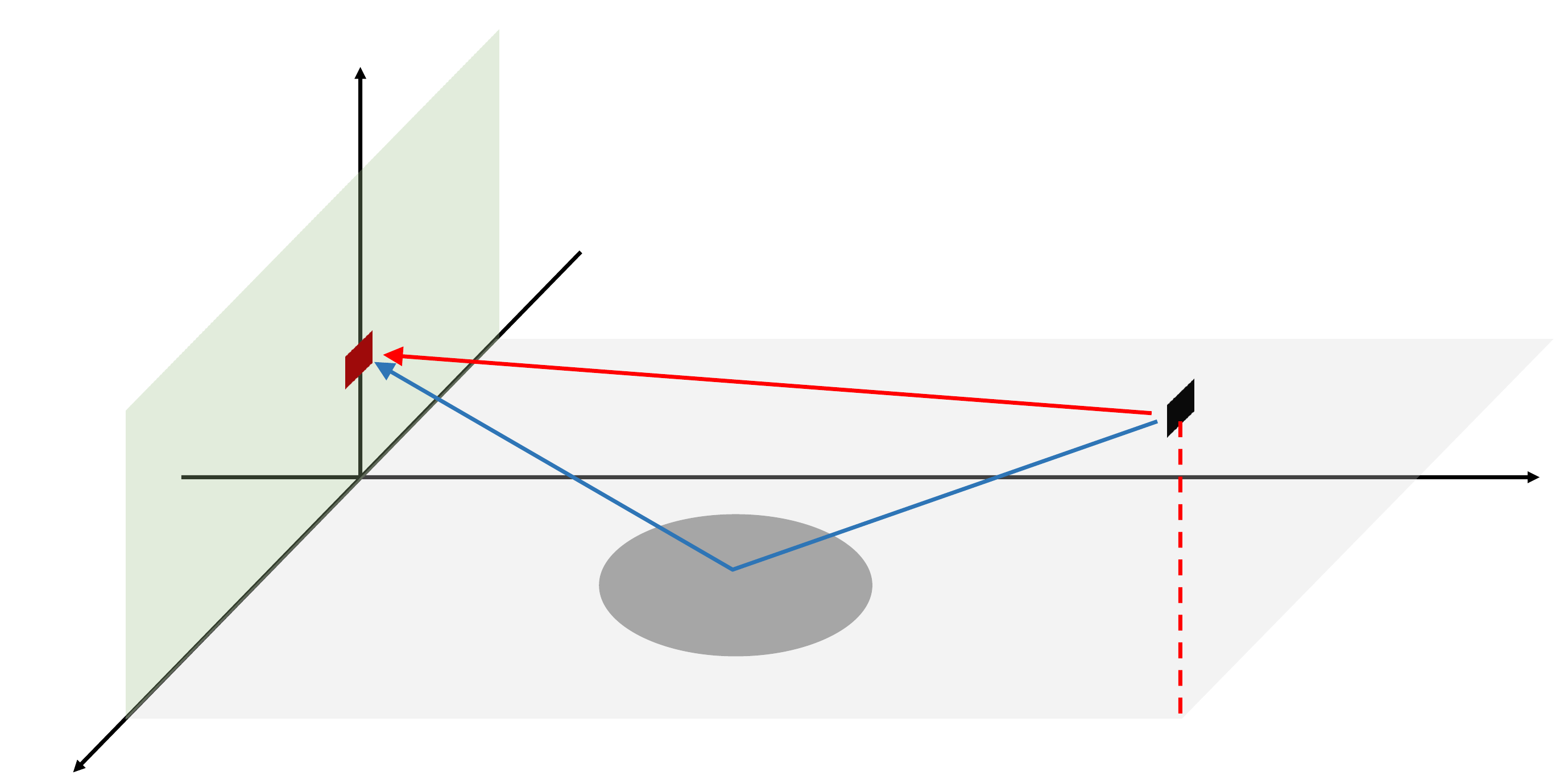}};
    % \draw [help lines] (0,0) grid (9,5);
    \node at (9, 2.1) {$\boldsymbol{x}$};
    \node at (0.5, 0) {$\boldsymbol{y}$};
    \node at (2.2, 4.5) {$\boldsymbol{z}$};
    \node at (2.3, 1.7) {0};
    \node at (8.4, 1.7) {20};
    \node at (0.95, 0.3) {5};
    \node at (7, 2.8) {$\boldsymbol{p}_{\mathrm{T}}$};
    \node at (2.5, 2.8) {$\boldsymbol{p}_{\mathrm{R}}$};
    \node at (4.4, 1.7) {$\boldsymbol{p}_{\mathrm{S}}$};
    \node at (4.4, 1.1) {$\mathcal{S}_{\mathrm{S}}$};
    % \node at (1.8, 2.5) {$\mathcal{S}_{\mathrm{R}}$};
\end{tikzpicture}

%% file: Fig/LegendFig3.tex
\begin{tikzpicture}
%\draw [help lines] (0,0) grid (9,2);
\node (Box1) [draw=none,fill=none] at (0.8, 0.3) {\shortstack[l]{ \ref{Num_LOS_DFT} \footnotesize Simulations: LOS, DFT \\
\ref{Ana_LOS_DFT} \footnotesize Analytical: LOS, DFT}};
\node (Box2) [draw=none,fill=none] at (5, 0.3) {\shortstack[l]{ \ref{Num_NLOS_DFT} \footnotesize Simulations: NLOS, DFT \\
\ref{Ana_NLOS_DFT} \footnotesize Analytical: NLOS, DFT}};
\node (Box3) [draw=none,fill=none] at (9.5, 0.3) {\shortstack[l]{ \ref{Num_LOS_Sum} \footnotesize Simulations: LOS, directional \\
\ref{Ana_LOS_Sum} \footnotesize Analytical: LOS, directional}};
\node (Box1) [draw=none,fill=none] at (14.3, 0.3) {\shortstack[l]{ \ref{Num_NLOS_Sum} \footnotesize Simulations: NLOS, directional \\
\ref{Ana_NLOS_Sum} \footnotesize Analytical: NLOS, directional}};
\end{tikzpicture}%

%% file: Fig/Angles.tex
\begin{tikzpicture}

\begin{axis}[%
width=0.84\textwidth,
height=0.84\textwidth,
at={(2,0.5)},
scale only axis,
xmin=-10,
xmax=40,
xlabel style={font=\color{white!15!black}},
xlabel={SNR (dB)},
ymode=log,
ymin=4e-06,
ymax=0.05,
yminorticks=true,
axis background/.style={fill=white},
ylabel style={font=\color{white!15!black}},
ylabel={RMSE(rad)},
xmajorgrids,
ymajorgrids,
yminorgrids
]
\addplot [color=blue, line width=1.0pt, mark size=1.5pt, mark=square*, mark options={solid, fill=blue, draw=blue}, forget plot]
  table[row sep=crcr]{%
-10	0.0166289266099174\\
0	0.00295839492121321\\
10	0.000797340834262308\\
20	0.000246063191287389\\
30	7.81662726061573e-05\\
40	2.47420588238407e-05\\
};\label{Num_LOS_DFT}
\addplot [color=blue, line width=1.0pt, forget plot]
  table[row sep=crcr]{%
-10	0.00774284794848202\\
0	0.00244850350935653\\
10	0.000774284794848202\\
20	0.000244850350935653\\
30	7.74284794848202e-05\\
40	2.44850350935653e-05\\
};\label{Ana_LOS_DFT}
\addplot [color=blue, dashed, line width=1.0pt, mark size=1.5pt, mark=*, mark options={solid, fill=blue, draw=blue}, forget plot]
  table[row sep=crcr]{%
-10	0.0302595995679146\\
0	0.00715003646732065\\
10	0.00208572757603315\\
20	0.000647611155392436\\
30	0.000204293634532548\\
40	6.48844360249945e-05\\
};\label{Num_NLOS_DFT}
\addplot [color=blue, dashed, line width=1.0pt, forget plot]
  table[row sep=crcr]{%
-10	0.0206018221036155\\
0	0.00651486817970265\\
10	0.00206018221036155\\
20	0.000651486817970265\\
30	0.000206018221036155\\
40	6.51486817970265e-05\\
};\label{Ana_NLOS_DFT}
\addplot [color=red, line width=1.0pt, mark size=1.5pt, mark=square*, mark options={solid, fill=red, draw=red}, forget plot]
  table[row sep=crcr]{%
-10	0.00182915309024986\\
0	0.000540336652250406\\
10	0.000170366487904021\\
20	5.29929323718008e-05\\
30	1.66912637967335e-05\\
40	5.43469992646086e-06\\
};\label{Num_LOS_Sum}
\addplot [color=red, line width=1.0pt, forget plot]
  table[row sep=crcr]{%
-10	0.00168629335935188\\
0	0.000533252781876873\\
10	0.000168629335935188\\
20	5.33252781876873e-05\\
30	1.68629335935188e-05\\
40	5.33252781876873e-06\\
};\label{Ana_LOS_Sum}
\addplot [color=red, dashed, line width=1.0pt, mark size=1.5pt, mark=*, mark options={solid, fill=red, draw=red}, forget plot]
  table[row sep=crcr]{%
-10	0.0212827764359188\\
0	0.0037064950973023\\
10	0.000992540781052765\\
20	0.000304238641436052\\
30	9.82473885808755e-05\\
40	3.12237025077887e-05\\
};\label{Num_NLOS_Sum}
\addplot [color=red, dashed, line width=1.0pt, forget plot]
  table[row sep=crcr]{%
-10	0.00975113773146095\\
0	0.00308358050094239\\
10	0.000975113773146095\\
20	0.000308358050094239\\
30	9.75113773146095e-05\\
40	3.08358050094239e-05\\
};\label{Ana_NLOS_Sum}
\end{axis}
%\node (Box1) [draw=none,fill=none] at (2.7, 4.35) {\shortstack[l]{ \ref{Num_LOS_DFT} \footnotesize Simulations: LOS, DFT \\
%\ref{Ana_LOS_DFT} \footnotesize Analytical: LOS, DFT}};
%\node (Box1) [draw=none,fill=none] at (1.9, 0.5) {\shortstack[l]{ \ref{Num_NLOS_DFT} \footnotesize Simulations: NLOS, DFT \\
%\ref{Ana_NLOS_DFT} \footnotesize Analytical: NLOS, DFT}};
\end{tikzpicture}%

%% file: Fig/Delay.tex
\begin{tikzpicture}

\begin{axis}[%
width=0.84\textwidth,
height=0.84\textwidth,
at={(2,0.5)},
scale only axis,
xmin=-10,
xmax=40,
xlabel style={font=\color{white!15!black}},
xlabel={SNR (dB)},
ymode=log,
ymin=2e-06,
ymax=0.03,
yminorticks=true,
ylabel style={font=\color{white!15!black}},
ylabel={RMSE (m)},
axis background/.style={fill=white},
xmajorgrids,
ymajorgrids,
yminorgrids
]
\addplot [color=blue, line width=1.0pt, mark size=1.5pt, mark=square*, mark options={solid, fill=blue, draw=blue}, forget plot]
  table[row sep=crcr]{%
-10	0.00393704403896897\\
0	0.000419774640808397\\
10	0.000137972563597526\\
20	4.17633166401588e-05\\
30	1.38109483951015e-05\\
40	4.31550169921436e-06\\
};
\addplot [color=blue, line width=1.0pt, forget plot]
  table[row sep=crcr]{%
-10	0.00135494631301014\\
0	0.000428471645635948\\
10	0.000135494631301014\\
20	4.28471645635948e-05\\
30	1.35494631301014e-05\\
40	4.28471645635948e-06\\
};
\addplot [color=blue, dashed, line width=1.0pt, mark size=1.5pt, mark=*, mark options={solid, fill=blue, draw=blue}, forget plot]
  table[row sep=crcr]{%
-10	0.00729998321717056\\
0	0.00175137303928206\\
10	0.000582412191078538\\
20	0.000176236824520099\\
30	5.73601880702922e-05\\
40	1.79809224156489e-05\\
};
\addplot [color=blue, dashed, line width=1.0pt, forget plot]
  table[row sep=crcr]{%
-10	0.00573455825175701\\
0	0.00181342654504654\\
10	0.000573455825175701\\
20	0.000181342654504654\\
30	5.73455825175701e-05\\
40	1.81342654504654e-05\\
};
\addplot [color=red, line width=1.0pt, mark size=1.5pt, mark=square*, mark options={solid, fill=red, draw=red}, forget plot]
  table[row sep=crcr]{%
-10	0.00072729248971496\\
0	0.000217732081860057\\
10	6.80288244564833e-05\\
20	2.24289559688516e-05\\
30	7.2997027290757e-06\\
40	2.20819107902246e-06\\
};
\addplot [color=red, line width=1.0pt, forget plot]
  table[row sep=crcr]{%
-10	0.000708398311151128\\
0	0.000224015215385422\\
10	7.08398311151128e-05\\
20	2.24015215385422e-05\\
30	7.08398311151128e-06\\
40	2.24015215385422e-06\\
};
\addplot [color=red, dashed, line width=1.0pt, mark size=1.5pt, mark=*, mark options={solid, fill=red, draw=red}, forget plot]
  table[row sep=crcr]{%
-10	0.00342970829955676\\
0	0.0010714836639038\\
10	0.000334278231183624\\
20	0.000108887151228694\\
30	3.36911195634681e-05\\
40	1.03766742685901e-05\\
};
\addplot [color=red, dashed, line width=1.0pt, forget plot]
  table[row sep=crcr]{%
-10	0.00335273710130492\\
0	0.00106022856358743\\
10	0.000335273710130492\\
20	0.000106022856358743\\
30	3.35273710130492e-05\\
40	1.06022856358743e-05\\
};
\end{axis}
% \node (Box1) [draw=none,fill=none] at (2.3, 4.35) {\shortstack[l]{ \ref{Num_LOS_Sum} \footnotesize Simulations: LOS, directional \\
%\ref{Ana_LOS_Sum} \footnotesize Analytical: LOS, directional}};
\end{tikzpicture}%

%% file: Fig/CG.tex
\begin{tikzpicture}

\begin{axis}[%
width=0.84\textwidth,
height=0.84\textwidth,
at={(2,0.5)},
scale only axis,
xmin=-10,
xmax=40,
xlabel style={font=\color{white!15!black}},
xlabel={SNR (dB)},
ymode=log,
ymin=1e-09,
ymax=2e-05,
yminorticks=true,
ylabel style={font=\color{white!15!black}},
ylabel={RMSE},
axis background/.style={fill=white},
xmajorgrids,
ymajorgrids,
yminorgrids
]
\addplot [color=blue, line width=1.0pt, mark size=1.5pt, mark=square*, mark options={solid, fill=blue, draw=blue}, forget plot]
  table[row sep=crcr]{%
-10	5.28754164111112e-06\\
0	7.46718988704069e-07\\
10	1.88530896399381e-07\\
20	5.63156355405104e-08\\
30	1.76980060728802e-08\\
40	5.62200799836243e-09\\
};
\addplot [color=blue, line width=1.0pt, forget plot]
  table[row sep=crcr]{%
-10	1.77351838459984e-06\\
0	5.60835756751798e-07\\
10	1.77351838459984e-07\\
20	5.60835756751798e-08\\
30	1.77351838459984e-08\\
40	5.60835756751798e-09\\
};
\addplot [color=blue, dashed, line width=1.0pt, mark size=1.5pt, mark=*, mark options={solid, fill=blue, draw=blue}, forget plot]
  table[row sep=crcr]{%
-10	1.06315994560941e-06\\
0	2.76672561618466e-07\\
10	8.43394785955596e-08\\
20	2.68153744648724e-08\\
30	8.56562791167274e-09\\
40	2.71931213966723e-09\\
};
\addplot [color=blue, dashed, line width=1.0pt, forget plot]
  table[row sep=crcr]{%
-10	8.54236788583744e-07\\
0	2.70133391303235e-07\\
10	8.54236788583744e-08\\
20	2.70133391303235e-08\\
30	8.54236788583744e-09\\
40	2.70133391303235e-09\\
};
\addplot [color=red, line width=1.0pt, mark size=1.5pt, mark=square*, mark options={solid, fill=red, draw=red}, forget plot]
  table[row sep=crcr]{%
-10	4.75945946112396e-07\\
0	1.24563442443699e-07\\
10	3.85033644272191e-08\\
20	1.21922941963722e-08\\
30	3.94905777591874e-09\\
40	1.24627540434237e-09\\
};
\addplot [color=red, line width=1.0pt, forget plot]
  table[row sep=crcr]{%
-10	3.86059202196754e-07\\
0	1.22082639060922e-07\\
10	3.86059202196754e-08\\
20	1.22082639060922e-08\\
30	3.86059202196754e-09\\
40	1.22082639060922e-09\\
};
\addplot [color=red, dashed, line width=1.0pt, mark size=1.5pt, mark=*, mark options={solid, fill=red, draw=red}, forget plot]
  table[row sep=crcr]{%
-10	1.1146776820333e-06\\
0	1.74683807967667e-07\\
10	4.47194128344663e-08\\
20	1.27968008562644e-08\\
30	4.45060546506217e-09\\
40	1.37430544363747e-09\\
};
\addplot [color=red, dashed, line width=1.0pt, forget plot]
  table[row sep=crcr]{%
-10	4.32636398959421e-07\\
0	1.36811641940507e-07\\
10	4.32636398959421e-08\\
20	1.36811641940507e-08\\
30	4.32636398959421e-09\\
40	1.36811641940507e-09\\
};
\end{axis}
%\node (Box1) [draw=none,fill=none] at (2.3, 4.35) {\shortstack[l]{ \ref{Num_NLOS_Sum} \footnotesize Simulations: NLOS, directional \\
%\ref{Ana_NLOS_Sum} \footnotesize Analytical: NLOS, directional}};
\end{tikzpicture}%

%% file: Fig/Localization.tex
\begin{tikzpicture}

\begin{axis}[%
width=0.8\textwidth,
height=0.3\textwidth,
at={(2,0.5)},
scale only axis,
xmin=-10,
xmax=40,
xlabel style={font=\color{white!15!black}},
xlabel={SNR (dB)},
ymin=0.00005,
ymax=2,
ymode=log,
ylabel style={font=\color{white!15!black}},
ylabel={RMSE (m)},
axis background/.style={fill=white},
xmajorgrids,
ymajorgrids
]
\addplot [color=blue, dashdotted, line width=1.0pt, mark size=1.5pt, mark=square*, mark options={solid, fill=blue, draw=blue}, forget plot]
  table[row sep=crcr]{%
-10	1.075765464132\\
0	0.149639362389055\\
10	0.0412247494097862\\
20	0.0119399513438065\\
30	0.0037350028629964\\
40	0.00122952461557811\\
}; \label{Num_Pos_60MHz}
\addplot [color=blue, line width=1.0pt, forget plot]
  table[row sep=crcr]{%
-10	0.0917737815410388\\
0	0.02902141791564\\
10	0.00917737815410388\\
20	0.002902141791564\\
30	0.000917737815410388\\
40	0.0002902141791564\\
}; \label{Ana_Pos_60MHz}
\addplot [color=red, dashdotted, line width=1.0pt, mark size=1.5pt, mark=square*, mark options={solid, fill=red, draw=red}, forget plot]
  table[row sep=crcr]{%
-10	0.486932525966813\\
0	0.0731994029192112\\
10	0.0166380200963032\\
20	0.00480806716565456\\
30	0.00157404663102637\\
40	0.000541172423780855\\
};\label{Num_Pos_Sum_60MHz}
\addplot [color=red, line width=1.0pt, forget plot]
  table[row sep=crcr]{%
-10	0.023107490116249\\
0	0.00730722997771758\\
10	0.0023107490116249\\
20	0.000730722997771758\\
30	0.00023107490116249\\
40	7.30722997771758e-05\\
};\label{Ana_Pos_Sum_60MHz}
\end{axis}
\node (Box1) [draw=none,fill=none] at (5.4, 2.3) {\shortstack[l]{\ref{Num_Pos_60MHz} \footnotesize Simulations: DFT \\
\ref{Ana_Pos_60MHz} \footnotesize Analytical: DFT}};
\node (Box1) [draw=none,fill=none] at (2, 0.5) {\shortstack[l]{ \ref{Num_Pos_Sum_60MHz} \footnotesize Simulations: directional \\
\ref{Ana_Pos_Sum_60MHz} \footnotesize Analytical: directional }};
\end{tikzpicture}%

%% file: Fig/SumRate.tex
\begin{tikzpicture}

\begin{axis}[%
width=0.8\textwidth,
height=0.3\textwidth,
at={(2,0.5)},
scale only axis,
xmin=-10,
xmax=40,
xlabel style={font=\color{white!15!black}},
xlabel={SNR (dB)},
ymin=2,
ymax=20,
ylabel style={font=\color{white!15!black}},
ylabel={Achievable rate (bit/s/Hz)},
axis background/.style={fill=white},
xmajorgrids,
ymajorgrids
]
\addplot [color=blue, dashdotted, line width=1.0pt, mark size=1.5pt, mark=square*, mark options={solid, fill=blue, draw=blue}, forget plot]
  table[row sep=crcr]{%
-10	2.40783975009046\\
-5	3.83626676748899\\
0	5.35736009099592\\
5	6.91417172930717\\
10	8.47841007281352\\
15	10.0507518066472\\
20	11.6220546287136\\
25	13.1939277157471\\
30	14.7676355663213\\
35	16.3394923427042\\
40	17.9124452799092\\
}; \label{SR_DFT_M}
\addplot [color=blue, line width=1.0pt, forget plot]
  table[row sep=crcr]{%
-10	2.45210256407635\\
-5	3.85987208544347\\
0	5.37577855441752\\
5	6.92980240782167\\
10	8.49632548435801\\
15	10.0668476443501\\
20	11.638639137953\\
25	13.2108325033865\\
30	14.7831529992465\\
35	16.3555137020189\\
40	17.9278871198076\\
}; \label{SR_DFT_P}
\addplot [color=red, dashdotted, line width=1.0pt, mark size=1.5pt, mark=square*, mark options={solid, fill=red, draw=red}, forget plot]
  table[row sep=crcr]{%
-10	4.02693832446251\\
-5	5.55399731171515\\
0	7.11064704305106\\
5	8.68011258033255\\
10	10.2503348454655\\
15	11.8231272647636\\
20	13.3942925616287\\
25	14.9651945526848\\
30	16.5396587033713\\
35	18.1130977865722\\
40	19.6832811362986\\
}; \label{SR_SUM_M}
\addplot [color=red, line width=1.0pt, forget plot]
  table[row sep=crcr]{%
-10	4.04902118273021\\
-5	5.57236440026958\\
0	7.1288634148343\\
5	8.6961822529165\\
10	10.2669573829626\\
15	11.8388290062922\\
20	13.4110477243514\\
25	14.9833762387545\\
30	16.5557394773468\\
35	18.1281136970456\\
40	19.7004913893104\\
}; \label{SR_SUM_P}
\end{axis}

%\begin{axis}[%
%width=0.2\textwidth,
%height=0.3\textwidth,
%at={(0.57\textwidth,0.06\textwidth)},
%scale only axis,
%xmin=-0.2,
%xmax=0.2,
%ymin=5.25,
%ymax=7.25,
%axis background/.style={fill=white},
%xmajorgrids,
%ymajorgrids
%]
%\addplot [color=blue, dashdotted, line width=1.0pt, mark size=1.5pt, mark=square*, mark options={solid, fill=blue, draw=blue}, forget plot]
%  table[row sep=crcr]{%
%-5	3.83626676748899\\
%0	5.35736009099592\\
%5	6.91417172930717\\
%};
%\addplot [color=blue, line width=1.0pt, forget plot]
%  table[row sep=crcr]{%
%-5	3.85987208544347\\
%0	5.37577855441752\\
%5	6.92980240782167\\
%};
%\addplot [color=red, dashdotted, line width=1.0pt, mark size=1.5pt, mark=square*, mark options={solid, fill=red, draw=red}, forget plot]
%  table[row sep=crcr]{%
%-5	5.55399731171515\\
%0	7.11064704305106\\
%5	8.68011258033255\\
%};
%\addplot [color=red, line width=1.0pt, forget plot]
%  table[row sep=crcr]{%
%-5	5.57236440026958\\
%0	7.1288634148343\\
%5	8.6961822529165\\
%};
%\end{axis}
%\draw [help lines] (0,0) grid (8, 5);
\node (Box1) [draw=none,fill=none] at (1.6, 2.2) {\shortstack[l]{ \ref{SR_DFT_P} \footnotesize Perfect CSI: DFT \\
\ref{SR_DFT_M} \footnotesize Proposed: DFT}}; 
\node (Box1) [draw=none,fill=none] at (5.1, 0.4) {\shortstack[l]{\ref{SR_SUM_P} \footnotesize Perfect CSI: directional \\
\ref{SR_SUM_M} \footnotesize Proposed: directional}}; 
\end{tikzpicture}%

%% file: Fig/AFCMP.tex
\begin{tikzpicture}

\begin{axis}[%
width=0.84\textwidth,
height=0.84\textwidth,
at={(2,0.5)},
scale only axis,
xmin=-10,
xmax=40,
xlabel style={font=\color{white!15!black}},
xlabel={SNR (dB)},
ymode=log,
ymin=1e-05,
ymax=2,
yminorticks=true,
ylabel style={font=\color{white!15!black}},
ylabel={RMSE (rad)},
axis background/.style={fill=white},
xmajorgrids,
ymajorgrids,
yminorgrids
]
\addplot [color=blue, dotted, line width=1.0pt, mark size=1.5pt, mark=square*, mark options={solid, fill=blue, draw=blue}, forget plot]
  table[row sep=crcr]{%
-10	0.457697799767582\\
0	0.451332851762914\\
10	0.283590415222115\\
20	0.0678996360111299\\
30	0.0207054944363707\\
40	0.00687745561736842\\
}; \label{AFT_DFT}
\addplot [color=red, dotted, line width=1.0pt, mark size=1.5pt, mark=square*, mark options={solid, fill=red, draw=red}, forget plot]
  table[row sep=crcr]{%
-10	0.01787202133827\\
0	0.00354202977000938\\
10	0.00100501319874493\\
20	0.000309951418157199\\
30	9.77369968573625e-05\\
40	3.11914010475277e-05\\
}; \label{AFM_DFT}
\addplot [color=blue, line width=1.0pt, mark size=1.5pt, mark=square*, mark options={solid, fill=blue, draw=blue}, forget plot]
  table[row sep=crcr]{%
-10	0.199708948856902\\
0	0.1047144555302\\
10	0.0204652733899348\\
20	0.00653541120180269\\
30	0.00210815812210959\\
40	0.000645751645029958\\
}; \label{AFT_Sum}
\addplot [color=red, line width=1.0pt, mark size=1.5pt, mark=square*, mark options={solid, fill=red, draw=red}, forget plot]
  table[row sep=crcr]{%
-10	0.0110658117805005\\
0	0.00175937896913574\\
10	0.000463412731567352\\
20	0.000141323993374912\\
30	4.56154913229398e-05\\
40	1.44523506492468e-05\\
}; \label{AFM_Sum}
\end{axis}
% \draw [help lines] (0,0) grid (7, 6);
\node (Box1) [draw=none,fill=none] at (3, 4.35) {\shortstack[l]{ \ref{AFT_DFT} \footnotesize Tensor: DFT \\
\ref{AFT_Sum} \footnotesize Tensor: directional}};
\node (Box2) [draw=none,fill=none] at (1.75, 0.55) {\shortstack[l]{ \ref{AFM_DFT} \footnotesize Proposed: DFT \\
\ref{AFM_Sum} \footnotesize Proposed: directional}};
\end{tikzpicture}%

%% file: Fig/CHCMP.tex
\begin{tikzpicture}

\begin{axis}[%
width=0.84\textwidth,
height=0.84\textwidth,
at={(2,0.5)},
scale only axis,
xmin=-10,
xmax=40,
xlabel style={font=\color{white!15!black}},
xlabel={SNR (dB)},
ymode=log,
ymin=1e-09,
ymax=0.0001,
yminorticks=true,
ylabel style={font=\color{white!15!black}},
ylabel={RMSE},
axis background/.style={fill=white},
xmajorgrids,
ymajorgrids,
yminorgrids
]
\addplot [color=blue, dotted, line width=1.0pt, mark size=1.5pt, mark=square*, mark options={solid, fill=blue, draw=blue}, forget plot]
  table[row sep=crcr]{%
-10	3.59190575606405e-05\\
0	1.95913571856364e-05\\
10	1.69867188565612e-05\\
20	3.98903720258639e-06\\
30	1.38290506515947e-06\\
40	5.07646479440679e-07\\
};
\addplot [color=red, dotted, line width=1.0pt, mark size=1.5pt, mark=square*, mark options={solid, fill=red, draw=red}, forget plot]
  table[row sep=crcr]{%
-10	3.31884511938869e-06\\
0	5.13652615762445e-07\\
10	1.336072351043e-07\\
20	4.11067600701878e-08\\
30	1.30101999567375e-08\\
40	4.13242972442655e-09\\
};
\addplot [color=blue, line width=1.0pt, mark size=1.5pt, mark=square*, mark options={solid, fill=blue, draw=blue}, forget plot]
  table[row sep=crcr]{%
-10	7.40836167054411e-06\\
0	4.61909783405366e-06\\
10	1.30503028631667e-06\\
20	4.49135644016814e-07\\
30	1.75053118575867e-07\\
40	5.47061330618007e-08\\
};
\addplot [color=red, line width=1.0pt, mark size=1.5pt, mark=square*, mark options={solid, fill=red, draw=red}, forget plot]
  table[row sep=crcr]{%
-10	7.31654961996003e-07\\
0	1.37923757958165e-07\\
10	3.93954114056924e-08\\
20	1.21562911446767e-08\\
30	3.8776720347882e-09\\
40	1.23575406810249e-09\\
};
\end{axis}
\node (Box1) [draw=none,fill=none] at (3, 4.35) {\shortstack[l]{ \ref{AFT_DFT} \footnotesize Tensor: DFT \\
\ref{AFT_Sum} \footnotesize Tensor: directional}};
\node (Box2) [draw=none,fill=none] at (1.75, 0.55) {\shortstack[l]{ \ref{AFM_DFT} \footnotesize Proposed: DFT \\
\ref{AFM_Sum} \footnotesize Proposed: directional}};
\end{tikzpicture}%

%% file: Fig/PosCMP.tex
\begin{tikzpicture}

\begin{axis}[%
width=0.84\textwidth,
height=0.84\textwidth,
at={(2,0.5)},
scale only axis,
xmin=-10,
xmax=40,
xlabel style={font=\color{white!15!black}},
xlabel={SNR (dB)},
ymin=0.0005,
ymax=20,
ymode=log,
ylabel style={font=\color{white!15!black}},
ylabel={RMSE (m)},
axis background/.style={fill=white},
xmajorgrids,
ymajorgrids
]
\addplot [color=blue, dotted, line width=1.0pt, mark size=1.5pt, mark=square*, mark options={solid, fill=blue, draw=blue}, forget plot]
  table[row sep=crcr]{%
-10	5.7289439583794\\
0	5.75206674794225\\
10	4.1297974344296\\
20	1.31549457200253\\
30	0.377569035507074\\
40	0.0791274695435019\\
};\label{POST_DFT}
\addplot [color=red, dotted, line width=1.0pt, mark size=1.5pt, mark=square*, mark options={solid, fill=red, draw=red}, forget plot]
  table[row sep=crcr]{%
-10	1.075765464132\\
0	0.149639362389055\\
10	0.0412247494097862\\
20	0.0119399513438065\\
30	0.0037350028629964\\
40	0.00122952461557811\\
};\label{POSM_DFT}
\addplot [color=blue, line width=1.0pt, mark size=1.5pt, mark=square*, mark options={solid, fill=blue, draw=blue}, forget plot]
  table[row sep=crcr]{%
-10	3.18609262625823\\
0	2.08866344645454\\
10	0.456297348643844\\
20	0.0868798248910842\\
30	0.0245597723378527\\
40	0.008754185657945\\
};\label{POST_SUM}
\addplot [color=red, line width=1.0pt, mark size=1.5pt, mark=square*, mark options={solid, fill=red, draw=red}, forget plot]
  table[row sep=crcr]{%
-10	0.486932525966813\\
0	0.0731994029192112\\
10	0.0166380200963032\\
20	0.00480806716565456\\
30	0.00157404663102637\\
40	0.000541172423780855\\
};\label{POSM_SUM}
\end{axis}
\node (Box1) [draw=none,fill=none] at (3, 4.35) {\shortstack[l]{ \ref{POST_DFT} \footnotesize Tensor: DFT \\
\ref{POST_SUM} \footnotesize Tensor: directional}};
\node (Box2) [draw=none,fill=none] at (1.75, 0.55) {\shortstack[l]{ \ref{POSM_DFT} \footnotesize Proposed: DFT \\
\ref{POSM_SUM} \footnotesize Proposed: directional}};
\end{tikzpicture}%

%% file: Fig/RunningTime.tex
\begin{tikzpicture}

\begin{axis}[%
width=0.8\textwidth,
height=0.3\textwidth,
at={(2,0.5)},
scale only axis,
xmin=2,
xmax=6,
xlabel style={font=\color{white!15!black}},
xlabel={Number of resolvable paths},
ymin=0.02,
ymax=2.5,
ymode=log,
ylabel style={font=\color{white!15!black}},
ylabel={CPU running time (s)},
axis background/.style={fill=white},
xmajorgrids,
ymajorgrids
]
\addplot [color=red, dotted, line width=1.0pt, mark size=1.5pt, mark=square*, mark options={solid, fill=red, draw=red}, forget plot]
  table[row sep=crcr]{%
2	0.0601329447350001\\
3	0.075777296626\\
4	0.103645846876\\
5	0.434756917563\\
6	2.332743532996\\
}; \label{RT_Tensor}
\addplot [color=blue, line width=1.0pt, mark size=1.5pt, mark=square*, mark options={solid, fill=blue, draw=blue}, forget plot]
  table[row sep=crcr]{%
2	0.049123353445\\
3	0.060994484258\\
4	0.094701099493\\
5	0.12006507139\\
6	0.129127882872\\
};\label{RT_Pro}
\end{axis}

%\begin{axis}[%
%width=0.2\textwidth,
%height=0.3\textwidth,
%at={(0.16\textwidth,0.245\textwidth)},
%scale only axis,
%xmin=2,
%xmax=4,
%ymin=0.04,
%ymax=0.103645846876,
%axis background/.style={fill=white},
%xmajorgrids,
%ymajorgrids
%]
%\addplot [color=red, dotted, line width=1.0pt, mark size=1.5pt, mark=square*, mark options={solid, fill=red, draw=red}, forget plot]
%  table[row sep=crcr]{%
%2	0.0601329447350001\\
%3	0.075777296626\\
%4	0.103645846876\\
%};
%\addplot [color=blue, line width=1.0pt, mark size=1.5pt, mark=square*, mark options={solid, fill=blue, draw=blue}, forget plot]
%  table[row sep=crcr]{%
%2	0.049123353445\\
%3	0.060994484258\\
%4	0.094701099493\\
%};
%\end{axis}
%\draw [help lines] (0,0) grid (8, 5);
\node (Box1) [draw=none,fill=none] at (1.2, 2.5) {\shortstack[l]{ \ref{RT_Tensor} \footnotesize Tensor \\
\ref{RT_Pro} \footnotesize Proposed }}; 
\end{tikzpicture}%